\newif\ifcombinedsubmission
\combinedsubmissiontrue

\documentclass[%
  aps,
  prl,
  reprint,
  superscriptaddress,
  longbibliography,
  floatfix
]{revtex4-2}

\ifcombinedsubmission\else
  \usepackage{xr}
\fi
\usepackage[T1]{fontenc}
\usepackage[utf8]{inputenc}
\usepackage{microtype}
\usepackage{amsmath,amssymb,amsthm}
\usepackage{mathtools}
\usepackage{braket}
\usepackage{graphicx}
\usepackage{xcolor}
\usepackage{hyperref}
\usepackage{dsfont}

\hypersetup{
  colorlinks=true,
  linkcolor=blue!40!black,
  citecolor=blue!40!black,
  urlcolor=blue!40!black,
  plainpages=false,
  pdfpagelabels=true
}

\newtheorem{theorem}{Theorem}
\newtheorem{lemma}{Lemma}
\newtheorem{proposition}{Proposition}
\newtheorem{corollary}{Corollary}
\theoremstyle{definition}
\newtheorem{definition}{Definition}
\newtheorem{assumption}{Assumption}
\theoremstyle{remark}
\newtheorem{remark}{Remark}

\theoremstyle{plain}

\DeclareMathOperator{\Tr}{Tr}
\DeclareMathOperator{\spr}{spr}
\DeclareMathOperator{\rank}{rank}

\newcommand{\Id}{\mathds{1}}
\newcommand{\id}{\Id}
\newcommand{\cH}{\mathcal{H}}
\newcommand{\cB}{\mathcal{B}}
\newcommand{\cR}{\mathcal{R}}
\newcommand{\eps}{\varepsilon}
\newcommand{\dTV}{d_{\mathrm{TV}}}
\newcommand{\RFQ}{R_F^{(Q)}}
\newcommand{\hRFQ}{\widehat{R}_F^{(Q)}}
\newcommand{\OmCM}{\Omega_{CM}^{(\cR)}}
\newcommand{\rhoM}{\rho_{M}^{(\cR)}}

\newcommand{\past}[1]{\overleftarrow{#1}}
\newcommand{\fut}[1]{\overrightarrow{#1}}
\newcommand{\pastfut}[1]{\overleftrightarrow{#1}}
\newcommand{\tpasto}{t_{\overleftarrow{0}}}
\newcommand{\Tpasto}{T_{\overleftarrow{0}}}
\newcommand{\tfuto}{t_{\overrightarrow{0}}}
\newcommand{\Tfuto}{T_{\overrightarrow{0}}}
\newcommand{\tpast}{\overleftarrow{t}}
\newcommand{\Tpast}{\overleftarrow{T}}
\newcommand{\tfut}{\overrightarrow{t}}
\newcommand{\Tfut}{\overrightarrow{T}}
\newcommand{\npast}{\overleftarrow{n}}
\newcommand{\Npast}{\overleftarrow{N}}
\newcommand{\nfut}{\overrightarrow{n}}
\newcommand{\Nfut}{\overrightarrow{N}}

\hypersetup{hidelinks}

\newcommand{\ClockEightTV}{0.047191}

\newcommand{\WalkClassThirtyOne}{0.010379}
\newcommand{\WalkClassThirtyTwo}{0.009890}

\ifcombinedsubmission\else
\fi
\newcommand{\suppref}[1]{%
  \ifcombinedsubmission\ref{#1}\else\ref*{supp-#1}\fi}

\begin{document}
\raggedbottom
\setcounter{secnumdepth}{2}

\title{Dimension Reduction for Quantum Adaptive Agents}

\author{Rishi Sundar}
\affiliation{Department of Physics \& Astronomy, University of Manchester,
             Manchester M13 9PL, United Kingdom}
\affiliation{Centre for Quantum Science and Engineering, University of Manchester,
             Manchester M13 9PL, United Kingdom}
\author{Thomas J. Elliott}
\affiliation{Department of Physics \& Astronomy, University of Manchester,
             Manchester M13 9PL, United Kingdom}
\affiliation{Department of Mathematics, University of Manchester,
             Manchester M13 9PL, United Kingdom}
\affiliation{Centre for Quantum Science and Engineering, University of Manchester,
             Manchester M13 9PL, United Kingdom}

\begin{abstract}
Adaptive agents realise complex reactive behaviours by using a memory of past input stimuli and output actions to guide structured future responses. Quantum adaptive agents can operate while storing less information in memory than optimal classical counterparts; yet, this does not necessarily translate into a reduced dimension of the memory that must be physically realised. We introduce a route-truncate-repair procedure that can convert entropic quantum memory advantages into reductions in memory dimension. Routing a reference input process through an agent yields a temporal matrix product state representation whose canonical bond is identified with the joint reference--agent memory. Truncating the agent's share of this bond and locally repairing the resulting dynamics produces a smaller, physically-valid agent that remains capable of responding to arbitrary input sequences. A fidelity-divergence certificate quantifies the resulting trade-off between accuracy and memory dimension. Benchmark adaptive processes exhibit substantial dimension reduction whilst preserving the underlying behaviour with high fidelity. The construction extends to feedback and coherent quantum interactions, with accuracy guarantees under arbitrary adaptive interrogation for a class of agents. For a resettable clock, four quantum memory dimensions attain an input--output fidelity divergence rate first matched with eleven states in our search over classical reduced models. These results establish a route from entropic memory advantages to practical, dimension-reduced adaptive quantum agents.
\end{abstract}
\maketitle

\section[Introduction]{Introduction}
\label{sec:introduction}

Physical systems rarely act in isolation; rather, they are driven by interactions with and stimuli from their surroundings. A
sensor is probed, a controller is programmed, an organism encounters its environment,
or a learning agent receives observations; in each case, what happens next can
depend on the whole history of what was supplied and what was returned. Adaptive
agents are systems that realise such input--output processes by conditioning each of their output
actions on the stream of input stimuli they receive. 

A (discrete-time) stochastic process, as for example, arising from a hidden Markov model (HMM)~\cite{BoxJenkins2015,Kalman1960,Rabiner1989}, assigns a probability law to its possible output strings, $p(y_{0:L})$.  An input-output process goes beyond this, assigning a conditional family of distributions
$p(y_{0:L}\mid x_{0:L})$, one for each possible series of inputs.  An adaptive
agent is a sequential physical realisation of this behaviour: At time~$t$ it
receives input stimulus $x_t$, produces output action $y_t$, and updates an internal memory that mediates
between past and future [Fig.~\ref{fig:agent_schematic}(a)], constituting an input-output HMM, or \emph{transducer}~\cite{Ljung1999,Bengio1996IOHMM,Littman2001PSR,Barnett_2015,Elliott2021QuantumAgents,rosas2025ai,kechrimparis2025}. The capacity of an agent's memory constrains the possible input-output processes it can realise, and thus the complexity of the behaviour it can exhibit in response to environmental stimuli.

Many internal models can generate the same observable process.  The computational mechanics programme~\cite{shalizi2001,crutchfield2012between} of complexity science merges histories that lead to identical conditional futures into equivalence classes, thus identifying \emph{causal states} for passive stochastic processes~\cite{crutchfield1989} and input-output processes~\cite{Barnett_2015} alike as the minimal predictive states of the process' behaviour. In turn, these causal states can be used to prescribe the minimal memory (classical) predictive models. Quantum models of stochastic processes have been shown to be capable of reproducing the same observable behaviour with smaller memory usage, by encoding causal states as overlapping quantum states~\cite{Gu_2012, mahoney2016occam, aghamohammadi2017extreme, Binder_2018, Loomis_2019, Liu_2019, ghafari2019dimensional, Elliott_2020,wu2023implementing}. Recently, analogous quantum memory advantages were proven possible for quantum adaptive agents realising input-output processes~\cite{Thompson_2017, Elliott2021QuantumAgents, thompson2025energeticadvantagesquantumagents}. Note that in all these works the processes considered are entirely classical, such that the inputs and outputs correspond to classical variables; quantum physics enters only through the information storage and dynamics of the internal memory system employed by a model or agent. We begin with this setting and later allow the stimuli themselves to be quantum, so that the agent may also be interrogated coherently.

\begin{figure}[t]
  \centering
  \includegraphics[width=.94\columnwidth]{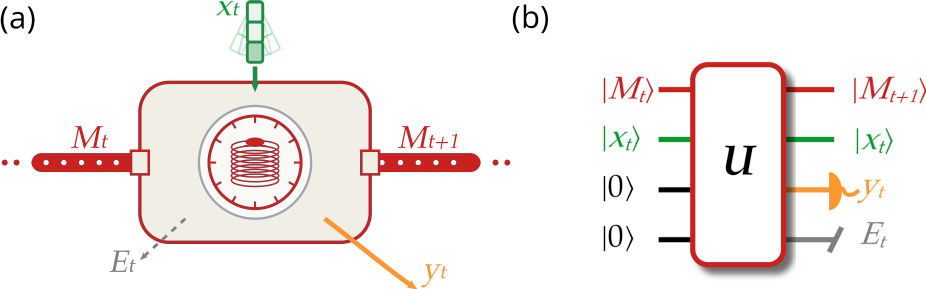}
  \caption{\textbf{Adaptive agent primitive.}
  (a) At each timestep~$t$ an adaptive agent receives an input stimulus~$x_t$, emits an output action~$y_t$, and
  updates its memory state from~$M_t$ to~$M_{t+1}$.  Inaccessible degrees of
  freedom are discarded into the environment,~$E_t$.  (b) Unitary circuit implementation of a single timestep for a quantum adaptive agent.}
  \label{fig:agent_schematic}
\end{figure}

Such quantum memory advantages have practical applications.  Feedback controllers,
adaptive sensors and metrological protocols, and learning agents all preserve
state between interventions~\cite{Lloyd2000,Hentschel2011,Briegel2012}.  As remarked above, with greater expressivity at fixed
memory size, quantum agents can thus support richer adaptive strategies, with online
execution~\cite{Elliott2021QuantumAgents,thompson2025energeticadvantagesquantumagents}. A key impediment to such applications is that the aforementioned quantum memory advantages are often framed in terms of information (i.e., entropic) efficiency; meanwhile, the associated quantum memory states may still span a high-dimensional
state space. That is, the information costs measure how much memory is occupied; dimension measures how large a Hilbert space the memory must be embedded within. This motivates the search for means to convert entropic quantum memory advantages into a practical dimension reduction, by trading off a small loss in accuracy for a significant reduction in memory dimension.

As we demonstrate herein, tensor networks~\cite{Verstraete_2006,Cirac_2021} provide the means to execute this trade-off in a controlled manner. 
Analogous approaches have been developed for the restricted setting of quantum models of passive stochastic processes based on matrix product state (MPS) representations~\cite{Yang2018MPSStoch} and their truncation~\cite{yang2024dimensionreductionquantumsampling,sundar2026quantumdimensionreductionhidden}, with their canonical bond spectrum coinciding with that of the model's steady-state memory, thus providing a natural compression target. However, as an input-output process branches out with different input stimuli, corresponding to counterfactual histories, it adopts a tree-like structure rather than a linear MPS [Fig.~\ref{fig:routed_history_network}(a)]. Such a tree can be a faithful representation of the agent with conditioning on a realised input string to select the correct dynamics. Yet, with it lacking a single transfer operator or canonical bond, it does not present a readily-accessible compression target. Moreover, a truncated tree network may not represent a physically-implementable agent, as the reduced maps are not guaranteed to constitute valid quantum instruments.

We resolve both these obstacles by adopting a route--truncate--repair construction.
\emph{Route}: Contract a stationary, passive reference input process as input to the agent, via a control rail.  The weighted branching family becomes a canonical, linear
temporal network whose bond is the steady-state of the memory under such driving [Fig.~\ref{fig:routed_history_network}(b)].  \emph{Truncate}: Retain the dominant part of the agent-memory information.  \emph{Repair}: Heal each projected
stimulus-conditioned update locally, producing a smaller, physically-valid instrument. The truncated agent is operable with arbitrary input sequences, beyond the reference process; that is, the reference determines which
histories the compression prioritises, but crucially, does not limit which stimuli the reduced agent may
accept. We take a history-dependent adaptive strategy as given and seek a smaller physical memory with which to realise its responses~\cite{Elliott2021QuantumAgents,thompson2025energeticadvantagesquantumagents}. Collectively, these components constitute our main result: That quantum adaptive agents can be mapped to an MPS representation when routed via an input driving process, which may then be truncated to devise a compressed quantum agent with dimension-reduced memory with certified accuracy guarantees.
We demonstrate application of our approach with two exemplar input-output processes, showing that significant dimension reduction is achievable with little distortion to the behaviour of the agent.
We extend the construction to feedback and coherent quantum interactions.

\begin{figure}[t]
  \centering
  \includegraphics[width=.94\columnwidth]{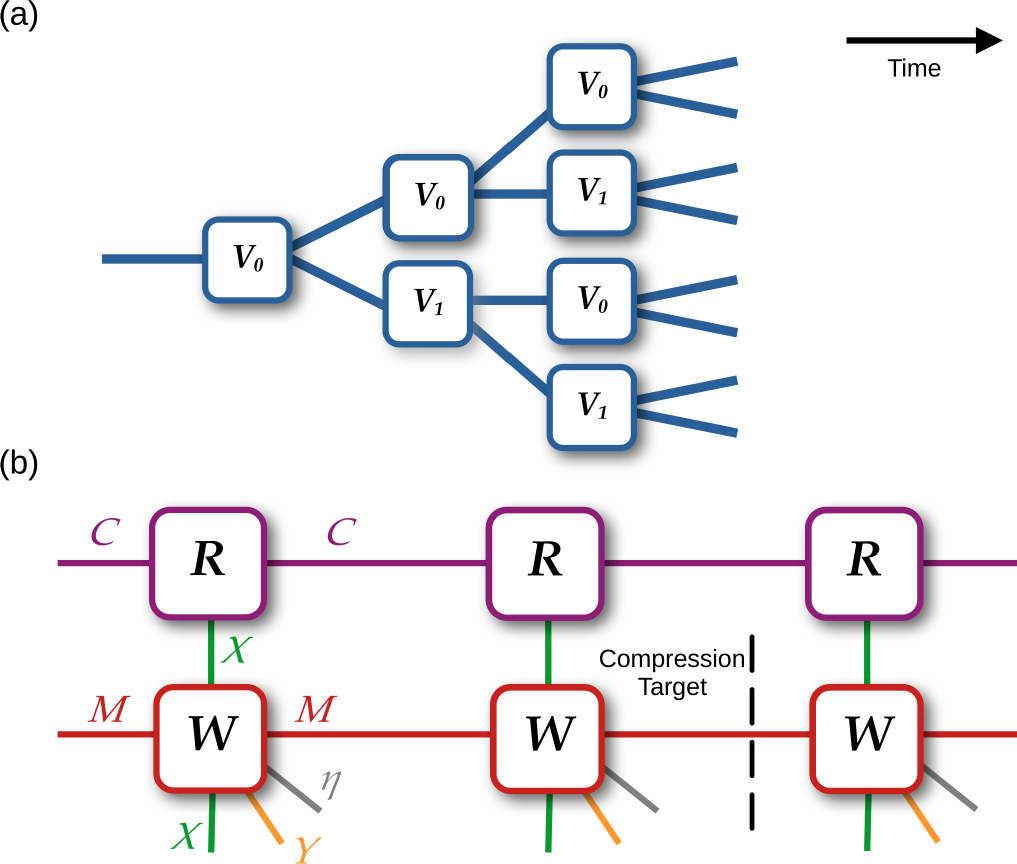}
  \caption{\textbf{Branching tree to routed MPS representation.}
  (a)~Treating the stimulus as a branch label produces $|\mathcal{X}|^L$ possible input
  histories after $L$ steps.  Each realised input string selects the correct
  conditional dynamics by following the path along the tree, but the unconditioned tree has no single transfer
  operator as a compression target.  (b)~Routing the stimulus as an additional wire supporting a stationary
  reference input process~$\cR$ contracted into the matrix product operator representation of the input-output process results in a uniform MPS representation of the joint process. The canonical bond carries the joint reference--agent memory; truncation acts only on the agent subsystem.}
  \label{fig:routed_history_network}

\end{figure}

\section[Adaptive agents]{Framework}
\label{sec:adaptive-agents}
\subsection{Input-output Processes}
Formally, a discrete-time input-output process is a family of conditional distributions that map input variables $X$ to output variables $Y$: $p(\fut{Y}_t|\past{X}_t,\past{Y}_t,\fut{X}_t)$, where $\past{X}_t$ corresponds to the entire string of past inputs at time $t$ (likewise $\past{Y}_t$ for outputs), and $\fut{X}_t$ and $\fut{Y}_t$ the corresponding future strings~\cite{Barnett_2015,rosas2025ai}. We use upper case to denote random variables, and lower case a corresponding realisation; collectively, $(\past{x}_t,\past{y}_t)$ is referred to as a history. The one-step description of the process, the distributions $p(Y_t|\past{X}_t, \past{Y}_t,X_t)$, is referred to as a strategy. We assume stationarity (i.e., time invariance), such that we may drop the time indexing.

A strategy can be implemented by an adaptive agent. A discrete-time adaptive agent is specified by the tuple $(\mathcal{X},\mathcal{Y},\{\sigma_m\},f,\Lambda)$. Here, $\mathcal{X}$ corresponds to the alphabet of input stimuli the agent can recognise, and $\mathcal{Y}$ the alphabet of output actions it can perform. To implement the strategy, it is endowed with a set of memory states $\{\sigma_m\}$ stored in a memory system $\mathsf{M}$, a memory encoding function $f:(\past{\mathcal{X}},\past{\mathcal{Y}})\to\{\sigma_m\}$ that maps histories to memory states, and a policy $\Lambda:\mathcal{X}\times\{\sigma_m\}\to\mathcal{Y}\times\{\sigma_m\}$ that, based on the input, conditionally produces an output and updates the memory state. The agent is said to be a faithful implementation of an input-output process if

\begin{equation}
 p_\Lambda(\fut Y\mid f(\past x,\past y),\fut x)
 =p(\fut Y\mid\past x,\past y,\fut x),
\label{eq:faithful-strategy}
\end{equation}
for all admissible histories and future inputs.

For a classical adaptive agent~\cite{Barnett_2015}, the memory states consist of a set of orthogonal states $\mathcal{S}$, and the policy a set of input-conditioned substochastic transition matrices $\{T^{y|x}\}$. For a quantum adaptive agent~\cite{Elliott2021QuantumAgents} the memory states are quantum states on a Hilbert space, and the policy a set of quantum instruments~\cite{Kretschmann2005} with Kraus operators
$\{K^{(x)}_{y,\eta}\}$ satisfying $\sum_{y,\eta}K^{(x)\dagger}_{y,\eta}K^{(x)}_{y,\eta}=\Id_M,$ where~$\eta$ labels inaccessible environment outcomes. Equivalently, one step is
the isometry
$V_x=\sum_{y,\eta}K^{(x)}_{y,\eta}\otimes\ket{y}_Y\otimes\ket{\eta}_E$.
The routed isometry $W=\sum_{x\in \mathcal{X}}\ket{x}\!\bra{x}_X\otimes V_x$ treats
the stimulus as a wire: it reads the input~$x$, applies the corresponding update, and
forwards~$x$ unchanged. 

\subsection{Classical agent construction}

A classical adaptive agent may be represented as an input--output transducer~\cite{rosas2025ai} with stimulus
alphabet \(\mathcal{X}\), action alphabet \(\mathcal{Y}\), internal state set \(\mathcal{S}\), and transition probabilities
\begin{equation}
T^{y|x}_{k j}
:=
p(y,\; k \mid x, j),
\qquad
\sum_{y,k}T^{y|x}_{k j}=1\forall j\in\mathcal{S}, x\in\mathcal{X}.
\label{em:eq:classical_transducer}
\end{equation}
Here \(j\) labels the incoming memory state and \(k\) the outgoing memory state.

A faithful classical (orthogonal-memory) embedding of a transducer employs one Kraus operator for each allowed transition,
\begin{equation}
K^{y,k,j|x}
=
\sqrt{T^{y|x}_{k j}}\;\ket{k}\!\bra{j}.
\label{em:eq:orthogonal_embedding}
\end{equation}
The output action is \(y\); hidden state labels \(j,k\) are discarded into the inaccessible environment. This is a
special case of the general quantum-instrument notation with
\[
\eta=(k,j),
\qquad
K^{(x)}_{y,\eta}=K^{y,k,j|x}.
\]
A transducer is said to be unifilar when the transitions between the memory states are deterministic given the output, i.e., \(k\) is fixed by \(j,x,y\). This parallels the notion of belief updating in reinforcement learning~\cite{cassandra1994acting}. Across all unifilar classical agents/transducers for stationary input-output processes, the $\varepsilon$-transducer of computational mechanics~\cite{Barnett_2015} is the provably memory-minimal classical model. It assigns each of the possible input-output pairs $(\past{x},\past{y})$ to a \emph{causal state}, according to an equivalence relation based on future output probabilities:
\begin{equation}
(\past{x}\!,\past{y})\!\sim\!(\past{x}'\!,\past{y}')\!\Leftrightarrow\! p(\fut{Y}|\past{x}\!,\past{y}\!,\fut{x})\!=\!p(\fut{Y}|\past{x}'\!,\past{y}'\!,\fut{x})\forall\fut{x}\!.
\end{equation}
The causal states thus represent equivalence classes. The causal states can be mapped, in one-to-one correspondence, to memory states of an agent (these memory states are also referred to as causal states). Endowing these memory states with conditional edge-emitting transitions $T_{s's}^{y|x}=p(y,s'|x,s)$ defined according to the conditional distributions associated with the constituent paths then prescribes the $\varepsilon$-transducer. See Refs.~\cite{Barnett_2015, Elliott2021QuantumAgents} for further details.

In this predictive setting, the causal states give the smallest exact classical
agent. Allowing errors lets an agent combine histories with similar, rather than
identical, futures, reducing the number of states it must retain. The classical
comparators below may also use stochastic memory updates; we compare their state
counts with the dimensions of the reduced quantum memories.

\subsection{Quantum agent construction}
\label{sec:quantum-agent-construction}

It has been shown that a quantum agent realising a stationary input-output process in the causal setting can be realised as an input-conditioned isometry acting on the memory states $\{\ket{\sigma_s}\}$, with the memory states pure and in one-to-one correspondence with the causal states of the process. This can equivalently be stated as a unitary [Fig.~\ref{fig:agent_schematic}(b)], defined implicitly through its action:
\begin{equation}
U\ket{\sigma_s}\ket{x}\ket{0}\ket{0}=\sum_{ys'}\sqrt{T_{s's}^{y|x}}\ket{\sigma_{s'}}\ket{x}\ket{y}\ket{\eta(x,y,s)}.
\end{equation}
Here, the subspaces, in order, correspond to the memory register, input register, output register, and environment register. Observe that the unitary leaves the input register unchanged; it treats the input as a conditioning. The specific forms of the final environment states $\{\ket{\eta(x,y,s)}\}$ are the only tunable parameters. From this equation, one can deduce the overlaps of the quantum memory states, assign them to a basis following a reverse Gram-Schmidt procedure, and then use this to define the specific form of $U$. See Ref.~\cite{Elliott2021QuantumAgents} for further details.

A specific assignment of the environment states, valid for starting from any $\varepsilon$-transducer, is as follows~\cite{Elliott2021QuantumAgents}: Consider an effective decomposition of the quantum memory states into a product of memory states specialised for each input, such that $\ket{\sigma_s}=\bigotimes_{x}\ket{\sigma_s^x}$; then, assign environment states such that $\ket{\eta(x,y,s)}=\bigotimes_{x'\neq x}\ket{\sigma_s^{x'}}$. This effective representation of the memory states can then be mapped to a Hilbert space of at most $|\mathcal{S}|$ dimensions following the reverse Gram-Schmidt procedure. We use this construction for the explicit examples in this work.

Compression acts directly on the operators
$K^{(x)}_{y,\eta}$; it can therefore also be applied to other quantum
realisations of the same strategy.

Our aim is therefore to construct physically valid agents with smaller memories and to quantify the distortion of the input--output behaviour they realise.

\section[Route--truncate--repair]{Route--truncate--repair}
\label{sec:construction}

\subsection[Routing]{Routing}
\label{sec:routing}

Stimuli are supplied by an input process $\cR$~\footnote{More generally, an input strategy, wherein the input distribution is itself conditioned on the previous outputs of the agent~\cite{Thompson_2017}.}. We here consider input processes as generated by (finite) classical HMMs with alphabet $\mathcal{X}$, orthogonal memory states $c\in\mathcal{C}$ on a memory system $\mathsf{C}$, and normalised transition probabilities $p_R(x,c'\mid c)$. Routing the process timestep by timestep through~$W$ gives joint Kraus operators on the combined memory system
$\mathsf{C}\otimes \mathsf{M}$, along with the input, output, and inaccessible environment registers:
\begin{equation}
L_{x,y,\eta,c,c'}
=\sqrt{p_R(x,c'\mid c)}\,\ket{c'}\!\bra{c}\otimes K^{(x)}_{y,\eta}.
\label{eq:routed_kraus}
\end{equation}
We introduce the shorthand $\omega=(x,y,\eta,c,c')$. Collectively, these Kraus operators define a quantum channel,
$\Phi_{\cR}(\cdot)=\sum_\omega L_\omega\cdot L_\omega^\dagger$.

The history MPS separates what is carried forward from what is left in the
interaction record. Its virtual bond carries the joint memory
$\mathsf C\otimes\mathsf M$ from one step to the next, whilst the physical
index $\omega$ records the emitted labels, including inaccessible environment
outcomes. The stationary bond state therefore tells us how this memory is occupied during
the interaction. Quantum models of passive stochastic processes are recovered by taking a single input symbol and a one-dimensional reference memory. The routed operators then reduce to $L_{y,\eta}=K_{y,\eta}$, giving the corresponding temporal MPS representation.

\begin{theorem}[Canonical routed MPS]
\label{thm:routed_history}
Let a finite-dimensional adaptive agent be driven by a stationary,
finite-memory, input process~$\cR$.  The routed
operators in Eq.~\eqref{eq:routed_kraus} generate a uniform MPS with bond space $\mathsf{C}\otimes \mathsf{M}$, and
\begin{equation}
\sum_\omega L_\omega^\dagger L_\omega=\Id_{CM}.
\label{eq:left_canonical}
\end{equation}
Hence the MPS is in left-canonical form.  If $\Phi_{\cR}$ is primitive, its canonical
bond state in the stationary history is the unique joint fixed point
$\OmCM$, with $\Phi_{\cR}(\OmCM)=\OmCM$ and $\Tr\OmCM=1$.  The eigenvalues of
$\OmCM$ are the squared Schmidt coefficients across a virtual bond cut.  For a
memoryless input process, $|\mathcal{C}|=1$ and the bond state reduces to the driven
stationary agent memory $\rho_\star^{(\cR)}$.
\end{theorem}

\noindent\textit{Proof sketch.} Completeness of each instrument and
normalisation of~$p_R$ give Eq.~\eqref{eq:left_canonical}, making
$U_{\cR}=\sum_\omega L_\omega\otimes\ket{\omega}$ an isometry whose
repetition produces the uniform pure history. Tracing a physical site of the
doubled tensor yields $\Phi_{\cR}$ because the $\ket{\omega}$ are orthogonal.
In left-canonical form the right fixed point is the bond density operator;
primitivity and purity identify it with~$\OmCM$ and its spectrum with squared
Schmidt coefficients. See Supplemental Material~\cite{SupplementalMaterial},
Secs.~\suppref{sm:sec:reference}--\suppref{sm:sec:canonical}, for the full proof.

This theorem identifies the compression target: The bond of the routed MPS. That is, for an input-output process subject to a memoryless driving input process,
the carried memory is the temporal virtual bond across $\mathsf{M}$. The choice of driving input process will affect the steady-state adopted by the memory, and thus will in general impact upon the structure of the compressed model. In particular, the statistical fidelity will be preferentially preserved with respect to this input driving, and so in practical deployment the reference driving input process used for compression should represent an expected typical input process. As such, the reference process determines what behaviour is prioritised in compression, analogous to a training set. Nevertheless, the reconstructed, compressed agents we develop here are capable of responding to any input process. The fixed-agent deployment test in Sec.~\ref{sec:clock-performance} shows how the reduced agent performs when the operating input distribution is changed after compression.

\subsection[Agent-local truncation]{Agent-local truncation}
\label{sec:agent-truncation}

With respect to a temporally-correlated (i.e., memoryful) reference input process the canonical bond of the routed MPS sits across
$\mathsf{C}\otimes \mathsf{M}$. However, only the latter subsystem is part of the agent itself, and so a physically-meaningful compression of the agent cannot modify the external memory of the generator of the input process $\mathcal{R}$. That is, directly compressing the complete bond of the routed MPS by truncating its spectrum is not tenable, since it would also involve compression of the input reference process. Hence truncations must be of the form $\Id_C\otimes P_M$, with the relevant spectrum being that of the driven agent's memory state marginal
$\rhoM=\Tr_C\OmCM=\sum_i\lambda_i\ket{i}\!\bra{i}$, with eigenvalues assigned in decreasing order.

\begin{corollary}[Optimal agent-local truncation]
\label{cor:discarded_weight_bound}
Among rank-$\tilde{d}_q$ projectors on the agent memory, the projector retaining maximal
stationary canonical weight is $P_M=\sum_{i\le \tilde{d}_q}\ket{i}\!\bra{i}$, up to
degeneracies, with discarded weight
\begin{equation}
\eps_M(\tilde{d}_q)=\Tr[(\Id_M-P_M)\rhoM]=\sum_{i>\tilde{d}_q}\lambda_i.
\label{eq:discarded_weight}
\end{equation}
\end{corollary}

Indeed, $\Tr[(\Id_C\otimes P_M)\OmCM]=\Tr(P_M\rhoM)$. Since $P_M$ has rank
$\tilde{d}_q$, this trace is maximised by placing its support on the $\tilde{d}_q$ largest
eigenvalues of~$\rhoM$ (Ky Fan's maximum principle~\cite{fan1951maximum}), giving
Eq.~\eqref{eq:discarded_weight}. This is the temporal analogue of density-matrix
truncation in DMRG~\cite{Schollwoeck2011,Verstraete_2006}: The reduced dimension allowance is concentrated
in the span most occupied by the steady-state of the driven memory. Note that this optimality is with respect to the
retained stationary weight in the memory; behavioural accuracy in reproducing the target strategy is certified separately.

\subsection[Repair to a valid compressed agent]{Repair to a valid compressed agent}
\label{sec:physical-repair}

Such projection alone does not result in a physically-valid quantum agent: Information can leak from the retained subspace to the discarded subspace, such that the projected Kraus operators will in general not be trace preserving. To reconstitute a valid compressed quantum agent, we must repair the Kraus operators such that they represent a set of valid quantum instruments. To do so, for each stimulus we define and construct:
\begin{equation}
\begin{aligned}
\bar K^{(x)}_{y,\eta}&:=P_MK^{(x)}_{y,\eta}P_M,\\
G_x&:=\sum_{y,\eta}\bar K^{(x)\dagger}_{y,\eta}
    \bar K^{(x)}_{y,\eta},\\
\widetilde K^{(x)}_{y,\eta}&:=\bar K^{(x)}_{y,\eta}G_x^{-1/2}.
\end{aligned}
\label{eq:polar_repair}
\end{equation}
The inverse is taken on the support of~$G_x$~\footnote{When $G_x > 0$, this polar factor represents the (unique) nearest isometry to the unrepaired projected update in Frobenius norm, making the choice of decomposition principled.}; Supplemental Material, Sec.~\suppref{sm:sec:repair}, provides the resolution for the case of rank-deficiency.

\begin{theorem}[Repair and asymptotic accuracy certificate]
\label{thm:repair_certificate}
\textit{(i) Validity.} For every stimulus~$x$, Eq.~\eqref{eq:polar_repair},
with the kernel completion of Supplemental Material, Sec.~\suppref{sm:sec:repair}, when needed, satisfies
$\sum_{y,\eta}\widetilde K^{(x)\dagger}_{y,\eta}
\widetilde K^{(x)}_{y,\eta}=P_M$ and therefore defines a physically-valid dimension
$\tilde{d}_q=\rank(P_M)$ quantum adaptive agent capable of responding to arbitrary stimulus trajectories.

\textit{(ii) Certification.}  Fix a common Kraus-label alphabet for the
original and reduced agents.  Assume the stationary bulks are normal and the boundary contractions reach a leading mixed-transfer sector without cancellation, as specified in Supplemental Material, Sec.~\suppref{sm:sec:certificates}; then their routed purified MPS representations $|\Psi_{\cR,L}\rangle$ and $|\widetilde\Psi_{\cR,L}\rangle$ obey
\begin{equation}
\begin{aligned}
\RFQ&:=-\lim_{L\to\infty}\frac{1}{2L}\log_2
  |\langle\widetilde\Psi_{\cR,L}|\Psi_{\cR,L}\rangle|
  =-\tfrac12\log_2\mu,\\
\mu&=\spr\!\left(Z\mapsto
\sum_{\omega}\widetilde L_\omega Z L_\omega^{\dagger}\right).
\end{aligned}
\label{eq:exact_certificate}
\end{equation}
Here, $\RFQ$ is the quantum fidelity divergence rate of the MPSs~\cite{Vanhecke2021Truncation,yang2024dimensionreductionquantumsampling}, and $\mathrm{spr}(\cdot)$ is the spectral radius. The corresponding statistical fidelity divergence rate~\cite{Yang_2020} of the realised classical
input--output processes, under driving by reference input process~$\cR$, is no greater than~$\RFQ$.
\end{theorem}

\noindent\textit{Proof sketch.}
Part~(i) is the polar identity
$G_x^{-1/2}G_xG_x^{-1/2}=P_M$ on the retained support, followed by the kernel
completion when necessary.  For part~(ii), every finite-history overlap is a
boundary contraction of powers of the rectangular mixed transfer operator in
Eq.~\eqref{eq:exact_certificate}; under the stated conditions its dominant
eigenvalue fixes the asymptotic rate, while boundary terms are subextensive.
Tracing out inaccessible environment registers and reading out the classical symbols cannot
decrease fidelity, which gives the data-processing bound for the statistical fidelity divergence rate. Details are in
Supplemental Material, Secs.~\suppref{sm:sec:certificates}--\suppref{sm:sec:repair}.

\section[Accuracy and interacting environments]{Accuracy and interacting environments}
\label{sec:accuracy-environments}

The fidelity-divergence certificate quantifies the error introduced by compression. We extend this comparison to finite experiments and relate accuracy to the spectrum of the driven memory. We also allow the driving environment to respond to the agent through classical feedback or coherent quantum interactions.

Under a chosen input process, the reduced agent's memory dimension and history distortion define an achievable point on a dimension--distortion curve, closely related to predictive rate--distortion~\cite{Still2010OptimalCausalInference,MarzenCrutchfield2016PRD}. The finite-history bounds constrain changes in the expected value of every bounded score computed from the interaction record, including task rewards.

\subsection[Finite-history accuracy]{Finite-history accuracy}
\label{sec:finite-history-accuracy}

For an experiment lasting $L$ steps, accuracy depends on the initial memories
as well as their subsequent dynamics. Initialising the mixed-transfer recurrence with these preparations gives a bound on the fidelity of the interaction histories. Choose normalised
purifications with coefficient matrices $A_0$ and $B_0$ for the original and
reduced initial memories, with a common purifying space, and set $Z_0=B_0A_0^\dagger$.
Iterate
\begin{equation}
Z_{t+1}=\sum_\omega\widetilde L_\omega Z_tL_\omega^\dagger.
\label{eq:finite-transfer}
\end{equation}
Uhlmann's theorem and data processing give
\begin{equation}
F_L^{(\mathrm{stat})}\geq\|Z_L\|_1,
\qquad
\dTV(p_L,\widetilde p_L)\leq\sqrt{1-\|Z_L\|_1^2}.
\label{eq:finite-certificate}
\end{equation}
Here $F_L^{(\mathrm{stat})}=\sum_w\sqrt{p_L(w)\widetilde p_L(w)}$ is root fidelity
on the input--output words~$w$. A value near one means that the original and
reduced agents assign similar probabilities to these histories. The certificate
also compares inaccessible records, so it can give a lower fidelity than the
visible histories themselves. Its derivation is given in Supplemental Material,
Sec.~\suppref{sm:sec:finite-certificate-new}.

The change in the expected value of any score
$g\in[a,b]$ satisfies
$|\mathbb E_{p_L}g-\mathbb E_{\widetilde p_L}g|\leq(b-a)\dTV(p_L,\widetilde p_L)$.
For example, taking $g$ to be one when the clock ticks within a chosen interval
and zero otherwise bounds the change in that event's probability.

\subsection[Accuracy of physical completion]{Accuracy of physical completion}
\label{sec:physical-accuracy}

To relate discarded memory occupation to behavioural error, we consider the
reset completion, which replaces the discarded component after each update by
a chosen state in the retained space. Let $P=P_M$, $Q=\Id-P$, and let
$\tau_P$ be any normalised state supported on $P$. The corresponding local
encoder is
\begin{equation}
\mathcal C_P(X)=PXP+\Tr(QX)\tau_P
\label{eq:local-encoder}
\end{equation}
Applying this encoder after an update returns the memory to the retained
space. The output probabilities of that update are preserved, and the altered
memory determines the subsequent responses. This defines the reset completion
described in Supplemental Material, Sec.~\suppref{sm:subsec:reset_completion}.

\begin{theorem}[Finite-horizon accuracy of physical completion]
\label{thm:physical-accuracy}
Consider the same causal environment in the original and reduced experiments,
acting through the agent's interface but not directly on its internal memory.
Initialise the reduced experiment by applying $\mathcal C_P$ locally to the
original memory. Define
\begin{equation}
E_L(P)=\sum_{t=0}^{L-1}\Tr\!\left[Q\rho_{M,t}\right],
\label{eq:discarded-occupation}
\end{equation}
where $\rho_{M,t}$ is the unconditional memory state in the original experiment
before round~$t$. The half trace distance between external outputs obeys
\begin{equation}
D_{L,\mathrm{reset}}\leq\min\{1,\sqrt{E_L}+2E_L\}.
\label{eq:physical-reset-bound}
\end{equation}
\end{theorem}

\noindent\textit{Proof sketch.}
Resolve each encoder into a keep branch and orthogonally flagged recovery branches.
The no-recovery branch is a sequence of projectors in the original experiment's
interaction picture. The quantum union bounds~\cite{Gao2015Union} control its
probability and disturbance; including the recovery branches gives
Eq.~\eqref{eq:physical-reset-bound}. The complete argument, including initial
correlations, is in Supplemental Material, Sec.~\suppref{sm:sec:physical-new}.
The polar repair used in the benchmarks obeys a corresponding extension proved
in Supplemental Material, Sec.~\suppref{sm:sec:polar-comparison-new}.

The quantity $E_L$ counts how much of the original memory lies outside the
retained space over the experiment. The leading eigenspace of
$\rho_{M,\mathrm{av}}=L^{-1}\sum_{t=0}^{L-1}\rho_{M,t}$ minimises this
occupation at fixed rank, so spectral concentration gives a sufficient physical
dimension at the chosen horizon and tolerance. Counting discarded occupation at every step can overestimate the disturbance
when successive memory states change slowly. For the $N=256,d=4$ clock, the
bound is nontrivial for five steps out of the 256 in a lifetime.
Section~\ref{sec:benchmarks} evaluates the clock's accuracy directly through
the transfer calculations. More detailed dimension bounds and a conditional-leakage
refinement are given in Supplemental Material, Secs.~\suppref{sm:sec:dimension-new}
and~\suppref{sm:subsec:conditional-leakage}.


\subsection[State-independent retention]{State-independent retention}
\label{sec:uniform-retention}

A stronger guarantee is available when the probability of remaining in the
retained space is the same for every incoming memory state and every stimulus.
The condition is
\begin{equation}
\sum_{y,\eta}K_{y,\eta}^{(x)\dagger}P_MK_{y,\eta}^{(x)}
=w\Id_M,\qquad w=1-\eps_M,
\label{eq:full-retention-main}
\end{equation}
independently of $x$. Polar repair then amounts to the same rescaling for
every stimulus. Since no choice of input can increase this loss probability,
the resulting accuracy bound holds even when later inputs are chosen in
response to earlier outputs.
\begin{proposition}[Accuracy under state-independent retention]
\label{prop:structured-accuracy}
Under Eq.~\eqref{eq:full-retention-main}, a common causal tester interacting
with the specified coherent implementations of the original and reduced
agents produces outputs satisfying
\begin{equation}
F_L\geq w^{L/2}|g_0|,\qquad D_L\leq\sqrt{1-w^L|g_0|^2},
\label{eq:structured-finite}
\end{equation}
where $g_0$ is the embedded initial purification overlap. For $g_0\ne0$ the
corresponding rate is at most $-\tfrac14\log_2w$, when the rate exists.
\end{proposition}
Under the full-memory identity~\eqref{eq:full-retention-main}, each agent
update multiplies the purification overlap by $\sqrt w$, independently of
the tester's choice of the next stimulus.
The proof is given in Supplemental Material,
Sec.~\suppref{sm:sec:walk-retention-new}.

\paragraph*{Selecting a certified dimension.}
The same condition connects the retained spectrum to the completed agent's
rate. For a leading spectral projector of rank $d$, let
$\eps_d=\Tr[(\Id-P_d)\rho_M]$. If this projector obeys
Eq.~\eqref{eq:full-retention-main}, the stationary-history certificate satisfies
\begin{equation}
R_F^{(\mathrm{stat})}\leq\RFQ
\leq-\tfrac14\log_2(1-\eps_d)
\label{eq:structured-spectral-bound}
\end{equation}
with the history and boundary conventions of Theorem~\ref{thm:repair_certificate}.
Thus a target rate $\delta_\star$ is certified by discarding no more than
$1-2^{-4\delta_\star}$ of the stationary weight. If leading projectors satisfying
the condition exist at every rank, the smallest rank selected by this bound is
the smoothed rank of that state. For the cyclic walk below, the QFDR equals
the bound, so its certificate can be read directly from the spectrum.
The general corollary and the walk's equality proof are in Supplemental
Material, Sec.~\suppref{sm:subsec:spectral-certification}.

\subsection[Closed-loop driving]{Closed-loop driving}
\label{sec:closed-loop}

An agent's actions can change what it encounters next. To include this feedback,
we let the reference be an environment with its own memory: in state $c$ it
supplies input $x$ with probability $q(x|c)$, observes the action $y$, and
updates to $c'$ with probability $E(c'|c,x,y)$. Routing now gives the joint
operators
\begin{equation}
L^{\mathrm{fb}}_{c,x,y,c',\eta}
=\sqrt{q(x|c)E(c'|c,x,y)}\,
|c'\rangle\langle c|\otimes K^{(x)}_{y,\eta}.
\label{eq:feedback-routing}
\end{equation}
The feedback probabilities sum to one, so these operators satisfy the same
completeness relation as before. Under the conditions of Theorem~\ref{thm:routed_history}, the stationary joint
memory, including correlations established by the feedback, is the canonical
bond state.

We truncate only $\mathsf M$ and repair the update separately for each
stimulus, preserving the agent's input and output alphabets. The reduced agent
is then coupled to the environment through the original transition law
$E(c'|c,x,y)$. Its altered actions can change subsequent inputs, and the error
comparison includes the resulting joint history. The
construction and the preparation of
correlated initial states are detailed in Supplemental Material,
Sec.~\suppref{sm:sec:feedback-new}.

\subsection[Quantum testers]{Quantum testers}
\label{sec:quantum-testers}

A quantum controller or probe can supply stimuli in superposition and use a
memory of earlier interactions to choose what to do next. We describe such an
interaction by a tester with memory $\mathsf C$ and stimulus register $\mathsf X$.
It lends $\mathsf X$ to the agent, receives it back with the action register
$\mathsf Y$, and records the result in $\mathsf F$. For a stationary tester one
round is
\begin{equation}
\begin{aligned}
T &: \mathsf C\otimes\mathsf X\otimes\mathsf Y
       \longrightarrow \mathsf C\otimes\mathsf X\otimes\mathsf F,\\
U_{\mathcal T}&=(T\otimes\Id_{\mathsf{ME}})
                    (\Id_{\mathsf C}\otimes W).
\end{aligned}
\label{eq:quantum-tester}
\end{equation}
Repeating these rounds gives the same routed-history construction on the joint
tester--agent memory, with compression acting only on the agent.

Polar repair produces a controlled isometry, allowing the reduced agent to accept stimuli in superposition or entangled with the tester. For a specified tester, the mixed transfer bounds the fidelity of the
interaction histories. Under uniform retention, Proposition~\ref{prop:structured-accuracy}
bounds how well any adaptive quantum tester can distinguish the original and
reduced devices~\cite{Gutoski2018StrategyFidelity,Gutoski2012StrategyNorm,GutoskiWatrous2007}. A coherent tester can distinguish
different dilations of the same classical instruments, so this statement is
made for the specified coherent implementation of the agent. The full
construction, accessible-output conventions and strategy-fidelity statement
are given in Supplemental Material, Sec.~\suppref{sm:sec:quantum-testers}.

\section[Benchmarks]{Benchmarks}
\label{sec:benchmarks}

We apply our compression and repair pipeline to quantum adaptive agents tasked with implementing two exemplar input-output processes: A resettable renewal clock~\cite{Elliott2021QuantumAgents}, and an input-dependent cyclic random walk~\cite{Garner_2017}. Their exact quantum agents store less information than the corresponding classical agents, $C_q<C_\mu$. Uncompressed quantum agents are constructed using the procedure of Ref.~\cite{Elliott2021QuantumAgents} (Sec.~\ref{sec:quantum-agent-construction}); for both examples these bear no dimensional memory advantage over minimal classical models prior to compression. The clock tests whether a fixed small memory can retain accurate behaviour as the temporal resolution increases. The walk's translation symmetry allows the stronger guarantees of Sec.~\ref{sec:uniform-retention} to be evaluated exactly.

\begin{figure*}[t]
  \centering
  \includegraphics[width=.88\textwidth]{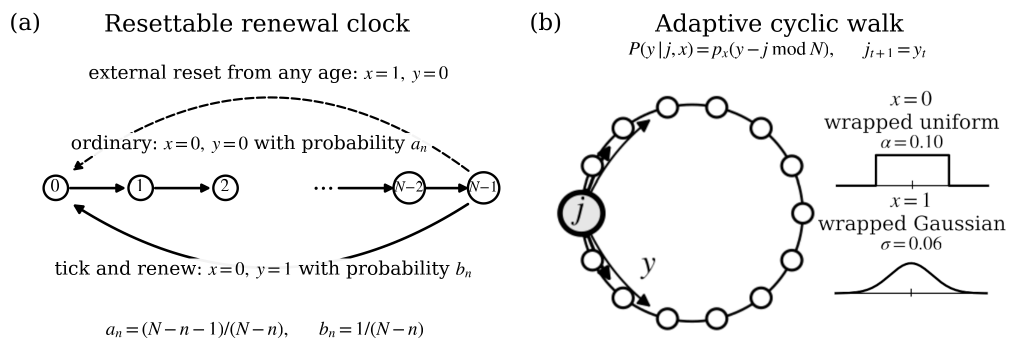}
  \caption{\textbf{Transition structures of benchmark processes.}
  (a)~Under ordinary evolution ($x=0$), a clock of age~$n$ either `survives' (emits $y=0$) and
  advances to $n+1$, or ticks ($y=1$) and renews to age zero.  The external stimulus~$x=1$ invokes a reset to age zero from any age, without a tick ($y=0$). 
  (b)~In the  cyclic walk process, a `walker' occupies a discrete position~$j$ on a ring.  The stimulus selects a  translation-invariant shift drawn from distribution $p_x(r)$ and announces its new position $y=j+r \bmod N$: The shift distribution is a uniform top hat for input  $x=0$, and a Gaussian for $x=1$.}
  \label{fig:benchmark_models}
\end{figure*}

\subsection{Benchmark processes}
\label{sec:benchmark-processes}

Here we briefly outline the details of the input-output processes used as benchmarks for our compression scheme. Their input-dependent HMM representations are depicted in Fig.~\ref{fig:benchmark_models}. See Supplemental Material, Sec.~\suppref{sm:sec:numerics}, for further details.

\paragraph*{Resettable renewal clock.}
Renewal processes with a binary output alphabet $\{0,1\}$ consist of a reset state, and a linear array of age states corresponding to the time (in the discrete-time case, number of timesteps) since the last 1 was emitted. Thus, in discrete-time the memory states can be labelled $n=0,\ldots,N-1$, corresponding to the age, with $n=0$ the reset state.  The survival probability $\Phi(n)$ is the probability that no tick has occurred after $n$ ordinary timesteps~\cite{marzen2015informational}. For a uniform waiting-time distribution~\cite{elliott2018superior}, such that the tick is equally-likely to occur at any age in the interval $[0,N-1]$, the survival probability is $\Phi(n)=(N-n)/N$. This leads to an HMM transition structure where $a_n:=P(0,n+1|n)=\Phi(n+1)/\Phi(n)$ and $b_n:=P(1,0|n)=1-a_n$. We endow this process with an input-dependence through a reset mechanism, wherein upon input $x=0$ the process undergoes evolution as above; upon input $x=1$ it is forced to immediately transition to the reset state $n=0$ while emitting output $y=0$~\cite{Elliott2021QuantumAgents}.

Following Ref.~\cite{Elliott2021QuantumAgents} using the construction above,
the corresponding quantum agent has memory states $\{\ket{\sigma_n}\}$ with
overlaps $c_{nm}:=\braket{\sigma_n|\sigma_m}=(N-\max(n,m))/\sqrt{(N-n)(N-m)}$.
Its driving reference is a memoryless process that supplies a reset symbol $x=1$ with probability $r_N=1-e^{-1/(2N)}$ per timestep, leading to steady-state distributions of the age states $\pi_n^{(\cR)}\propto(1-r_N)^n\Phi(n)$ and quantum agent memory $\rhoM=\sum_n\pi_n^{(\cR)}\ket{\sigma_n}\!\bra{\sigma_n}$.

Increasing $N$ refines the clock's resolution within the same physical lifetime
$\tau$: $\Delta t=\tau/N$, whilst the physical reset rate remains $1/(2\tau)$.
The exact encoding must distinguish more age states as the resolution improves;
we ask how much of this additional memory can be removed.

\paragraph*{Input-dependent cyclic walk.}
The system consists of a position $j$ on a unit circle, which is discretised into $N$ contiguous bins of uniform size. At each timestep the system undergoes a probabilistic shift given by a (bin-integrated) distribution that is site-invariant, such that if the system is in site $j$, and the shift is of size $k$ sites to the right, the system transitions to site $j+k \bmod N$, while emitting an output corresponding to the new site~\cite{Garner_2017}. We consider two shift distributions: A uniform, top-hat distribution of half-width~$0.10$ (active on input $x=0$); and a zero-mean Gaussian of width $\sigma=0.06$ (on input $x=1$).  

For each stimulus define $\ket{\phi_j^{(x)}}=\sum_y\sqrt{p_x(y-j)}\ket{y}$ and
$\ket{\sigma_j}=\ket{\phi_j^{(0)}}\otimes\ket{\phi_j^{(1)}}$ for the quantum memory states, following the effective representation above. Thus, the quantum memory state overlaps are $c_{jk}:=B_0(k-j)B_1(k-j)$, where $B_x(d)=\sum_r\sqrt{p_x(r)p_x(r-d)}$. The reference input distribution is memoryless and uniform across the two stimuli. The accuracy measures are defined below; numerical details are given in Supplemental Material, Sec.~\suppref{sm:sec:numerics}.

\subsection[Comparing reduced agents]{Comparing reduced agents}
\label{sec:benchmark-comparison}

We compare classical agents with $d$ internal states and quantum agents with
Hilbert space dimension $d=\tilde d_q$, using the probabilities of the
input--output histories they generate. Each model starts in its own steady
state under the same input strategy. The statistical fidelity-divergence
rate~\cite{Yang_2020} measures how quickly those histories become
distinguishable from the target:
\begin{equation}
R_F^{(\mathrm{stat})}
=-\lim_{L\to\infty}\frac{\log_2 F_L^{(\mathrm{stat})}}{2L}
\leq R_F^{(Q)},
\label{eq:benchmark-rate-comparison}
\end{equation}
under the assumptions of Theorem~\ref{thm:repair_certificate}. With base-two
logarithms, $D_{1/2}(p_L\Vert\widetilde p_L)=-2\log_2F_L^{(\mathrm{stat})}$,
so the input--output FDR is one quarter of the corresponding R\'enyi divergence
rate. A smaller value means that history fidelity decays more slowly. QFDR
also compares inaccessible records, whose removal cannot decrease fidelity;
it therefore bounds the rate seen in the input--output statistics. Filled and
hollow diamonds in Fig.~\ref{fig:compression_results} show the directly
calculated rate and its upper bound for the same repaired agent. Their
separation measures the conservativeness of the certificate.

For the classical comparison, we order the $N$ causal states by age for the
clock and by position for the walk, and merge them into $d$ contiguous blocks
$g(n)=\lfloor nd/N\rfloor$. The reduced agent stores the block label. Its
responses and memory updates are the original ones averaged with the stationary
weights within each block (Supplemental Material,
Sec.~\suppref{sm:subsec:classical-baseline}).
This \emph{stationary aggregation} is used for both benchmarks, as shown by
the orange points in Fig.~\ref{fig:compression_results}.

For the clock we also search for better classical agents. Every tick and every
external reset returns the clock to age zero, so the aggregated clock restarts
from the same state after each event. A classical clock with $d$ internal
states that restarts from a common distribution $\alpha$ after both kinds of
event is \emph{regenerative}: a row-substochastic matrix $C$ describes
evolution without a tick---its missing row weight is the probability of
ticking---and
\begin{equation}
\widetilde a_\ell=\alpha C^{\ell-1}(\boldsymbol1-C\boldsymbol1),
\qquad \widetilde s_k=\alpha C^k\boldsymbol1
\label{eq:classical-clock}
\end{equation}
give its waiting-time law and survival probability under ordinary evolution.
Starting from aggregated clocks and other initial models, we optimise $C$
and $\alpha$ for input--output FDR at each dimension. The regenerative family
includes nonunifilar models, whose hidden updates remain stochastic after the
output is observed (Supplemental Material, Sec.~\suppref{sm:sec:numerics}).
We retain the most accurate candidate found at each dimension. A comparable search for the walk has far more
parameters, since each state and stimulus requires a distribution over all $N$
positions; in a small exploratory study, a behaviourally fitted classical walk
outperformed the direct quantum reduction (Supplemental Material,
Sec.~\suppref{sm:subsec:classical-scope}). The quantum agents are not fitted in
either benchmark.

\begin{figure*}[tp]
\centering
\includegraphics[width=\textwidth]{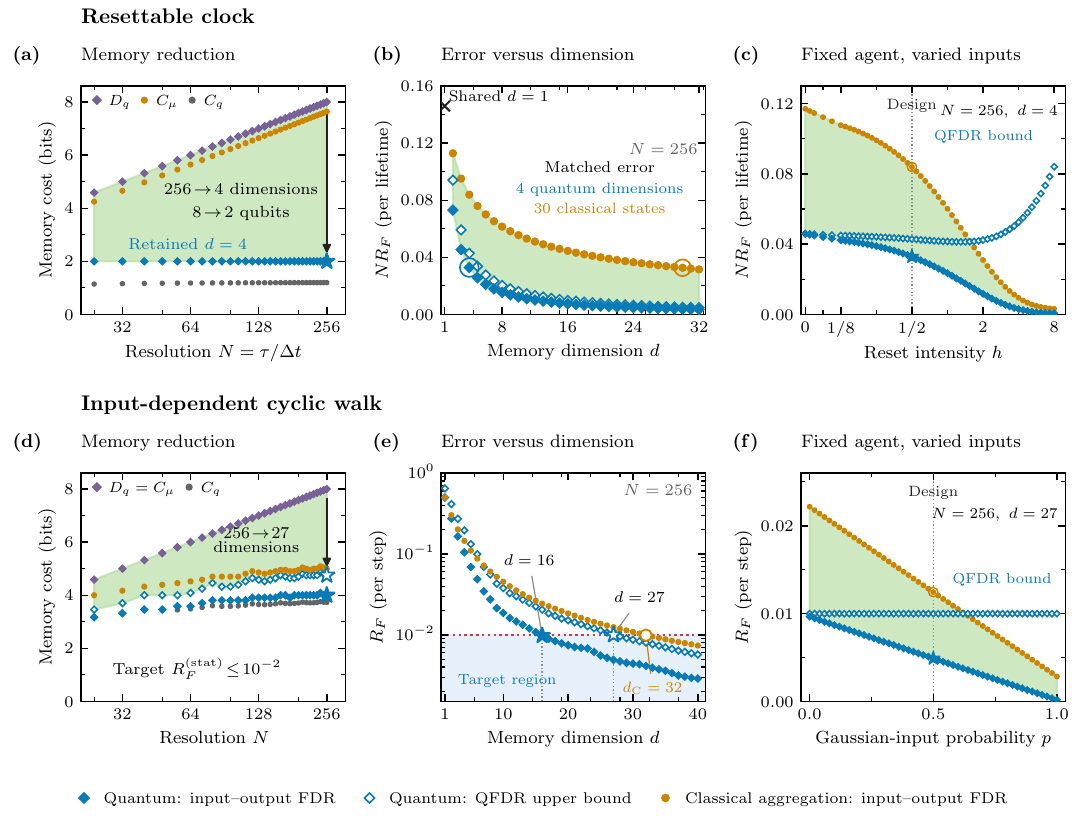}
\caption{\textbf{Memory saved and behaviour retained.}
Clock (top) and cyclic walk (bottom). Filled blue diamonds give the directly
evaluated input--output FDR of the reduced quantum agents, hollow diamonds the
corresponding QFDR upper bound, and orange circles the classical stationary
aggregation of Sec.~\ref{sec:benchmark-comparison}.
\textbf{(a,d)}~Memory costs as the resolution $N$ grows: $D_q=\log_2N$ is the
original register cost and $C_\mu,C_q$ the exact models' entropies. The clock
retains $d=4$; the walk shows the first ranks meeting the target $10^{-2}$.
\textbf{(b,e)}~Error versus retained dimension at $N=256$. Rings in (b) mark
the four-dimensional quantum clock and the first aggregated classical agent
matching its observable error; the stronger fitted classical search is
discussed in the text.
\textbf{(c,f)}~One fixed agent under varied iid inputs at $N=256$.
Both classical and quantum agents retain their design at the marked point:
$h_0=1/2$, $d=4$ for the clock; equal input probabilities, $d=27$ for the walk.
In (c), $h$ varies the reset intensity, including no external resets at $h=0$.
In (f), $p$ varies the probability of the Gaussian stimulus from zero (top-hat
only) to one (Gaussian only). Clock rates are per lifetime and walk rates per
step. Dotted red guides mark the walk target; green shading shows register
savings in (a,d) and error differences in (b,c,f).}
\label{fig:compression_results}
\end{figure*}

\subsection[Clock performance]{Clock performance}
\label{sec:clock-performance}

The clock must distinguish progressively more ages as its timestep is shortened,
yet a much smaller memory can retain its behaviour. Quantum coarse graining
has previously reduced memory for autonomous temporal dynamics without
sacrificing time resolution~\cite{Elliott2021CoarseGraining}; here the clock must
also respond to an external reset stimulus. At $N=256$, the construction reduces the original 256-dimensional encoding to four dimensions or two qubits instead of eight [Fig.~\ref{fig:compression_results}(a)]. Panel (b) shows the accuracy that accompanies this saving.
For the memory and dimension comparisons, the lifetime $\tau$ and reset intensity
$h=\gamma\tau=1/2$ are fixed, with
$r_N=1-e^{-1/(2N)}$ at $N=\tau/\Delta t$. Expressing the rate as $NR_F$ keeps
the physical unit of comparison fixed as the timestep becomes shorter.

The direct quantum reduction has lower input--output FDR than the aggregated
classical clock at every multistate point tested in panel (b). At
$N=256,d=4$, its error is about $0.033$ per lifetime against $0.084$ for
aggregation; even its QFDR certificate, about $0.043$, establishes the same
ordering. Matching the observable error requires thirty aggregated states,
highlighted in panel (b). The best fitted regenerative clock found at $d=4$ has error about $0.056$.
Within that search, eleven classical states first match the four-dimensional
quantum clock's input--output error, or seven when the comparison uses its QFDR certificate
(Supplemental Material, Table~\suppref{sm:tab:clock-matched-error}). The same
observable ordering holds at $d=2$, although there the certificate alone does
not resolve the difference from the fitted clock. The resolution sweep in
Supplemental Material, Fig.~\suppref{sm:fig:resolution-sweeps}(a), shows that
four quantum dimensions continue to reproduce the clock accurately whilst the
original memory grows with the temporal resolution.

At $d=1$ there is no internal state to carry information between interactions:
either a classical or a quantum agent responds through a law $q(y|x)$. The
cross in panel (b) gives the aggregated memoryless response; fitting that law
improves it slightly, while direct rank-one quantum truncation produces a
different response recorded in the supplement. All such one-dimensional laws
are available to either representation.

To relate the rates to finite experiments, Supplemental Material,
Sec.~\suppref{sm:subsec:clock-timing-diagnostic}, also reports fidelities over
one and two lifetimes and the distribution of the first tick after a reset.
These show how the error depends on the observation time as well as the
retained dimension.

\paragraph*{Operation away from the design reference.}
Figure~\ref{fig:compression_results}(c) instead holds the $N=256,d=4$ agents
fixed and varies their inputs. The quantum projector and repaired updates,
and the classical aggregation probabilities, are all constructed at $h_0=1/2$
and retained throughout the sweep. The quantum clock has lower input--output
FDR than this fixed classical agent at every tested reset intensity from zero
to eight, including ordinary evolution without external resets. 

The stronger comparison with a fixed fitted classical clock also favours the
quantum agent at the seven tested intensities from $1/8$ to $8$
(Supplemental Material, Table~\suppref{sm:tab:clock-off-design}). The quantum agent's QFDR
certificate establishes that ordering at the sampled intensities $h\leq1$;
at larger intensities the gap between the certificate and the visible error
widens. For the sampled $h\geq1$, reselecting the most occupied subspace at
the deployment intensity gives a larger input--output FDR than retaining the
original design. The frozen-aggregation calculation and its preparation are
described in Supplemental Material, Sec.~\suppref{sm:subsec:fixed-input-deployment}.

\subsection[Structured walk reduction]{Structured walk reduction}
\label{sec:walk-performance}

Translation symmetry makes the cyclic walk a particularly useful case. A
translation-invariant projector retains the same weight of every outgoing
memory state, so the walk realises the state-independent retention condition
of Sec.~\ref{sec:uniform-retention} exactly. Its accuracy is therefore
controlled even when later inputs depend on earlier outputs.

For the specified product encoding and iid driving, the bound is attained by
the purified-history certificate. Writing $\eps_d$ for the discarded
stationary weight, we obtain
\begin{align}
\RFQ&=-\tfrac14\log_2(1-\eps_d),\label{eq:walk-exact-certificate}\\
d_{\mathrm{cert}}(\delta_\star)
&=d_{\,1-2^{-4\delta_\star}}(\rho_M).
\label{eq:walk-spectral-rank}
\end{align}
Here $d_\zeta(\rho_M)$ is the smallest rank retaining weight at least $1-\zeta$.
Thus, within this spectral construction, the discarded memory weight itself
fixes the certified dimension. The input--output error can be smaller; the
equality proof is given in Supplemental Material,
Sec.~\suppref{sm:subsec:spectral-certification}.

After output $y$, the reduced walk memory is $P_M|\sigma_y\rangle/\sqrt w$,
so the visible input--output fidelity can also be evaluated directly.
At $N=256$, the quantum agent first reaches the input--output error target
of $10^{-2}$ at 16 dimensions, compared with 32 states for classical
aggregation [Fig.~\ref{fig:compression_results}(d,e)]. Using the QFDR upper
bound to certify the same target selects 27 quantum dimensions.

For this family at the input--output target $10^{-2}$, the retained register
cost remains close to the memory entropy: $\log_2d-C_q$ ranges from about
$0.013$ to $0.330$ bits across all 30 resolutions, with a mean of $0.21$ bits,
whilst $D_q-C_q$ increases from about $1.4$ to $4.3$ bits. The dependence on
the target is shown in Supplemental Material,
Table~\suppref{sm:tab:renyi-tracking}.

Both spectral reduction and stationary aggregation preserve the uniform
stationary position distribution. Of these two direct constructions, the
quantum reduction has lower input--output FDR at every tested $1<d<N$ in
Fig.~\ref{fig:compression_results}(e) and the associated rank scans. At the design mixture and
$d=27$, even the QFDR certificate remains below the classical aggregation
error. At $d=1$, both direct reductions give the same uniform response law.
The resolution sweep at fixed dimension is shown in Supplemental Material,
Fig.~\suppref{sm:fig:resolution-sweeps}(b).

\paragraph*{One memory for both input dynamics.}
In Fig.~\ref{fig:compression_results}(f), the $N=256,d=27$ walk is constructed
with equally probable stimuli and then driven with Gaussian-input probability
$p$ between zero and one. Its two repaired updates are held fixed, as are the
classical aggregation probabilities. The quantum agent has lower input--output
FDR at every displayed mixture, including the two single-stimulus endpoints.

The uniform position distribution is stationary for either input, so changing
this mixture leaves the original memory ensemble and the chosen spectral
subspace unchanged. Accordingly, the QFDR certificate is constant across the
sweep, whilst the directly evaluated error depends on which input is supplied.
The reduced agent therefore meets the certificate target of $10^{-2}$ for both
input dynamics and all their iid mixtures. Proposition~\ref{prop:structured-accuracy} also bounds the error when later
inputs are chosen in response to earlier outputs.

\subsection[A feedback-driven clock]{A feedback-driven clock}
\label{sec:feedback-benchmark}

To test compression in a feedback loop, we place the clock in an environment
whose reset intensity depends on its outputs. A natural tick switches the
environment to the higher intensity, whilst an external reset returns it to
the lower one; between events, the environment may also return to the lower
mode. The two physical reset intensities are $h_0=1/4$ and $h_1=3/2$, with
the full transition rules specified in Supplemental Material,
Sec.~\suppref{sm:subsec:feedback-numerics}.

The clock's memory becomes correlated with the environmental mode. We select
the retained subspace from the agent marginal of their joint stationary state
and reconnect the reduced clock to the same environment. Its outputs determine
subsequent inputs through the original feedback rules.

\begin{figure*}[tbp]
\centering\includegraphics[width=\textwidth]{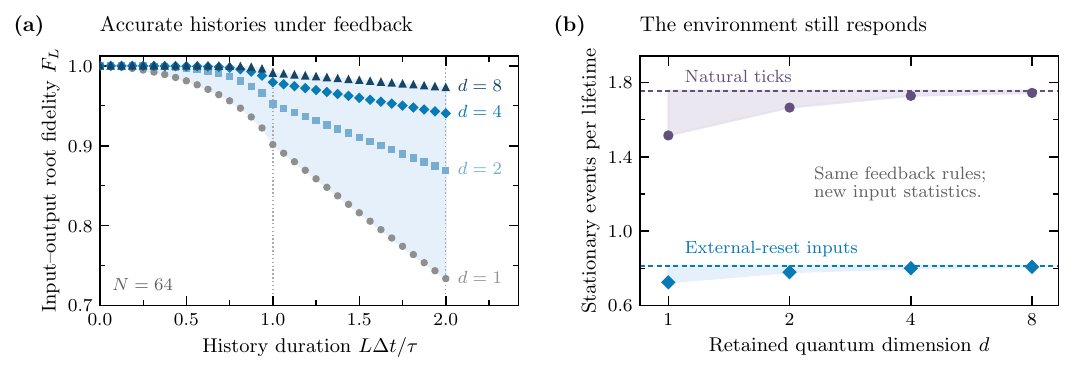}
\caption{\textbf{Retaining behaviour in a feedback loop.}
The original and reduced $N=64$ clocks interact with the same two-mode
environment, whose reset intensity responds to the clock's outputs.
\textbf{(a)}~Root fidelity of their complete input--output histories over two
clock lifetimes, including the stimuli generated in response to each clock;
vertical guides mark one and two lifetimes.
\textbf{(b)}~Stationary tick and external-reset frequencies, scaled to expected
event counts per lifetime. Dashed lines give the original process and points
the reduced processes. Shading shows how the retained dimension changes the
history fidelity and event frequencies.}
\label{fig:feedback}
\end{figure*}

Reducing the $N=64$ clock from six memory qubits to two gives input--output
history fidelities of about $0.98$ and $0.94$ over one and two lifetimes
[Fig.~\ref{fig:feedback}(a)]. These histories include the stimuli generated
by the environment in response to the clock. With eight memory dimensions,
the two-lifetime fidelity rises to about $0.97$.

Panel (b) shows the corresponding response of the environment. At $d=4$, the
stationary external-reset frequency shifts by less than two per cent, alongside
the change in natural-tick frequency induced by the reduced clock's outputs.

\subsection[Physical registers]{Physical registers}
\label{sec:registers}

Reducing the persistent memory also narrows the logical circuit needed to
realise each round. For classical stimuli and measured actions, each outcome of an instrument on
$d$ memory dimensions has Kraus rank at most $d^2$. A hidden ancilla of that
size therefore suffices, in addition to the input and output registers.
For the benchmark agents, each outcome has a single Kraus operator or
prepares a pure memory state, so a hidden ancilla of dimension at most $d$
suffices. Under the register convention of
Supplemental Material, Table~\suppref{sm:tab:implementation-width}, the
$N=256$ clock admits an 18-qubit round before reduction and a 6-qubit round
at $d=4$. The walk retains an eight-qubit output register naming its 256
positions; its corresponding widths are 25 and at most 19 qubits at $d=27$.

The register counts above apply to implementations of the measured input-output instruments. Preserving
the clock's specified coherent implementation also requires the coherences
between input branches to be retained. The environmental states of the specified coherent clock span five dimensions
at $d=4$, giving a 7-qubit round (Supplemental Material,
Sec.~\suppref{sm:subsec:quantum-tester-numerics}). The counts give achievable widths for these implementations. Gate counts,
control precision and fault-tolerant overheads require a separate choice of
hardware and compilation.

\section[Discussion and outlook]{Discussion and outlook}
\label{sec:discussion}

By introducing a route–truncate–repair procedure we have created a pipeline by which driven quantum adaptive agents can be represented by an MPS, and subsequently compressed whilst retaining high accuracy. This provides a controlled route from quantum advantages in information-theoretic memory efficiency for implementing complex strategies to practical reductions in memory dimension, bringing quantum adaptive agents closer to implementation on near-term hardware. Crucially, the resulting compressed agents remain physically valid -- not merely smaller representations -- and can respond to arbitrary input sequences, not just the driving reference.

A resettable clock retains
its behaviour in a small, fixed register as its time resolution is refined,
and does so more accurately than the classical clocks we found with the same
number of states. For the cyclic walk, translation symmetry makes the memory
spectrum itself sufficient to certify how far the register can be reduced,
even when later inputs depend on earlier outputs. Placed back into a responding
environment, a compressed clock sustains the feedback between the two, whilst
the environment itself is left unchanged.

Beyond the spectral truncation considered here, the routed MPS representation opens the door to more sophisticated tensor network optimisation methods, including variational compression at fixed dimension~\cite{Vanhecke2021Truncation,ZaunerStauber2018VUMPS}, which has proven effective for compressing quantum models of stochastic processes~\cite{martathesis,banchi2024accuracy,yang2024dimensionreductionquantumsampling} and recurrent quantum models~\cite{Lyu2026Variational}. Sharper bounds for slowly evolving memories would extend the useful duration of the accuracy guarantees. Classical memory bounds could establish quantum advantages beyond the reduction procedures compared here. The feedback and tester constructions provide a basis for studying reduced agents within control, sensing and simulation systems, where the operating environment can carry memory or quantum coherence. Assessing their implementation will require gate counts, control precision and error-correction costs alongside the register widths considered here. The quantum-comb description~\cite{Chiribella_2009,Pollock_2018,MilzModi2021} places these questions within the wider study of temporal quantum information.

\begin{acknowledgments}
\textit{Acknowledgments.} This work was funded by the University of Manchester Dame Kathleen Ollerenshaw Fellowship.
\end{acknowledgments}

\textit{Data Availability.}
The source code, figure-source data, validation outputs, and scripts required
to reproduce the numerical results are openly available in the versioned
Zenodo archive of Ref.~\cite{SundarElliott2026Software}.

\bibliographystyle{apsrev4-2}
\bibliography{references}

\ifcombinedsubmission
  \clearpage
  \onecolumngrid
  \def\supplementincluded{1}
\ifdefined\supplementincluded
\else
  \documentclass[
    aps,
    prl,
    10pt,
    superscriptaddress,
    nofootinbib,
    longbibliography,
    floatfix
  ]{revtex4-2}
  \usepackage{xr}
\usepackage[T1]{fontenc}
\usepackage[utf8]{inputenc}
\usepackage{microtype}
\usepackage{amsmath,amssymb,amsthm}
\usepackage{mathtools}
\usepackage{braket}
\usepackage{graphicx}
\usepackage{xcolor}
\usepackage{hyperref}
\usepackage{dsfont}

\hypersetup{
  colorlinks=true,
  linkcolor=blue!40!black,
  citecolor=blue!40!black,
  urlcolor=blue!40!black,
  plainpages=false,
  pdfpagelabels=true
}

\newtheorem{theorem}{Theorem}
\newtheorem{lemma}{Lemma}
\newtheorem{proposition}{Proposition}
\newtheorem{corollary}{Corollary}
\theoremstyle{definition}
\newtheorem{definition}{Definition}
\newtheorem{assumption}{Assumption}
\theoremstyle{remark}
\newtheorem{remark}{Remark}
\newtheorem{observation}{Numerical observation}
\theoremstyle{plain}

\DeclareMathOperator{\Tr}{Tr}
\DeclareMathOperator{\spr}{spr}
\DeclareMathOperator{\rank}{rank}

\newcommand{\Id}{\mathds{1}}
\newcommand{\id}{\Id}
\newcommand{\cH}{\mathcal{H}}
\newcommand{\cB}{\mathcal{B}}
\newcommand{\cR}{\mathcal{R}}
\newcommand{\eps}{\varepsilon}
\newcommand{\dTV}{d_{\mathrm{TV}}}
\newcommand{\RFQ}{R_F^{(Q)}}
\newcommand{\hRFQ}{\widehat{R}_F^{(Q)}}
\newcommand{\OmCM}{\Omega_{CM}^{(\cR)}}
\newcommand{\rhoM}{\rho_{M}^{(\cR)}}

\newcommand{\past}[1]{\overleftarrow{#1}}
\newcommand{\fut}[1]{\overrightarrow{#1}}
\newcommand{\pastfut}[1]{\overleftrightarrow{#1}}
\newcommand{\tpasto}{t_{\overleftarrow{0}}}
\newcommand{\Tpasto}{T_{\overleftarrow{0}}}
\newcommand{\tfuto}{t_{\overrightarrow{0}}}
\newcommand{\Tfuto}{T_{\overrightarrow{0}}}
\newcommand{\tpast}{\overleftarrow{t}}
\newcommand{\Tpast}{\overleftarrow{T}}
\newcommand{\tfut}{\overrightarrow{t}}
\newcommand{\Tfut}{\overrightarrow{T}}
\newcommand{\npast}{\overleftarrow{n}}
\newcommand{\Npast}{\overleftarrow{N}}
\newcommand{\nfut}{\overrightarrow{n}}
\newcommand{\Nfut}{\overrightarrow{N}}

\hypersetup{hidelinks}

  \makeatletter
  \begingroup
    \let\bibcite\@gobbletwo
    \externaldocument[main-]{main}
  \endgroup
  \begingroup
    \let\newlabel\@gobbletwo
    \def\bibcite#1#2{\expandafter\gdef\csname b@#1\endcsname{#2}}
    \externaldocument{main}
  \endgroup
  \makeatother

  \begin{document}
  \title{Supplemental Material:\\ Dimension Reduction for Quantum Adaptive Agents}
  \author{Rishi Sundar}
\affiliation{Department of Physics \& Astronomy, University of Manchester,
             Manchester M13 9PL, United Kingdom}
\affiliation{Centre for Quantum Science and Engineering, University of Manchester,
             Manchester M13 9PL, United Kingdom}
\author{Thomas J. Elliott}
\affiliation{Department of Physics \& Astronomy, University of Manchester,
             Manchester M13 9PL, United Kingdom}
\affiliation{Department of Mathematics, University of Manchester,
             Manchester M13 9PL, United Kingdom}
\affiliation{Centre for Quantum Science and Engineering, University of Manchester,
             Manchester M13 9PL, United Kingdom}

  \maketitle
\fi

\ifdefined\supplementincluded
  \newcommand{\mainref}[1]{\ref{#1}}
  \begin{center}
    {\large\bfseries Supplemental Material\\[0.35em]
    Dimension Reduction for Quantum Adaptive Agents}
  \end{center}
\else
  \newcommand{\mainref}[1]{\ref*{main-#1}}
\fi

\setcounter{section}{0}
\setcounter{subsection}{0}
\setcounter{equation}{0}
\setcounter{figure}{0}
\setcounter{table}{0}
\setcounter{theorem}{0}
\setcounter{lemma}{0}
\setcounter{proposition}{0}
\setcounter{corollary}{0}
\setcounter{definition}{0}
\setcounter{assumption}{0}
\setcounter{remark}{0}
\setcounter{observation}{0}
\setcounter{page}{1}
\setcounter{secnumdepth}{2}

\renewcommand{\thesection}{SM.\Alph{section}}
\renewcommand{\thesubsection}{\thesection.\arabic{subsection}}
\renewcommand{\theequation}{S\arabic{equation}}
\renewcommand{\thefigure}{S\arabic{figure}}
\renewcommand{\thetable}{S\arabic{table}}
\renewcommand{\thepage}{S\arabic{page}}
\renewcommand{\thetheorem}{S\arabic{theorem}}
\renewcommand{\thelemma}{S\arabic{lemma}}
\renewcommand{\theproposition}{S\arabic{proposition}}
\renewcommand{\thecorollary}{S\arabic{corollary}}
\renewcommand{\thedefinition}{S\arabic{definition}}
\renewcommand{\theassumption}{S\arabic{assumption}}
\renewcommand{\theremark}{S\arabic{remark}}
\renewcommand{\theobservation}{S\arabic{observation}}
\makeatletter
\renewcommand{\p@subsection}{}
\makeatother

\providecommand{\theHsection}{}
\providecommand{\theHsubsection}{}
\providecommand{\theHequation}{}
\providecommand{\theHfigure}{}
\providecommand{\theHtable}{}
\providecommand{\theHtheorem}{}
\providecommand{\theHlemma}{}
\providecommand{\theHproposition}{}
\providecommand{\theHcorollary}{}
\providecommand{\theHdefinition}{}
\providecommand{\theHassumption}{}
\providecommand{\theHremark}{}
\providecommand{\theHobservation}{}
\renewcommand{\theHsection}{SM.\Alph{section}}
\renewcommand{\theHsubsection}{\theHsection.\arabic{subsection}}
\renewcommand{\theHequation}{SM.\arabic{equation}}
\renewcommand{\theHfigure}{SM.\arabic{figure}}
\renewcommand{\theHtable}{SM.\arabic{table}}
\renewcommand{\theHtheorem}{SM.\arabic{theorem}}
\renewcommand{\theHlemma}{SM.\arabic{lemma}}
\renewcommand{\theHproposition}{SM.\arabic{proposition}}
\renewcommand{\theHcorollary}{SM.\arabic{corollary}}
\renewcommand{\theHdefinition}{SM.\arabic{definition}}
\renewcommand{\theHassumption}{SM.\arabic{assumption}}
\renewcommand{\theHremark}{SM.\arabic{remark}}
\renewcommand{\theHobservation}{SM.\arabic{observation}}

\section{Notation, formalisms, and basic wiring}
\label{sm:sec:notation}

This Supplement employs a particular object throughout: the routed one-step isometry. Different
sections wire this object in different ways. We first fix the viewpoints and notation so
that the later proofs can refer to them without reintroducing the construction.

\subsection{Three viewpoints}

\begin{description}
\item[(i) Routed control-site description]
The adaptive update is written with an explicit stimulus register $\mathsf{X}$. This is the
one-step controlled process before any reference input process is chosen.

\item[(ii) History-iMPS description under stationary reference driving]
After choosing a stationary reference input process \(\cR\), the reference is
routed through the control rail and contracted into the agent, giving a translation-invariant
history tensor whose bond is the joint reference--agent memory \(\cH_C\otimes\cH_M\) (Sec.~\ref{sm:sec:reference}).
Memoryless driving (\(|C|=1\)) is the special case in which the bond factorises to the agent
memory $\mathsf{M}$ alone; this is the regime of the iid clock and walk comparisons.

\item[(iii) Coherent comb/process-tensor description]
If the stimulus rail is kept as an input-output wire rather than contracted with a diagonal
reference input process, iterating the same routed update defines a stationary comb. This is the natural
object for coherent or temporally-correlated interventions.
\end{description}

Sec.~\ref{sm:sec:trees} introduces a branching-tree representation for classical counterfactual stimulus
histories. That tree is a bookkeeping device for trajectory conditioning; it agrees with the
routed description after conditioning it on a classical stimulus word.

\subsection{Registers and alphabets}

We work with finite stimulus and action alphabets $\mathcal{X}$ and $\mathcal{Y}$, with computational bases
\(\{\ket{x}\}_{x\in \mathcal{X}}\) and \(\{\ket{y}\}_{y\in \mathcal{Y}}\). The agent carries a finite-dimensional
quantum memory system $\mathsf{M}$ with Hilbert space \(\cH_M\).

At one time step we use the following registers:
\begin{itemize}
  \item The stimulus register \(\mathsf{X}\).
  \item The action register \(\mathsf{Y}\).
  \item The inaccessible environment register \(\mathsf{E}\).
\end{itemize}

\subsection{One-step controlled instrument and routed isometry for quantum agents}

For each stimulus symbol \(x\in \mathcal{X}\), let
\[
\{K^{(x)}_{y,\eta}:\cH_M\to\cH_M\}_{y,\eta}
\]
be Kraus operators satisfying
\begin{equation}
\sum_{y,\eta} K^{(x)\dagger}_{y,\eta} K^{(x)}_{y,\eta} = \Id_M.
\label{sm:eq:kraus_completeness}
\end{equation}
A Stinespring isometry for the input-conditioned instrument is
\begin{equation}
V_x
=
\sum_{y,\eta}
K^{(x)}_{y,\eta}\otimes \ket{y}_Y \otimes \ket{\eta}_E
:\cH_M\to \cH_M\otimes\cH_Y\otimes\cH_E.
\label{sm:eq:stinespring_isometry}
\end{equation}
We package the stimulus dependence into the routed isometry
\begin{equation}
W
=
\sum_{x\in \mathcal{X}}
\ket{x}\!\bra{x}_X \otimes V_x
:\cH_M\otimes\cH_X
\to
\cH_M\otimes\cH_Y\otimes\cH_X\otimes\cH_E.
\label{sm:eq:router_def}
\end{equation}
The map \(W\) is block-diagonal in the computational basis of \(X\) and forwards the
stimulus register unchanged.

\subsection{Stationary memoryless driving and coherent labels}

Let \(\cR\) be a stationary memoryless reference input process with distribution
\(p_{\cR}(x)\). The corresponding input stimulus at each timestep for this process corresponds to the diagonal classical state
\begin{equation}
\omega_{\cR}^X
=
\sum_x p_{\cR}(x)\ket{x}\bra{x}.
\label{sm:eq:omega_diag}
\end{equation}
In the routed description, one may contract the coherent form of this stimulus input, $\sum_x \sqrt{p_{\cR}(x)}\ket{x}$, into \(W\) to obtain the one-step history isometry
\begin{equation}
U_{\cR}
=
\sum_{x,y,\eta}
\sqrt{p_{\cR}(x)}\,
K^{(x)}_{y,\eta}\otimes \ket{x}_X\otimes \ket{y}_Y\otimes \ket{\eta}_E
:
\cH_M \to \cH_M\otimes\cH_X\otimes\cH_Y\otimes\cH_E.
\label{sm:eq:memoryless_history_isometry}
\end{equation}
The register $\mathsf{X}$ is part of the emitted history degrees of freedom. When the other physical legs
\((X,Y,E)\) are contracted with their conjugates, orthogonality of the label states removes
cross terms and yields the classical average over stimuli. Thus the doubled transfer map of
the pure history tensor is exactly the driven memory channel \(\Phi_{\cR}\).

\subsection{Time iteration as a wiring recipe}

\paragraph*{Reduced history-iMPS view.}
At each time step \(t\), one:
\begin{enumerate}
  \item Prepares a fresh coherent label encoding of the fixed reference input process \(\cR\).
  \item Applies the one-step history isometry \(U_{\cR}\) to the current memory. This is realisable as a unitary operator when supplemented with ancillae initialised in some fiducial state $\ket{0}$.
  \item Records the emitted registers \((X_t,Y_t,E_t)\) as the physical degrees of freedom
  at that time step.
  \item Identifies the outgoing memory as the bond input to time step \(t+1\).
\end{enumerate}

\paragraph*{Coherent comb view.}
Apply the same local routed isometry \(W\), but keep the stimulus input-output wire open
rather than contracting it with a diagonal reference input process.

\subsection{Stationary bulk regime}

Stationary memoryless driving induces the memory channel
\begin{equation}
\Phi_{\cR}(\rho)
=
\sum_x p_{\cR}(x)\sum_{y,\eta}
K^{(x)}_{y,\eta}\,\rho\,K^{(x)\dagger}_{y,\eta}.
\label{sm:eq:memoryless_reference_channel}
\end{equation}
We assume \(\Phi_{\cR}\) is primitive, such that it has a unique full-rank stationary
state \(\rho_\star^{(\cR)}\) and converges exponentially to it. In the corresponding
history-iMPS language, this is the stationary bulk regime in which boundary effects are
subextensive.

\subsection{Observed input--output process and readout map}

The reference-weighted observed process is obtained from the dilation by a fixed
readout map. At each time step, trace out inaccessible environment registers and any
hidden reference-memory transition labels, and measure the stimulus label and
action register in the computational basis. For a memoryless reference this means
measuring the registers \(X_t\) and \(Y_t\) and tracing out \(E_t\). For a
finite-memory reference it means recording \(x_t,y_t\) and coarse-graining over
\(c_t,c_{t+1}\), and \(\eta_t\). For a length-\(L\) window this defines a CPTP map
\begin{equation}
\Lambda_L: \rho_{\mathrm{window}} \mapsto p(x_{0:L-1},y_{0:L-1}),
\label{sm:eq:readout_map}
\end{equation}
used in the data-processing arguments of Sec.~\ref{sm:sec:certificates}. If one further marginalises the
stimulus string, the output-only process is obtained by another classical
stochastic map, so the same data-processing bound applies a fortiori to output
marginals.
\subsection{Transfer-map and canonical-form conventions}

We use the Hilbert--Schmidt identification
\[
\mathrm{vec}:\cB(\cH_M)\to \cH_M\otimes\cH_M,
\qquad
\mathrm{vec}(\ket{i}\!\bra{j})=\ket{i}\otimes\ket{j}.
\]
With this convention,
\begin{equation}
\mathrm{vec}(BZA^\dagger)
=
(B\otimes \overline A)\,\mathrm{vec}(Z).
\label{sm:eq:vec_convention}
\end{equation}
Thus a map \(\Phi(Z)=\sum_\alpha L_\alpha Z R_\alpha^\dagger\) is represented by
\(\sum_\alpha L_\alpha\otimes \overline{R_\alpha}\).

For a uniform history iMPS with site tensors \(\{A^s\}_s\), the doubled transfer map is
\[
\mathcal T(\cdot)=\sum_s A^s \cdot A^{s\dagger},
\]
and the mixed transfer map used for overlaps is
\[
\mathcal E_{\widetilde A,A}(\cdot)=\sum_s \widetilde A^s \cdot A^{s\dagger}.
\]
The corresponding matrix representation of the mixed transfer is
\[
M_{\widetilde A,A}
=
\sum_s \widetilde A^s\otimes \overline{A^s}.
\]
Some numerical routines use the opposite vectorisation order, in which the Kronecker factors
appear reversed. The two representations are permutation-similar and therefore have the same
spectrum.

In canonical gauge, the left fixed point of \(\mathcal T\) is \(\Id\) and the right fixed
point is the bond state across a time cut. For the stationary history iMPS associated with a
fixed memoryless reference input process \(\cR\), that right fixed point is
\(\rho_\star^{(\cR)}\). Writing
\[
\rho_\star^{(\cR)}
=
\sum_{i=1}^{d_q} \lambda_i \ket{i}\!\bra{i},
\]
we define the canonical truncation projector
\[
P=\sum_{i=1}^{\tilde{d}_q} \ket{i}\!\bra{i}
\]
and discarded weight
\[
\eps_M=\sum_{i>\tilde{d}_q} \lambda_i.
\]

\section{Reference input processes and the routed canonical bond}
\label{sm:sec:reference}

This Section establishes the definitional base of the formalism: The routed
construction for stationary finite-memory, output-independent reference input
processes, proving Main Text Theorem~\mainref{thm:routed_history}. The memoryless case used in the
benchmarks is recovered as the one-state special case in Sec.~\ref{sm:sec:canonical}.

\begin{definition}[Stationary finite-memory reference input process]
A \emph{reference input process} $\cR$ consists of a finite classical
memory set $\mathcal{C}$ and a transition--emission kernel
\begin{equation}
R(x,c'\mid c)\ge 0,\qquad \sum_{x\in X,\,c'\in C}R(x,c'\mid c)=1
\quad\text{for every }c\in \mathcal{C}.
\label{sm:eq:reference_kernel}
\end{equation}
At each time step the process holds memory $c$, emits stimulus $x$, and
updates to $c'$ with probability $R(x,c'\mid c)$. Crucially, $R$ does not
depend on the agent's actions: $\cR$ is an open-loop operating ensemble, not
a reactive environment.
\end{definition}

\begin{definition}[Routed Kraus operators]
Let the agent have stimulus-conditioned instruments
$\{K^{(x)}_{y,\eta}\}_{y,\eta}$ with
$\sum_{y,\eta}K^{(x)\dagger}_{y,\eta}K^{(x)}_{y,\eta}=\Id_M$ for every $x$.
On $\cH_C\otimes\cH_M$, with $\cH_C=\operatorname{span}\{\ket{c}\}_{c\in \mathcal{C}}$,
define
\begin{equation}
L_{c,x,c',y,\eta}
\;:=\;
\sqrt{R(x,c'\mid c)}\;\ket{c'}\!\bra{c}\otimes K^{(x)}_{y,\eta}.
\label{sm:eq:routed_kraus}
\end{equation}
The full branch label is $\omega=(c,x,c',y,\eta)$: the reference transition
$(c,c')$, the stimulus $x$, the action $y$, and the agent's environment
label $\eta$.
\end{definition}

\begin{lemma}[Trace preservation of the routed transfer channel]
\label{sm:lem:routed_isometry}
The map
\begin{equation}
\Phi_{\cR}(\Omega)
=\sum_{c,x,c',y,\eta}
L_{c,x,c',y,\eta}\,\Omega\,L^{\dagger}_{c,x,c',y,\eta}
\label{sm:eq:routed_channel}
\end{equation}
is completely positive and trace preserving on $\cB(\cH_C\otimes\cH_M)$.
\end{lemma}

\begin{proof}
Complete positivity is immediate from the Kraus form. For trace
preservation,
\[
L^{\dagger}_{c,x,c',y,\eta}L_{c,x,c',y,\eta}
=R(x,c'\mid c)\,\ket{c}\!\bra{c}\otimes
K^{(x)\dagger}_{y,\eta}K^{(x)}_{y,\eta}.
\]
Summing over $(y,\eta)$ with the per-stimulus completeness relation, then
over $(x,c')$ with Eq.~\eqref{sm:eq:reference_kernel}, and finally over $c$, gives
$\sum_{\omega}L^{\dagger}_{\omega}L_{\omega}
=\sum_{c}\ket{c}\!\bra{c}\otimes\Id_M=\Id_{C}\otimes\Id_M$.
\end{proof}

\begin{remark}[Left-canonical by construction]
\label{sm:rem:left_canonical}
Lemma~\ref{sm:lem:routed_isometry} states precisely that the routed site-tensor family
$\{L_{\omega}\}_{\omega}$ is in \emph{left-canonical} gauge:
$\sum_\omega L_\omega^\dagger L_\omega=\Id_{CM}$. No gauge fixing is
required; instrument completeness together with normalisation of the
reference kernel place the routed history tensor in canonical form
automatically. This observation is the engine of the whole construction. The same remark covers the
memoryless case, where the site tensor is
$A^{(x,y,\eta)}=\sqrt{p_{\cR}(x)}\,K^{(x)}_{y,\eta}$.
\end{remark}

\begin{remark}[Why the reference transitions carry Kraus labels]
\label{sm:rem:no_coherent_sum}
It is tempting to sum coherently over the classical reference transitions
and define the single tensor
$A^{x,y,\eta}=\sum_{c,c'}\sqrt{R(x,c'\mid c)}\,\ket{c'}\!\bra{c}\otimes
K^{(x)}_{y,\eta}$. This tensor is \emph{not} left-canonical in general: its
normalisation contains the off-diagonal terms
$\sum_{x,c'}\sqrt{R(x,c'\mid c)R(x,c'\mid d)}\,\ket{c}\!\bra{d}$, $c\neq d$,
which vanish only under special unifilarity/orthogonality conditions on
$R$. To preserve the classical character of the reference memory, the
reference transition must be tensorized as a classical instrument with
explicit branch label $(c,x,c')$, as in Eq.~\eqref{sm:eq:routed_kraus}.
\end{remark}

\begin{assumption}[Mixing]
\label{sm:ass:primitive_reference}
$\Phi_{\cR}$ is primitive on $\cH_C\otimes\cH_M$ (unique full-rank
stationary state, exponential convergence). When $\Phi_{\cR}$ is not
primitive, all statements below hold after restricting to a chosen recurrent
sector of $\Phi_{\cR}$ and replacing ``full rank'' by ``full rank on the
sector''.
\end{assumption}

\begin{proposition}[Canonical bond state of the routed reference-weighted
history]
\label{sm:prop:canonical_bond_state}
Under Assumption~\ref{sm:ass:primitive_reference}, iterating the routed isometry
\[
U_{\cR}
=\sum_{\omega}L_{\omega}\otimes\ket{\omega}_{\mathrm{phys}}
:\;\cH_C\otimes\cH_M\;\to\;\cH_C\otimes\cH_M\otimes\cH_{\mathrm{phys}}
\]
defines a uniform pure history iMPS on the time line with bond space
$\cH_C\otimes\cH_M$ and physical label $\omega$. Its transfer map equals
$\Phi_{\cR}$, and in the left-canonical gauge of
Remark~\ref{sm:rem:left_canonical} the canonical bond state across any time cut
is the joint stationary reference--agent state
\begin{equation}
\Phi_{\cR}\big(\OmCM\big)=\OmCM,\qquad \Tr\OmCM=1.
\end{equation}
The squared Schmidt coefficients across a cut are the eigenvalues of
$\OmCM$.
\end{proposition}

\begin{proof}
$U_{\cR}$ is an isometry by Lemma~\ref{sm:lem:routed_isometry}, and the same isometry
is applied at every step, so the history state is a uniform iMPS with bond
$\cH_C\otimes\cH_M$. Tracing the physical label of one doubled site gives
$\mathcal T(X)=\sum_\omega L_\omega X L_\omega^\dagger=\Phi_{\cR}(X)$,
because the orthonormal labels $\ket{\omega}$ remove all cross terms. In
left-canonical gauge the left fixed point of $\mathcal T$ is $\Id_{CM}$
(Remark~\ref{sm:rem:left_canonical}) and the right fixed point is the bond state
across a cut; under Assumption~\ref{sm:ass:primitive_reference} that fixed point is the
unique stationary state $\OmCM$. Purity of the global history state then
identifies the eigenvalues of $\OmCM$ with the squared Schmidt coefficients,
exactly as in the memoryless proof.
\end{proof}

\begin{remark}[Scope of the stationarity statement]
Routing a stationary reference input process through the control legs
converts the agent's controlled dynamics, under that operating ensemble,
into a stationary temporal tensor network --- the reference-weighted
controlled history. Controlled counterfactual histories do not in general
become stationary networks, and the certificate is correspondingly relative
to $\cR$.
\end{remark}

\subsection{Finite-memory feedback environments}
\label{sm:sec:feedback-new}
Feedback makes the input history depend on the agent's previous actions.
We describe this dependence with a finite classical environment: in state $c$
it chooses $x$ with probability $q(x|c)$ and updates to $c'$ according to
$E(c'|c,x,y)$ after the agent emits $y$. The extra dependence on $y$ leaves
the normalisation argument intact. For the operators in Main Text
Eq.~\mainref{eq:feedback-routing},
\begin{align}
\sum_{x,y,c',\eta}L_{c,x,y,c',\eta}^{\mathrm{fb}\dagger}
L_{c,x,y,c',\eta}^{\mathrm{fb}}
&=|c\rangle\langle c|\otimes
\sum_xq(x|c)\sum_{y,\eta}K_{y,\eta}^{(x)\dagger}K_{y,\eta}^{(x)}\\
&=|c\rangle\langle c|\otimes\Id_M.
\end{align}
Summing over $c$ gives $\Id_{CM}$. Thus the repeated history isometry remains
left canonical and its doubled transfer is the joint CPTP channel. The
canonical bond is again identified with stationary memory, with the support
and mixing assumptions now imposed on this interacting joint channel.

Its stationary state $\Omega_{CM}$ supplies the agent marginal
$\rho_M=\Tr_C\Omega_{CM}$ and retained projector $\Id_C\otimes P_M$.
The environment's internal state $c$ is not an input available to the repaired
agent: only the supplied stimulus $x$ selects its update. Reconnecting that
agent to the unchanged rules $q,E$ retains the original interaction structure.
Its new outputs may change the environment's later inputs, so the resulting
joint stationary state is calculated separately. This lets the accuracy test
include the consequences of compression for both sides of the interaction.

For a finite experiment starting from a correlated state, applying the local
encoder in Main Text Eq.~\mainref{eq:local-encoder} preserves the environment
marginal. The proof in Sec.~\ref{sm:sec:physical-new} consequently applies without
assuming that $\Omega_{CM}$ commutes with the retained projector. This encoded
start and the reduced channel's own stationary start are different preparations.
The numerical example in Sec.~\ref{sm:subsec:feedback-numerics} evaluates both and
uses the latter for its principal stationary-history comparison.

\section{Branching tree representation and trajectory conditioning}
\label{sm:sec:trees}

This Section formalises the branching picture used in Fig.~\mainref{fig:routed_history_network}(a). The tree is
not introduced as a computational method. Its role is to make precise what goes wrong if the
stimulus is treated as a branch label rather than as a wire: classical conditioning on a
single realised stimulus word selects one path, but the unconditioned object is a tree rather than
a one-dimensional transfer network.

\subsection{One-step branching tensor}

Fix stimulus-conditioned Kraus families
\(\{K^{(x)}_{y,\eta}\}_{y,\eta}\) for each \(x\in \mathcal{X}\), satisfying
\[
\sum_{y,\eta} K^{(x)\dagger}_{y,\eta}K^{(x)}_{y,\eta}=\Id_M,
\]
and let
\begin{equation}
V_x
=
\sum_{y,\eta}
K^{(x)}_{y,\eta}\otimes \ket{y}_{Y^{(x)}}\otimes \ket{\eta}_{E^{(x)}}
:\cH_M\to \cH_{M^{(x)}}\otimes\cH_{Y^{(x)}}\otimes\cH_{E^{(x)}}
\label{sm:eq:tree_stinespring_isometry}
\end{equation}
be a Stinespring isometry whose output memory is written into an \(x\)-labelled child branch.

Introduce, for each \(x\in \mathcal{X}\), branch registers
\((M^{(x)},Y^{(x)},E^{(x)})\) and a branch-label register \(x_{\mathrm{flag}}\)
with basis \(\{\ket{x}\}_{x\in \mathcal{X}}\). Choose a fixed fiducial state \(\ket{0}\) for all branch
registers not selected at that step.

Define the non-normalised one-step branching tensor
\begin{equation}
T:\cH_M\to
\Big(\bigotimes_{x\in \mathcal{X}}\cH_{M^{(x)}}\otimes\cH_{Y^{(x)}}\otimes\cH_{E^{(x)}}\Big)
\otimes \cH_{x_{\mathrm{flag}}}
\label{sm:eq:branching_tensor}
\end{equation}
by
\begin{equation}
T
=
\sum_{x\in \mathcal{X}}
\Bigg[
\big(V_x\big)_{M\to M^{(x)}Y^{(x)}E^{(x)}}
\otimes
\Big(\!\!\bigotimes_{x'\neq x}\ket{0}_{M^{(x')}Y^{(x')}E^{(x')}}\Big)
\Bigg]
\otimes \ket{x}_{x_{\mathrm{flag}}}.
\label{sm:eq:branching_tensor_explicit}
\end{equation}
Iterating the same local tensor at every node generates a uniform rooted \(|\mathcal{X}|\)-ary tree.

\subsection{Trajectory conditioning}

Fix a realised classical stimulus word \(x_{0:L-1}\in \mathcal{X}^L\). Trajectory conditioning means that,
at each depth \(t\), one retains only the branch labelled by \(x_t\) and traces all other
branches. Equivalently, one follows the unique path selected by \(x_{0:L-1}\) through the
tree and applies the standard readout on the retained output registers.

This produces:
\begin{itemize}
  \item A conditioned memory CP map from the root memory to the depth-\(L\) memory.
  \item A classical output distribution for the observed actions along that path.
\end{itemize}

\begin{proposition}[Conditioned tree and routed description coincide]
\label{sm:prop:tree_router_equiv}
Let $W=\sum_{x\in X}\ket{x}\!\bra{x}_X\otimes V_x$ be the routed isometry of the Main Text, with \(V_x\) as in
Eq.~\eqref{sm:eq:tree_stinespring_isometry}. Fix any horizon \(L\), any classical stimulus word
\(x_{0:L-1}\in \mathcal{X}^L\), and any initial memory state \(\rho\).

Then the following two constructions give the same objects:
\begin{enumerate}
  \item[(i)] Conditioning the branching tree along the path \(x_{0:L-1}\) and tracing out all
  off-path branches.
  \item[(ii)] Conditioning the routed description by conditioning the control wire to \(x_t\) at
  each time step and tracing out the inaccessible environment.
\end{enumerate}
More precisely, both constructions induce the same conditioned memory map
\begin{align}
\Phi_{x_{0:L-1}}(\rho)
:=
\sum_{y_{0:L-1},\eta_{0:L-1}}
&K^{(x_{L-1})}_{y_{L-1},\eta_{L-1}}\cdots
K^{(x_0)}_{y_0,\eta_0}\,
\rho K^{(x_0)\dagger}_{y_0,\eta_0}\cdots
K^{(x_{L-1})\dagger}_{y_{L-1},\eta_{L-1}},
\label{sm:eq:conditioned_map}
\end{align}
and, after readout, the same classical output distribution on
\(y_{0:L-1}\).
\end{proposition}

\begin{proof}
In the routed description, conditioning the control wire at time \(t\) to \(x_t\) reduces
the one-step map to \(V_{x_t}\). Tracing over the environment then gives the Kraus sum
\[
\rho \mapsto \sum_{y_t,\eta_t}
K^{(x_t)}_{y_t,\eta_t}\,\rho\,K^{(x_t)\dagger}_{y_t,\eta_t}.
\]
Composing these maps over \(t=0,\dots,L-1\) yields
Eq.~\eqref{sm:eq:conditioned_map}.

In the branching tree, trajectory conditioning retains exactly the branch written by the
\(x_t\)-term of Eq.~\eqref{sm:eq:branching_tensor_explicit} at depth \(t\) and traces all other
branches. Along the retained path, the only non-fiducial dynamics is again \(V_{x_t}\).
The resulting contraction therefore yields the same ordered product of Kraus operators and
hence the same conditioned memory map. Applying the same readout---measure the retained
action registers in the computational basis and trace all environments---gives the same
classical output distribution.
\end{proof}

\subsection{Caution: Coherent stimuli are not trajectory conditioning}

The coincidence above is a statement about \emph{classical} conditioning. It applies when one
conditions on a realised classical stimulus word in the computational basis.

For coherent superpositions of stimuli, tracing ``unchosen branches'' is not a physical
linear operation. In that regime the tree picture is no longer the correct coherent
description. The appropriate object is the routed isometry \(W\) with an explicit stimulus
wire, or equivalently the comb/process tensor viewpoint of Sec.~\ref{sm:sec:comb}.

\section{Uniform history iMPS under memoryless reference driving}
\label{sm:sec:history_imps}

We now give the explicit construction for stationary memoryless reference driving --- the
regime of the benchmarks, and the one-state special case of the routed canonical bond of
Sec.~\ref{sm:sec:reference}. In this regime the routed one-step update admits a reduced history-iMPS description
whose bond space is the agent memory \(\mathsf{M}\).

\subsection{From routed update to one-step history isometry}

Recall the routed isometry
\begin{equation}
W
=
\sum_{x\in \mathcal{X}}
\ket{x}\!\bra{x}_X\otimes V_x,
\qquad
V_x
=
\sum_{y,\eta}
K^{(x)}_{y,\eta}\otimes \ket{y}_Y\otimes \ket{\eta}_E,
\label{sm:eq:routed_reference_isometry}
\end{equation}
acting as
\[
W:\cH_M\otimes\cH_X \to \cH_M\otimes\cH_Y\otimes\cH_X\otimes\cH_E.
\]

Fix a stationary memoryless reference input process \(\cR\) with distribution
\(p_{\cR}(x)\). As explained in Sec.~\ref{sm:sec:notation}, contracting the routed update with this diagonal
reference input process, while retaining an orthogonal label for the sampled stimulus, gives the one-step
history isometry
\begin{equation}
U_{\cR}
=
\sum_{x,y,\eta}
\sqrt{p_{\cR}(x)}\,
K^{(x)}_{y,\eta}\otimes \ket{x}_A\otimes \ket{y}_Y\otimes \ket{\eta}_E
:
\cH_M \to \cH_M\otimes\cH_A\otimes\cH_Y\otimes\cH_E.
\label{sm:eq:history_isometry}
\end{equation}
This is the isometric version of the local history tensor \(A_{\cR}\)
introduced below. The notation \(A_{\cR}\) emphasizes the local history tensor, while \(U_{\cR}\) emphasizes
the same object as an isometry from the incoming memory to the outgoing memory and physical
registers.

Since \(\sum_x p_{\cR}(x)=1\) and each \(x\)-conditioned Kraus family is complete,
\(U_{\cR}\) is an isometry:
\[
U_{\cR}^\dagger U_{\cR}=\Id_M.
\]

\subsection{Uniform history iMPS}

At each time step \(t\), introduce fresh physical registers \((\mathsf{X}_t,\mathsf{Y}_t,\mathsf{E}_t)\) and apply the
same isometry \(U_{\cR}\) to the incoming memory \(M_t\), producing an outgoing memory
\(M_{t+1}\) and the emitted physical registers. Iterating this same local map at every time
step yields a translation-invariant pure state in the time direction, i.e.\ a uniform history
iMPS with bond space \(\cH_M\).

Choosing the per-step physical basis
\[
\cH_{\mathrm{phys}}=\cH_X\otimes\cH_Y\otimes\cH_E,
\qquad
\ket{s}\equiv \ket{x}_X\ket{y}_Y\ket{\eta}_E,
\]
the corresponding site tensors are
\begin{equation}
A_{\cR}^{\,s}
=
A_{\cR}^{(x,y,\eta)}
:=
\sqrt{p_{\cR}(x)}\,K^{(x)}_{y,\eta}.
\label{sm:eq:site_tensor_main}
\end{equation}
Thus the driven history is represented by a uniform iMPS generated by the single-site tensor
family \(\{A_{\cR}^{\,s}\}_s\).

\begin{proposition}[History iMPS under stationary memoryless driving]
\label{sm:prop:history_imps}
Under stationary memoryless reference driving \(\cR\), repeated application
of the one-step history isometry \(U_{\cR}\) defines a uniform history iMPS on the
time line with bond space given by the agent memory \(\mathsf{M}\) and physical legs \((X,Y,E)\).
\end{proposition}

\begin{proof}
The same one-step isometry \(U_{\cR}\) is applied at every time step, and its output
memory is wired into the next step. The resulting time-directed tensor network is therefore
translation invariant, with bond space \(\cH_M\) and local physical space
\(\cH_A\otimes\cH_Y\otimes\cH_E\). This is precisely a uniform history iMPS.
\end{proof}

\subsection{Transfer map and stationary bulk}

The doubled transfer map of the history iMPS generated by
Eq.~\eqref{sm:eq:site_tensor_main} is
\begin{equation}
\mathcal T_{\cR}(\cdot)
=
\sum_s A_{\cR}^{\,s} \cdot A_{\cR}^{\,s\dagger}
=
\sum_x p_{\cR}(x)\sum_{y,\eta}
K^{(x)}_{y,\eta}\,\cdot\,K^{(x)\dagger}_{y,\eta}.
\label{sm:eq:transfer_equals_reference_channel}
\end{equation}
Hence
\begin{equation}
\mathcal T_{\cR}=\Phi_{\cR}.
\end{equation}
This closure on the agent memory alone is the key simplification of the stationary memoryless regime.

If \(\Phi_{\cR}\) is primitive, then \(\mathcal T_{\cR}\) has a unique
full-rank fixed point and defines the stationary bulk to which the later canonical form and
truncation arguments apply. In that regime boundary effects are subextensive and the bond
state is well-defined independent of boundary conditions.

For finitely correlated diagonal driving, the same construction survives on an enlarged
bond that includes the controller state; for coherent interventions, the natural object is
instead the comb/process-tensor viewpoint of Sec.~\ref{sm:sec:comb}.

\section{Canonical bond state in the memoryless case}
\label{sm:sec:canonical}

This Section proves Main Text Theorem~\mainref{thm:routed_history} in the memoryless case: In the routed
history-iMPS representation, the canonical bond state is the stationary memory state
\(\rho_\star^{(\cR)}\) of the driven agent. The general statement, with the joint bond
\(\OmCM\), is established in Sec.~\ref{sm:sec:reference}; this is its one-state special case.

\begin{lemma}[Transfer map equals the reference-averaged memory channel]
\label{sm:lem:transfer_equals_phi}
Fix a stationary memoryless reference input process \(\cR\) and let
\(\ket{\Psi_{\cR}}\) be the uniform history iMPS constructed in Sec.~\ref{sm:sec:history_imps}. Let
\(\mathcal T_{\cR}\) denote its doubled transfer map, obtained by contracting one
time step of the iMPS tensor with its conjugate and tracing the physical legs.

Then
\begin{equation}
\mathcal T_{\cR}(\cdot)
=
\Phi_{\cR}(\cdot)
=
\sum_x p_{\cR}(x)\sum_{y,\eta}
K^{(x)}_{y,\eta}\,\cdot\,K^{(x)\dagger}_{y,\eta}.
\label{sm:eq:transfer_equals_reference_channel_again}
\end{equation}
In particular, the right fixed point of \(\mathcal T_{\cR}\) is the stationary memory
state \(\rho_\star^{(\cR)}\).
\end{lemma}

\begin{proof}
From Sec.~\ref{sm:sec:history_imps}, the one-step history isometry is
\[
U_{\cR}
=
\sum_{x,y,\eta}
\sqrt{p_{\cR}(x)}\,
K^{(x)}_{y,\eta}\otimes \ket{x}_X\otimes \ket{y}_Y\otimes \ket{\eta}_E.
\]
By definition,
\[
\mathcal T_{\cR}(\cdot)
:=
\Tr_{X,Y,E}\!\left[\,U_{\cR} \cdot U_{\cR}^\dagger\,\right].
\]
Expanding in the orthonormal bases of \(X\), \(Y\), and \(E\) removes all cross terms and
gives
\[
\mathcal T_{\cR}(X)
=
\sum_{x,y,\eta}
p_{\cR}(x)\,
K^{(x)}_{y,\eta}\,X\,K^{(x)\dagger}_{y,\eta},
\]
which is precisely \(\Phi_{\cR}(\cdot)\). The fixed-point statement is immediate.
\end{proof}

\begin{proof}[Proof of Main Text Theorem~\mainref{thm:routed_history}, memoryless case]
By Lemma~\ref{sm:lem:transfer_equals_phi},
\[
\mathcal T_{\cR}=\Phi_{\cR}.
\]
If \(\Phi_{\cR}\) is primitive, then it has a unique full-rank stationary state
\(\rho_\star^{(\cR)}\), which is therefore the unique full-rank right fixed point of
\(\mathcal T_{\cR}\).

In canonical gauge, the left fixed point of the transfer map is \(\Id\) and the right fixed
point is the bond state across a time cut. Hence the canonical bond state is
\[
\sigma_{\cR}=\rho_\star^{(\cR)}.
\]
Since the global history state is pure, the Schmidt spectrum across a time cut is the
eigenvalue spectrum of the bond state, and the bond entropy is
\(S(\rho_\star^{(\cR)})\).
\end{proof}

\section{Agent-local retained-weight truncation}
\label{sm:sec:agent_local}

For finite-memory references the canonical bond is the joint space
$\mathsf{C}\otimes \mathsf{M}$. An unconstrained canonical truncation of $\OmCM$ compresses
the routed history representation but may also compress or entangle the
external reference memory $\mathsf{C}$; it therefore does not by itself define a
smaller physical agent. Routed-history compression and agent-memory
reduction must be distinguished. Genuine agent reduction restricts to
projectors that act on the agent factor alone.

\begin{definition}[Agent-local projector; driven-agent marginal]
An \emph{agent-local projector} is $\Pi=\Id_C\otimes P_M$ with
$P_M:\cH_M\to\cH_M$, $\rank(P_M)=\tilde{d}_q$. The \emph{driven-agent
marginal} is
\begin{equation}
\rhoM:=\Tr_C\,\OmCM.
\end{equation}
\end{definition}

\begin{lemma}[Retained-weight optimality; Main Text Corollary~\mainref{cor:discarded_weight_bound}]
\label{sm:lem:ky_fan}
For any agent-local projector, the retained stationary canonical weight is
\begin{equation}
w(P_M)=\Tr\!\big[(\Id_C\otimes P_M)\,\OmCM\big]=\Tr\!\big[P_M\,\rhoM\big].
\end{equation}
Writing $\rhoM=\sum_i\lambda_i\ket{i}\!\bra{i}$ with
$\lambda_1\ge\lambda_2\ge\cdots$, the maximum of $w(P_M)$ over rank-$\tilde{d}_q$
projectors is attained by $P_{\tilde{d}_q}=\sum_{i\le \tilde{d}_q}\ket{i}\!\bra{i}$, with
discarded stationary agent-memory weight
\begin{equation}
\eps_M(\tilde{d}_q)=1-w(P_{\tilde{d}_q})=\sum_{i>\tilde{d}_q}\lambda_i.
\end{equation}
\end{lemma}

\begin{proof}
The first equality is the definition of the partial trace. For the second,
expand in the eigenbasis of $\rhoM$:
$\Tr[P_M\rhoM]=\sum_i\lambda_i\,t_i$ with $t_i:=\braket{i|P_M|i}$. The
diagonal of a rank-$\tilde{d}_q$ projector satisfies $0\le t_i\le1$ and
$\sum_i t_i=\Tr P_M=\tilde{d}_q$: the truncation distributes a total
\emph{dimension budget} $\tilde{d}_q$ of retention across the stationary
occupancies $\lambda_i$ of the driven memory. The linear functional
$\sum_i\lambda_i t_i$ over this polytope is maximised by loading the budget
entirely onto the $\tilde{d}_q$ largest occupancies, $t_i=1$ for $i\le \tilde{d}_q$ and
$t_i=0$ otherwise, which is attained by $P_{\tilde{d}_q}$ (with the usual freedom
under degeneracy of $\lambda_{\tilde{d}_q}$). Physically: At fixed retained
dimension, keep the directions the driven memory actually occupies --- the
temporal analogue of the density-matrix truncation rule underlying DMRG.
(The extremal statement is Ky Fan's maximum principle~\cite{fan1951maximum}.)
\end{proof}

\begin{remark}[Scope of the optimality claim]
Lemma~\ref{sm:lem:ky_fan} establishes optimality for retained stationary
canonical weight \emph{only}. It does not assert optimality of $P_{\tilde{d}_q}$
for the fidelity-divergence rate, for finite-horizon total variation, for
worst-case controlled error, for closed-loop performance, or across input
processes other than $\cR$. We therefore call it the \emph{agent-local
retained-weight truncation} rather than an optimal agent compression.
\end{remark}

\section{Purification conventions and fidelity certificates}
\label{sm:sec:certificates}

The quantum fidelity divergence rate (QFDR)~\cite{Vanhecke2021Truncation,yang2024dimensionreductionquantumsampling}, denoted \(\RFQ\) in the Main
Text, compares two specified \emph{purified} history states. It is therefore
defined only after their Stinespring label alphabets have been aligned. This
section fixes that convention and derives the exact mixed-transfer
certificates used below.

\paragraph*{Common label alphabet and zero padding.}
The original and polar-completed histories use the same labels
$\omega=(c,x,c',y,\eta)$. A repair with extra hidden branches, such as the
reset completion of Sec.~\ref{sm:sec:repair}, is compared on the enlarged common alphabet by
assigning zero operators to the absent branches of the original tensor.
Tensors on different bond spaces are compared directly through a rectangular
mixed transfer. Once this choice of dilation and common alphabet is fixed,
the QFDR is unambiguous.

\paragraph*{Agent-local operators on the joint bond.}
The retained projector $P=P_M$ and all repair maps act only on the agent
memory $\cH_M$; on the joint routed bond $\cH_C\otimes\cH_M$ they act through
$\Id_C\otimes(\cdot)$. For memoryless driving $\mathsf{C}$ is trivial and this
distinction disappears.

\begin{lemma}[Exact mixed-transfer certificate]
\label{sm:lem:mixed_transfer_exact}
Let $\{A^\omega\}$ and $\{\widetilde A^{\omega}\}$ generate normalized
uniform history iMPS $\ket{\Psi}$, $\ket{\widetilde\Psi}$ over the common
alphabet, with the standard injectivity (normal/primitive) assumptions and
nonzero overlap of the stationary boundaries with the dominant mixed-transfer
sector. Then
\begin{equation}
\RFQ(\Psi,\widetilde\Psi)
=-\tfrac12\log_2\mu,
\qquad
\mu=\spr\Big(Z\mapsto\textstyle\sum_\omega\widetilde A^\omega Z
A^{\omega\dagger}\Big),
\label{sm:eq:qfdr_mu_exact}
\end{equation}
with boundary factors contributing only subextensively, and $\mu$ is invariant under gauge
transformations $\widetilde A^\omega\mapsto B\widetilde A^\omega B^{-1}$ of
either tensor.
\end{lemma}

\begin{proof}
The finite-window overlap is a contraction of $L$ copies of the mixed
transfer between fixed boundary operators; with both states normalized,
boundary contributions are subextensive and the exponential rate is the
spectral radius. Gauge invariance: conjugating $\widetilde A$ by $B$
conjugates the mixed transfer by $Z\mapsto BZ$, a similarity.
\end{proof}

\begin{lemma}[Non-negativity of the rate]
\label{sm:lem:cs_upper_bound}
Let $\nu$ be the spectral radius of the self transfer
$Z\mapsto\sum_\omega\widetilde A^\omega Z\widetilde A^{\omega\dagger}$
(so $\nu=1$ for a normalized tensor). Then $\mu\le\sqrt{\nu}\le 1$, hence
$\RFQ\ge 0$.
\end{lemma}

\begin{proof}
Cauchy--Schwarz for the overlap of finite windows of two states with norms
governed by their self transfers:
$|\braket{\widetilde\Psi_L|\Psi_L}|\le
\|\widetilde\Psi_L\|\,\|\Psi_L\|$, whose exponential rates are
$\nu^{1/2}$ and $1$.
\end{proof}

\begin{proposition}[Exact purified QFDR of a valid reduced agent]
\label{sm:prop:repaired_qfdr}
Let $\widetilde{\mathcal A}$ be any valid reduced agent on $P\cH_M$, and
let $\widetilde L_\omega$ and $L_\omega$ be its routed Kraus operators and
those of the original agent on a fixed common physical-label alphabet.  Under
the assumptions of Lemma~\ref{sm:lem:mixed_transfer_exact},
\begin{equation}
\RFQ\big(\Psi,\widetilde\Psi\big)
=-\tfrac12\log_2\mu,
\qquad
\mu=\spr\Big(Z\mapsto\textstyle\sum_\omega
\widetilde L_\omega Z L_\omega^\dagger\Big).
\label{sm:eq:raw_mixed_spectral_radius}
\end{equation}
This applies in particular to the stimulus-wise polar completion of Sec.~\ref{sm:sec:repair}.
\end{proposition}

\begin{proof}
Both routed histories are normalized uniform tensors, and their finite-window
overlap is generated by the rectangular mixed transfer in
Eq.~\eqref{sm:eq:raw_mixed_spectral_radius}.  Lemma~\ref{sm:lem:mixed_transfer_exact} gives the asymptotic rate.
\end{proof}

\begin{corollary}[Reset completion]
\label{sm:cor:reset_qfdr}
For the reset completion of Sec.~\ref{sm:sec:repair}, the original tensor is zero on the added
recovery labels.  Its QFDR therefore reduces to
\begin{equation}
\widehat R_F^{(Q)}
=-\tfrac12\log_2\mu_P,
\qquad
\mu_P=\spr\Big(Z\mapsto\textstyle\sum_\omega
(P L_\omega P)Z L_\omega^\dagger\Big),
\label{sm:eq:reset_qfdr}
\end{equation}
where $P$ denotes the agent-local projector $\Id_C\otimes P_M$ on the routed
bond.
\end{corollary}

\begin{proof}
The keep branch has tensor $PL_\omega P$, while every recovery branch pairs
with a zero tensor from the original history.  Equation~\eqref{sm:eq:reset_qfdr}
then follows from Proposition~\ref{sm:prop:repaired_qfdr}.
\end{proof}

\subsection{Data processing to the observed classical process}

Let \(\Lambda_L\) denote the fixed length-\(L\) readout map of Sec.~\ref{sm:sec:notation}. It traces
inaccessible environment registers and hidden reference-memory labels, measures
the visible stimulus and action symbols, and produces the reference-weighted
classical input--output distribution. For classical length-\(L\) input--output
distributions, we use the Bhattacharyya fidelity
\begin{equation}
F_{\mathrm{stat}}(p_L,\widetilde p_L):=
\sum_{x_{0:L-1},y_{0:L-1}}
\sqrt{p_L(x_{0:L-1},y_{0:L-1})
\widetilde p_L(x_{0:L-1},y_{0:L-1})}.
\label{sm:eq:classical_fidelity}
\end{equation}
The corresponding classical fidelity divergence rate~\cite{Yang_2020} is
\[
R_F^{(\mathrm{stat})}(p,\widetilde p):=
-\lim_{L\to\infty}\frac{1}{2L}\log_2 F_{\mathrm{stat}}(p_L,\widetilde p_L),
\]
whenever the limit exists.

\begin{lemma}[Data processing]
\label{sm:lem:dp}
Let
\[
\rho_L = |\Psi_L\rangle\!\langle\Psi_L|,
\qquad
\widetilde\rho_L = |\widetilde\Psi_L\rangle\!\langle\widetilde\Psi_L|,
\]
and let
\[
p_L=\Lambda_L(\rho_L),
\qquad
\widetilde p_L=\Lambda_L(\widetilde\rho_L)
\]
be the corresponding classical output distributions. Then
\begin{equation}
F_{\mathrm{stat}}(p_L,\widetilde p_L)\ge F(\rho_L,\widetilde\rho_L),
\label{sm:eq:dp_finite}
\end{equation}
and therefore
\begin{equation}
R_F^{(\mathrm{stat})}(p,\widetilde p)
\le
\RFQ(\Psi,\widetilde\Psi),
\label{sm:eq:dp_rate}
\end{equation}
whenever the corresponding limits exist.
\end{lemma}

\begin{proof}
Uhlmann fidelity is monotone under CPTP maps. Applying this to the fixed readout map
\(\Lambda_L\) gives Eq.~\eqref{sm:eq:dp_finite}. Dividing by \(2L\) and taking
\(L\to\infty\) yields Eq.~\eqref{sm:eq:dp_rate}.
\end{proof}

\subsection{Finite-history certificate with terminal memory traced out}
\label{sm:sec:finite-certificate-new}
We compare the emitted interaction records after tracing out each agent's
terminal memory. The resulting certificate accommodates different terminal
memory dimensions through the rectangular cross-operator below.
Let $|a_L\rangle\in\mathsf M\otimes\mathsf R_L$ and
$|b_L\rangle\in\widetilde{\mathsf M}\otimes\mathsf R_L$ be normalised histories,
where $\mathsf R_L$ comprises the common initial purifier and all emitted labels.
Their terminal memory dimensions may differ. Define the rectangular cross-operator
$Z_L=\Tr_{\mathsf R_L}|b_L\rangle\langle a_L|$.
Writing $A_0,B_0$ for the coefficient matrices of the initial purifications gives
$Z_0=B_0A_0^\dagger$; contraction of one common physical index gives
$Z_{t+1}=\sum_\omega\widetilde L_\omega Z_tL_\omega^\dagger$.

\begin{proposition}[Finite-history trace-norm certificate]
\label{sm:prop:finite-trace}
For the specified purifications and common emitted labels,
\begin{equation}
F\!\left(\Tr_{\mathsf M}|a_L\rangle\langle a_L|,
\Tr_{\widetilde{\mathsf M}}|b_L\rangle\langle b_L|\right)=\|Z_L\|_1.
\end{equation}
Reading out the same visible histories and discarding the initial purifier and
hidden labels gives
$F_L^{(\mathrm{stat})}\geq\|Z_L\|_1$ and
$\dTV(p_L,\widetilde p_L)\leq\sqrt{1-\|Z_L\|_1^2}$.
\end{proposition}
\begin{proof}
Uhlmann's theorem maximises the overlap of purifications over isometries on the
purifying terminal memories. Embedding both terminal memories in a common space
when necessary, this is the variational characterisation of $\|Z_L\|_1$.
Equivalently, the coefficient-matrix form of the same identity gives the trace
norm of their cross-product~\cite{Watrous2018TQI}. The common visible readout is a
CPTP map, so fidelity cannot decrease. The upper Fuchs--van de Graaf inequality
for root fidelity gives the stated trace-distance bound~\cite{Fuchs1999}.
\end{proof}
This calculation applies at the chosen horizon without requiring stationarity
or primitivity. It concerns the emitted history: if the output to be compared
also includes a tester's final accessible memory, that memory must be retained.
The initial purifications and hidden labels specify the certificate, whilst
the observed record remains the input--output history.

For an asymptotic rate, the relevant spectral sectors are those reached by
$Z_0$. Equality with the full spectral radius requires access to a leading
sector; a scalar terminal contraction also requires noncancellation. These
conditions concern the large-$L$ limit, not the finite-history identity above.

The certificates above compare a completed agent with the original process.
To connect that comparison to the memory discarded during compression, we next
consider physical completions of the projected dynamics. Their finite-horizon
bounds include the recovered branches as well as the retained history.

\section{Repairing a truncation to a valid agent}
\label{sm:sec:repair}

This Section proves the validity part of Main Text Theorem~\mainref{thm:repair_certificate}.
Canonical truncation of the routed history iMPS identifies a retained memory subspace, but
it does not by itself define a stimulus-conditioned quantum instrument on that subspace. If
the original agent is applied to a state supported on the retained memory, part of the
resulting state can leave the retained subspace. This section records two local completions.
Sec.~\ref{sm:subsec:polar_completion} gives the stimulus-wise polar completion used in
the Main Text. Secs.~\ref{sm:subsec:reset_completion}--\ref{sm:subsec:finite_horizon_reset}
give an alternative reset completion that preserves one-step action probabilities exactly and
admits finite-horizon bounds.

For the reset construction of
Secs.~\ref{sm:subsec:reset_completion}--\ref{sm:subsec:finite_horizon_reset}, let
\[
\rho_\star^{(\cR)}
=
\sum_{i=1}^{d_q} \lambda_i |i\rangle\!\langle i|,
\qquad
\lambda_1\ge \lambda_2\ge\cdots\ge \lambda_{d_q},
\]
be the stationary memory state of the driven channel \(\Phi_{\cR}\), and let
\[
P=\sum_{i=1}^{\tilde{d}_q} |i\rangle\!\langle i|,
\qquad
Q=\Id-P,
\qquad
\eps_M=\Tr[Q\rho_\star^{(\cR)}]
=\sum_{i>\tilde{d}_q}\lambda_i.
\]
We assume the support condition
\[
\Tr[P\rho_\star^{(\cR)}]>0,
\qquad\text{equivalently } 0\le\eps_M<1,
\]
which holds automatically for the leading-eigenspace truncations of Sec.~\ref{sm:sec:agent_local} with
\(\eps_M<1\); if instead \(\Tr[P\rho_\star^{(\cR)}]=0\) the
projector retains no stationary weight, and one may complete the instrument
with an arbitrary reset state on \(P\cH_M\), though no accuracy certificate
then applies. Under the support condition the normalised retained stationary memory
\begin{equation}
\bar\rho
=
\frac{P\rho_\star^{(\cR)}P}{1-\eps_M}
\label{sm:eq:rho_bar}
\end{equation}
is well defined.

\subsection{Action-preserving reset completion}
\label{sm:subsec:reset_completion}

As a complementary repair, define the reset channel
\begin{equation}
\mathcal C_\star:\mathcal B(Q\cH_M)\to\mathcal B(P\cH_M),
\qquad
\mathcal C_\star(\sigma)=\Tr[\sigma]\,\bar\rho.
\label{sm:eq:reset_channel}
\end{equation}
Let \(\{B_a:Q\cH_M\to P\cH_M\}_a\) be any Kraus representation of
\(\mathcal C_\star\), so that
\begin{equation}
\sum_a B_a\sigma B_a^\dagger=\Tr[\sigma]\,\bar\rho,
\qquad
\sum_a B_a^\dagger B_a=Q.
\label{sm:eq:reset_kraus}
\end{equation}
For each stimulus \(x\), define Kraus operators on the retained memory space
\(P\cH_M\) by
\begin{equation}
\widehat K^{(x)}_{y,\eta,0}=PK^{(x)}_{y,\eta}P,
\qquad
\widehat K^{(x)}_{y,\eta,a}
=
B_a QK^{(x)}_{y,\eta}P.
\label{sm:eq:reset_repair_kraus}
\end{equation}
The label \(0\) denotes the branch in which the original update remains in the retained
subspace; the labels \(a\) denote recovery branches after leakage into \(Q\cH_M\). These are
hidden recovery labels, not new observed actions. The observed action remains \(y\), and
observed statistics are obtained by summing over \(\eta\) and the recovery label.

\begin{proposition}[Validity of the reset-completed agent]
\label{sm:prop:reset_valid}
For every stimulus \(x\), the Kraus family
\(\{\widehat K^{(x)}_{y,\eta,b}\}_{y,\eta,b}\), with
\(b\in\{0\}\cup\{a\}\), defines a trace-preserving quantum instrument on
\(P\cH_M\):
\[
\sum_{y,\eta,b}
\widehat K^{(x)\dagger}_{y,\eta,b}
\widehat K^{(x)}_{y,\eta,b}
=P.
\]
Moreover, for any retained input state \(\rho=P\rho P\), the one-step action statistics
agree exactly with those of the original agent:
\[
\sum_{\eta,b}
\Tr\!\left[
\widehat K^{(x)}_{y,\eta,b}\rho
\widehat K^{(x)\dagger}_{y,\eta,b}
\right]
=
\sum_{\eta}
\Tr\!\left[
K^{(x)}_{y,\eta}\rho K^{(x)\dagger}_{y,\eta}
\right].
\]
\end{proposition}

\begin{proof}
Using Eq.~\eqref{sm:eq:reset_kraus},
\[
\begin{aligned}
\sum_{y,\eta,b}
\widehat K^{(x)\dagger}_{y,\eta,b}
\widehat K^{(x)}_{y,\eta,b}
&=
\sum_{y,\eta}
P K^{(x)\dagger}_{y,\eta}P K^{(x)}_{y,\eta}P
+
\sum_{y,\eta,a}
P K^{(x)\dagger}_{y,\eta}Q B_a^\dagger B_a QK^{(x)}_{y,\eta}P
\\
&=
\sum_{y,\eta}
P K^{(x)\dagger}_{y,\eta}(P+Q)K^{(x)}_{y,\eta}P
\\
&=
P,
\end{aligned}
\]
where the last line uses completeness of the original \(x\)-conditioned instrument. For the
action statistics, fix \(y\) and \(\rho=P\rho P\). The same decomposition gives
\[
\begin{aligned}
\sum_{\eta,b}
\Tr\!\left[
\widehat K^{(x)}_{y,\eta,b}\rho
\widehat K^{(x)\dagger}_{y,\eta,b}
\right]
&=
\sum_{\eta}
\Tr\!\left[
P K^{(x)}_{y,\eta}\rho K^{(x)\dagger}_{y,\eta}P
\right]
+
\sum_{\eta}
\Tr\!\left[
Q K^{(x)}_{y,\eta}\rho K^{(x)\dagger}_{y,\eta}Q
\right]
\\
&=
\sum_{\eta}
\Tr\!\left[
K^{(x)}_{y,\eta}\rho K^{(x)\dagger}_{y,\eta}
\right],
\end{aligned}
\]
as claimed.
\end{proof}

\paragraph*{Scope.} The repaired agent is valid for every stimulus
\(x\): the reference input process is used to choose and to certify the
truncation, not to make the repaired maps valid. In the finite-memory
setting the repair acts on the agent factor only, with
\(\rho_\star^{(\cR)}\to\rho_M^{(\cR)}\) throughout; the
reference memory \(\mathsf{C}\) is external and untouched. The leakage and
finite-horizon estimates below are stated for memoryless references.

\subsection{Leakage operator and recovery probability}
\label{sm:subsec:leakage}

For a memoryless reference input process \(\cR\), define the leakage operator
\begin{equation}
\varsigma_{\cR}
=
\sum_x p_{\cR}(x)
\sum_{y,\eta}
P K^{(x)\dagger}_{y,\eta}QK^{(x)}_{y,\eta}P,
\qquad
\gamma_{\cR}:=\|\varsigma_{\cR}\|_\infty.
\label{sm:eq:leakage_operator}
\end{equation}
For any retained memory state \(\rho=P\rho P\), the probability that one step of the
reset-completed agent uses a recovery branch under the reference input process is
\begin{equation}
\ell_{\cR}(\rho)
=
\Tr[\varsigma_{\cR}\rho].
\label{sm:eq:leakage_probability}
\end{equation}
In particular, \(\ell_{\cR}(\rho)\le \gamma_{\cR}\).

\begin{proposition}[Leakage of the retained stationary memory]
\label{sm:prop:leakage_retained_stationary}
The retained stationary memory obeys
\begin{equation}
\ell_{\cR}(\bar\rho)
\le
\frac{\eps_M}{1-\eps_M}.
\label{sm:eq:leakage_eps_bound}
\end{equation}
Moreover, if \(\lambda_{\tilde{d}_q}\) is the smallest retained eigenvalue of
\(\rho_\star^{(\cR)}\), then
\begin{equation}
\gamma_{\cR}
\le
\min\!\left\{1,\frac{\eps_M}{\lambda_{\tilde{d}_q}}\right\}.
\label{sm:eq:gamma_eps_lambda_bound}
\end{equation}
\end{proposition}

\begin{proof}
Since \(P\) is a spectral projector of \(\rho_\star^{(\cR)}\),
\[
\rho_\star^{(\cR)}
=
P\rho_\star^{(\cR)}P
+
Q\rho_\star^{(\cR)}Q.
\]
Using stationarity,
\[
\eps_M
=
\Tr[Q\rho_\star^{(\cR)}]
=
\Tr[Q\Phi_{\cR}(\rho_\star^{(\cR)})].
\]
By positivity,
\[
\Tr[Q\Phi_{\cR}(P\rho_\star^{(\cR)}P)]
\le
\Tr[Q\Phi_{\cR}(\rho_\star^{(\cR)})]
=
\eps_M.
\]
But
\[
\Tr[Q\Phi_{\cR}(P\rho_\star^{(\cR)}P)]
=
\Tr[\varsigma_{\cR}P\rho_\star^{(\cR)}P].
\]
Dividing by \(1-\eps_M\) gives
\[
\ell_{\cR}(\bar\rho)
=
\Tr[\varsigma_{\cR}\bar\rho]
\le
\frac{\eps_M}{1-\eps_M}.
\]

For the operator-norm bound, first note that
\[
0\le \varsigma_{\cR}
\le
\sum_xp_{\cR}(x)
\sum_{y,\eta}P K^{(x)\dagger}_{y,\eta}K^{(x)}_{y,\eta}P
=P,
\]
so \(\gamma_{\cR}\le1\). Also
\[
P\rho_\star^{(\cR)}P
\ge
\lambda_{\tilde{d}_q}P.
\]
Therefore
\[
\lambda_{\tilde{d}_q}\Tr[\varsigma_{\cR}]
\le
\Tr[\varsigma_{\cR}P\rho_\star^{(\cR)}P]
\le
\eps_M.
\]
Since \(\varsigma_{\cR}\ge0\), \(\|\varsigma_{\cR}\|_\infty\le\Tr[\varsigma_{\cR}]\), which
yields
\[
\gamma_{\cR}
\le
\frac{\eps_M}{\lambda_{\tilde{d}_q}}.
\]
Combining the two bounds proves Eq.~\eqref{sm:eq:gamma_eps_lambda_bound}.
\end{proof}

\subsection{Finite-horizon bound for the reset completion}
\label{sm:subsec:one_step_reset_error}
\label{sm:subsec:finite_horizon_reset}

For iid driving, the probability of leaving the retained subspace gives a
complementary estimate. The bound below uses transition leakage rather than
summing the discarded occupation of the original experiment; neither quantity
uniformly controls the other. We retain this estimate alongside Main Text
Theorem~\mainref{thm:physical-accuracy} and the conditional refinement in
Sec.~\ref{sm:subsec:conditional-leakage}.

For each stimulus~$x$ and action~$y$, let
\[
\mathcal I_y^{(x)}(\rho)
:=
\sum_\eta K^{(x)}_{y,\eta}\rho K^{(x)\dagger}_{y,\eta}.
\]
For a retained input state $\rho=P\rho P$, the reset completion gives
\[
\widehat{\mathcal I}_{y,\mathrm{rec}}^{(x)}(\rho)
=
P\mathcal I_y^{(x)}(\rho)P
+
\mathcal C_\star\!\left(Q\mathcal I_y^{(x)}(\rho)Q\right).
\]
To retain the visible outcome record, define
\begin{align}
\mathcal J_{\cR}(\rho)
&:=
\sum_{x,y}p_{\cR}(x)\ket{x,y}\!\bra{x,y}\otimes
\mathcal I_y^{(x)}(\rho),
\notag\\
\widehat{\mathcal J}^{\mathrm{rec}}_{\cR}(\rho)
&:=
\sum_{x,y}p_{\cR}(x)\ket{x,y}\!\bra{x,y}\otimes
\widehat{\mathcal I}_{y,\mathrm{rec}}^{(x)}(\rho).
\label{sm:eq:outcome_recorded_channels}
\end{align}

\begin{lemma}[One-step outcome-recorded reset error]
\label{sm:lem:one_step_reset_error}
For every retained state $\rho=P\rho P$,
\begin{equation}
\frac12\left\|
\mathcal J_{\cR}(\rho)-
\widehat{\mathcal J}^{\mathrm{rec}}_{\cR}(\rho)
\right\|_1
\le
\sqrt{\ell_{\cR}(\rho)}+\ell_{\cR}(\rho)
\le
\sqrt{\gamma_{\cR}}+\gamma_{\cR}.
\label{sm:eq:one_step_reset_error}
\end{equation}
\end{lemma}

\begin{proof}
Set $\Omega=\mathcal J_{\cR}(\rho)$,
$\Pi=\Id_{XY}\otimes P$, and $\bar\Pi=\Id-\Pi$. Then
$\Tr(\bar\Pi\Omega)=\ell_{\cR}(\rho)$. The reset completion removes the
$\Pi$--$\bar\Pi$ coherences and replaces the $\bar\Pi$ block by a positive
operator of the same trace, while leaving the classical record unchanged.
Consequently,
\[
\frac12\left\|\Omega-
\widehat{\mathcal J}^{\mathrm{rec}}_{\cR}(\rho)\right\|_1
\le
\|\Pi\Omega\bar\Pi\|_1+\Tr(\bar\Pi\Omega).
\]
Positivity gives
$\|\Pi\Omega\bar\Pi\|_1
\le\sqrt{\Tr(\Pi\Omega)\Tr(\bar\Pi\Omega)}
\le\sqrt{\ell_{\cR}(\rho)}$, proving the result.
\end{proof}

Let $p^{\mathrm{orig}}_{0:L-1}$ be the length-$L$ joint visible input--output
distribution of the original agent under the stationary memoryless reference
$\cR$, initialised in $\rho_\star^{(\cR)}$, and let
$\widehat p^{\mathrm{rec}}_{0:L-1}$ be the corresponding distribution of the
reset-completed agent initialised in $\bar\rho$. The same bound holds for
output-only histories by marginalising the stimulus record.

\begin{proposition}[Finite-horizon bound for reset-completed repair]
\label{sm:prop:finite_horizon_reset}
For every $L\ge1$,
\begin{equation}
\dTV\!\left(
 p^{\mathrm{orig}}_{0:L-1},
 \widehat p^{\mathrm{rec}}_{0:L-1}
\right)
\le
\eps_M+(L-1)\bigl(\sqrt{\gamma_{\cR}}+\gamma_{\cR}\bigr).
\label{sm:eq:finite_horizon_general}
\end{equation}
\end{proposition}

\begin{proof}
Since $P$ is a spectral projector of $\rho_\star^{(\cR)}$,
\[
\rho_\star^{(\cR)}
=(1-\eps_M)\bar\rho+\eps_M\tau_Q,
\qquad
\tau_Q=Q\rho_\star^{(\cR)}Q/\eps_M,
\]
with the second term omitted when $\eps_M=0$. Replacing the initial state by
$\bar\rho$ therefore changes any visible history distribution by at most
$\eps_M$ in total variation.

For $k=0,\ldots,L$, let $p^{(k)}$ use the reset-completed instrument for the
first $k$ steps and the original instrument thereafter. For
$1\le k\le L-1$, the common repaired prefix ensures that the conditional
memory entering the replaced step is supported on $P\cH_M$. Applying
Lemma~\ref{sm:lem:one_step_reset_error} to each classical prefix and using
contractivity under the common suffix gives
\[
\dTV\bigl(p^{(k-1)},p^{(k)}\bigr)
\le
\sqrt{\gamma_{\cR}}+\gamma_{\cR}.
\]
The final replacement contributes zero after tracing the terminal memory,
because Proposition~\ref{sm:prop:reset_valid} gives identical one-step visible
action probabilities for every retained input state. Summing the $L-1$
nonzero hybrid differences and adding the initial $\eps_M$ term proves
Eq.~\eqref{sm:eq:finite_horizon_general}.
\end{proof}

\subsection{Stimulus-wise polar completion used in Main Text}
\label{sm:subsec:polar_completion}

The repair used for the main-text certificate normalises the projected Kraus family
separately for every stimulus. This makes validity independent of the reference process used
to choose the retained subspace.

For each \(x\), define
\[
\bar K^{(x)}_{y,\eta}:=PK^{(x)}_{y,\eta}P,
\qquad
G_x:=\sum_{y,\eta}\bar K^{(x)\dagger}_{y,\eta}\bar K^{(x)}_{y,\eta}.
\]
Stacking the projected operators as
\begin{equation}
A_x=\sum_{y,\eta}\ket{y,\eta}\otimes\bar K^{(x)}_{y,\eta}
\qquad\text{gives}\qquad A_x^\dagger A_x=G_x.
\label{sm:eq:stacked_polar}
\end{equation}
Equivalently,
\[
G_x
=
P-
\varsigma_x,
\qquad
\varsigma_x:=
\sum_{y,\eta}P K^{(x)\dagger}_{y,\eta}QK^{(x)}_{y,\eta}P
\succeq0.
\]
The operator \(\varsigma_x\) is the stimulus-resolved leakage operator. On the support of \(G_x\),
define
\[
\widetilde K^{(x)}_{y,\eta}
=
\bar K^{(x)}_{y,\eta}G_x^{-1/2},
\]
where \(G_x^{-1/2}\) denotes the Moore--Penrose pseudoinverse square root. Then
\[
\sum_{y,\eta}
\widetilde K^{(x)\dagger}_{y,\eta}
\widetilde K^{(x)}_{y,\eta}
=
\Pi_x,
\]
where \(\Pi_x\) projects onto \(\mathrm{supp}(G_x)\). Define
\[
\delta_x:=\|P-G_x\|_\infty,
\qquad
\delta:=\max_x\delta_x.
\]
If \(\delta<1\), then each \(G_x\) is full rank on \(P\cH_M\), so
\(\Pi_x=P\) and the polar repair defines a trace-preserving instrument on
all of the retained memory space.

Moreover, \(\widetilde A_x=A_xG_x^{-1/2}\) is the polar factor of \(A_x\).
When \(G_x\) is full rank on \(P\cH_M\), it is the unique minimiser over
isometries on the retained input space:
\begin{equation}
\min_{U^\dagger U=P}\|A_x-U\|_{\mathrm F}^2
=
\|A_x-\widetilde A_x\|_{\mathrm F}^2
=
\tilde{d}_q+\Tr G_x-2\Tr\sqrt{G_x}.
\label{sm:eq:polar_nearest_isometry}
\end{equation}
Since \(0\le G_x\le P\),
\begin{equation}
\|A_x-\widetilde A_x\|_{\mathrm F}
\le
\sqrt{\tilde{d}_q}\,
\bigl(1-\sqrt{1-\delta_x}\bigr).
\label{sm:eq:polar_leakage_distance}
\end{equation}
Thus the repair is the nearest isometric completion of the projected update
for each stimulus separately.

When some \(G_x\) are rank-deficient, one may complete the instrument on the kernel of
\(G_x\) by adding arbitrary Kraus operators whose squared sum is \(P-\Pi_x\). This kernel
completion restores trace preservation on all of \(P\cH_M\), but the action and memory update
on the inaccessible kernel are not fixed by the projected dynamics.

Any additional hidden labels introduced by such a kernel completion are
included in the common comparison alphabet by zero-padding the corresponding
operators of the original agent. For every retained dimension used in the
figures, the computed \(G_x\) are full rank, so no arbitrary kernel completion
enters the reported rates.

For a stationary finite-memory reference, let \(\widetilde L_\omega\) be the
routed Kraus operators obtained from Eq.~\eqref{sm:eq:routed_kraus} and the polar-completed
\(\widetilde K^{(x)}_{y,\eta}\). Proposition~\ref{sm:prop:repaired_qfdr} gives the exact
purified-history rate
\begin{equation}
R_{F,\mathrm{pol}}^{(Q)}
=-\tfrac12\log_2\spr\!\left(
Z\mapsto\sum_\omega\widetilde L_\omega Z L_\omega^\dagger
\right).
\label{sm:eq:polar_qfdr}
\end{equation}

The two completions solve complementary problems. The polar completion is the primary
construction because it is the nearest isometry to the projected update and is the agent whose input--output rate and QFDR certificate are used in the Main Text. The reset completion instead preserves one-step
action statistics exactly and remains available for any projector, at the price of explicit
recovery dynamics after leakage events.

\subsection{Physical finite-horizon bound with a local encoder}
\label{sm:sec:physical-new}
The local encoder retains the projected component and returns the discarded
component to a chosen state. Including both branches gives a trace-preserving
evolution whose external-output error satisfies the bound in Main Text
Theorem~\mainref{thm:physical-accuracy}.
Write $P=JJ^\dagger$, where $J$ embeds the $d$-dimensional retained memory,
$Q=\Id-P$, and choose a retained density operator $\tau_P$.
The encoder $\mathcal C_P(X)=PXP+\Tr(QX)\tau_P$ has a keep Kraus operator $P$
and recovery operators $\sqrt{t_a}|a\rangle\langle q_j|$, where
$\tau_P=\sum_at_a|a\rangle\langle a|$ and $\{|q_j\rangle\}$ spans $Q$.
Its trace-preserving property implies
\begin{equation}
\Tr_M[(\operatorname{id}_C\otimes\mathcal C_P)(\Omega_{CM})]
=\Tr_M\Omega_{CM}.
\label{sm:eq:encoder-env}
\end{equation}
Thus preparing the smaller memory leaves the environment's marginal unchanged.
A normalised joint projection would instead condition that marginal on the
projection outcome.

Purify the initial joint state and every operation of the common causal
environment. The environment acts through the agent's interface, not directly
on its internal memory. Implement the reset completion by following each
original update with the encoder, and apply the same encoder initially.
There are $L$ relevant insertions: the initial encoding and $L-1$ intermediate
encodings. The final encoding on inaccessible terminal memory cannot affect
external outputs and can be omitted.

To define all insertions on a common Hilbert space, extend the original
isometries to unitaries on fresh pure ancillas and pull each memory projector
back through the original prefix. The resulting interaction-picture keep branch
is a product of orthogonal projectors. If $|\psi\rangle$ is the original purification in that picture and
$|\phi\rangle$ its unnormalised keep-only counterpart, put
$s=\langle\phi|\phi\rangle$ and
$e=\sum_{t=0}^{L-1}\Tr(Q\rho_{M,t})$. Each projector's failure probability in the
unmodified experiment is the corresponding summand. The success and disturbance bounds in Gao's
Theorem~1~\cite{Gao2015Union}, applied to the pure-state dilation, give
\begin{equation}
1-s\leq4e,\qquad
1-|\langle\psi|\phi/\sqrt{s}\rangle|^2\leq e
\label{sm:eq:gao-inputs}
\end{equation}
when the keep branch is nonzero. For $e<1/4$ it is necessarily nonzero.
Gao's trace-norm convention is the full norm; here $D$ denotes half the norm.

After tracing the orthogonal recovery flags, the complete reset output is
$|\phi\rangle\langle\phi|+R$, with $R\succeq0$ and $\Tr R=1-s$.
For $t=1-s$, the trace norm of the difference between a normalised and a
subnormalised pure-state projector, followed by the triangle inequality for $R$,
yields
\begin{equation}
D\leq\tfrac12\!\left[t+\sqrt{t^2+4(1-t)e}\right].
\end{equation}
This expression is increasing in $t$ for $0\leq e\leq1$. For $e<1/4$, substituting
$t\leq4e$ gives $D\leq\sqrt e+2e$. For $e\geq1/4$, clipping at one is valid.
Tracing the remaining hidden registers proves
\begin{equation}
D_{L,\mathrm{reset}}\leq\min\{1,\sqrt{E_L}+2E_L\}.
\end{equation}
The reduced experiment begins with the local encoding of the original joint
state. The argument includes correlated driving and feedback, allowing this
state to have coherences between the retained and discarded memory subspaces. In the iid stationary spectral case, choosing
$\tau_P=P\rho_\star P/\Tr(P\rho_\star)$ recovers the usual normalised retained
initial memory. For a coherent interface, the completion here is the original
specified dilation followed by this memory-local encoder.

\subsection{Tracking the branch without recovery}
\label{sm:subsec:conditional-leakage}
The occupation estimate discards information about how successive projections
are related. Keeping the probability of each first recovery gives a refinement
of the same physical comparison; the underlying first-failure argument follows
the sequential-measurement reasoning of Ref.~\cite{ODonnell2021Union}.

\begin{proposition}[Conditional-leakage refinement of reset completion]
\label{sm:prop:conditional-leakage}
In the setting of Main Text Theorem~\mainref{thm:physical-accuracy}, put
$m=L-1$ and $\eps_t=\Tr[Q\rho_{M,t}]$ for $t=0,\ldots,m$.
Let $\Sigma_t$ be the unnormalised state of the branch with no recovery through
insertion $t$, propagated through the same joint dynamics and projected onto
the retained agent memory at each insertion. Write
$s_t=\Tr\Sigma_t$, so $s_0=1-\eps_0$, and define
\begin{equation}
 b_m=\max\!\left\{0,\,1-\eps_0-
       \sum_{t=1}^{m}\sqrt{\eps_t(s_{t-1}-s_t)}\right\}.
 \label{sm:eq:conditional-overlap}
\end{equation}
Then the complete reset output obeys
\begin{equation}
 D_{L,\mathrm{reset}}\leq
 \min\!\left\{1,\,
 \frac{1-s_m+\sqrt{(1+s_m)^2-4b_m^2}}{2}\right\}.
 \label{sm:eq:conditional-reset-bound}
\end{equation}
\end{proposition}

\begin{proof}
Purify the original experiment and use the unitary extensions described in
Sec.~\ref{sm:sec:physical-new}. Pull the retained projectors back to the common
initial Hilbert space, writing them as $\Pi_t$ with $Q_t=\Id-\Pi_t$.
For the original normalised state $|\psi\rangle$, let
$|\phi_0\rangle=\Pi_0|\psi\rangle$ and
$|\phi_t\rangle=\Pi_t|\phi_{t-1}\rangle$.
Then $s_t=\|\phi_t\|^2$ and
$\|Q_t\psi\|^2=\eps_t$. In particular,
\begin{align}
 \langle\psi|\phi_t\rangle
 &=\langle\psi|\phi_{t-1}\rangle
   -\langle Q_t\psi|Q_t\phi_{t-1}\rangle,\notag\\
 |\langle Q_t\psi|Q_t\phi_{t-1}\rangle|
 &\leq\sqrt{\eps_t(s_{t-1}-s_t)}.
 \label{sm:eq:conditional-first-failure}
\end{align}
Here $s_{t-1}-s_t$ is the probability that the first recovery occurs at $t$.
Starting from $\langle\psi|\phi_0\rangle=1-\eps_0$ and summing gives
$|\langle\psi|\phi_m\rangle|\geq b_m$.
After tracing the orthogonal recovery flags, the physical output on the
remaining dilation registers is
$|\phi_m\rangle\langle\phi_m|+R$, where $R\succeq0$ and
$\Tr R=1-s_m$. The exact pure/subnormalised-pure trace norm is
\begin{equation}
 \big\||\psi\rangle\langle\psi|-|\phi_m\rangle\langle\phi_m|\big\|_1
 =\sqrt{(1+s_m)^2-4|\langle\psi|\phi_m\rangle|^2}.
\end{equation}
The triangle inequality for $R$ and contractivity under the remaining readout
give Eq.~\eqref{sm:eq:conditional-reset-bound}. Common final isometries do not
change these quantities. If the branch vanishes, the displayed bound is one.
\end{proof}

The probabilities $s_t$ follow a trace-decreasing evolution conditioned on
the absence of recovery. For iid inputs, retained coordinates give
\begin{align}
 \Sigma_0&=J^\dagger\rho_{M,0}J,\notag\\
 \Sigma_t&=\sum_{x,y,\eta}p_{\cR}(x)
 A^{(x)}_{y,\eta}\Sigma_{t-1}A^{(x)\dagger}_{y,\eta},\qquad
 A^{(x)}_{y,\eta}=J^\dagger K^{(x)}_{y,\eta}J.
 \label{sm:eq:conditional-keep-map}
\end{align}
A reference or tester with memory requires the corresponding joint evolution,
with projections acting only on the agent. Where $s_{t-1}>0$, the conditional
failure probability is $\ell_t=(s_{t-1}-s_t)/s_{t-1}$; the first-failure weight
in Eq.~\eqref{sm:eq:conditional-overlap} is $s_{t-1}\ell_t$, not $\ell_t$ alone.

For the spectral clock at $N=256$, $h=1/2$, the conditional reset bound remains
nontrivial through 113, 226 and 450 steps for $d=2,4,8$, respectively. At
$L=256$ it is $0.905$ for $d=8$. These values use the conditional recurrence with the local initial encoder.
The benchmark figures use direct history calculations with stationary
preparations. The corresponding bound for polar completion includes the
additional repair term derived below. The transition-leakage estimate of
Sec.~\ref{sm:subsec:finite_horizon_reset} is also available for its iid
preparation; where their hypotheses agree, the smallest of these bounds may
be used.

\subsection{Comparison with polar completion}
\label{sm:sec:polar-comparison-new}
Let $A$ be the stacked projected map for one stimulus and write
$A=UB$, $B=\sqrt G$, with $U$ a completed polar isometry. Set
$\delta=\|\Id_d-G\|_\infty$ and $b=\sqrt{1-\delta}$, so
$b\Id_d\preceq B\preceq\Id_d$. Any physical recovery completion of the keep map
has missing effect $\Id_d-G$. On a pure input, possibly entangled with a spectator,
put $g=\langle B^2\rangle$ and $h=\langle B\rangle$.
The half trace distance between the polar output and the completed keep output
is bounded by
\begin{equation}
D\leq\tfrac12\sqrt{(1+g)^2-4h^2}+\tfrac12(1-g).
\end{equation}
The scalar inequality $B^2\preceq(1+b)B-b\Id_d$ gives
$h\geq(g+b)/(1+b)$ and hence $h^2\geq b^2(g+1-b^2)$ for
$b^2\leq g\leq1$. The displayed bound is consequently at most $1-b^2=\delta$.
For singular $G$, $\delta=1$ gives the trivial bound, provided the chosen polar
completion is an isometry with a common visible readout. Convexity and purification
extend the statement to arbitrary inputs, so it is a half-diamond bound.

Replacing one round at a time in a common causal experiment and using
contractivity gives an additional $L\max_x\delta_x$. For independently drawn
classical stimuli it can be replaced by $L\sum_xp_{\cR}(x)\delta_x$.
Combining this with Main Text Theorem~\mainref{thm:physical-accuracy} gives the
polar-repair extension
\begin{equation}
D_{L,\mathrm{polar}}\leq
\min\{1,\sqrt{E_L}+2E_L+L\max_x\delta_x\}.
\label{sm:eq:physical-polar-bound}
\end{equation}
Discarded occupation measures the weight outside the retained space in the
original experiment. The parameter $\delta_x$ measures the largest probability
of leaving that space on the next update from a retained input state.

\subsection{An achievable physical dimension}
\label{sm:sec:dimension-new}
Let $\rho_{M,\mathrm{av}}=L^{-1}\sum_{t=0}^{L-1}\rho_{M,t}$. Then
$E_L=L\Tr[(\Id-P)\rho_{M,\mathrm{av}}]$, so Ky Fan's principle shows that its
leading rank-$d$ eigenspace minimises the reset-completion bound. It need not
minimise the polar leakage term or the actual behavioural error.
Define the smoothed rank
\begin{equation}
d_\zeta(\rho)=\min\left\{d:\sum_{j>d}\lambda_j(\rho)\leq\zeta\right\}.
\label{sm:eq:smoothed-rank}
\end{equation}
For $0<\eta\leq1$, $\eps\leq\eta^2/(4L)$ implies
$\sqrt{L\eps}+2L\eps\leq\eta/2+\eta^2/2\leq\eta$.
Thus $d_{\eta^2/(4L)}(\rho_{M,\mathrm{av}})$ is a sufficient physical dimension
for the reset-completed agent under the specified experiment and encoder.

For $0<\alpha<1$, write
$S_\alpha(\rho)=\log_2\Tr(\rho^\alpha)/(1-\alpha)$.
Retaining eigenvalues at least
$t=(\zeta/\Tr\rho^\alpha)^{1/(1-\alpha)}$ discards at most $\zeta$ and uses no
more than $\Tr(\rho^\alpha)/t^\alpha$ eigenvalues. Consequently,
\begin{equation}
d_\zeta(\rho)\leq
\min\!\left\{\dim\mathsf M,
\left\lceil2^{S_\alpha(\rho)}\zeta^{-\alpha/(1-\alpha)}\right\rceil\right\}.
\label{sm:eq:renyi-rank}
\end{equation}
This supplies an achievable memory dimension for a physical agent at a chosen
horizon and error tolerance. It bounds a particular construction, rather than
the optimal dimension--distortion curve. The analogous von Neumann estimate
$\eps_d\leq S(\rho)/\log_2(d+1)$ can be substantially looser. The dependence
on $L$ remains part of the statement; an asymptotic rate requires a separate
limit argument.

\section{Optional variational refinement at fixed bond dimension}
\label{sm:sec:variational}

This Section records an optional variational refinement at fixed bond
dimension. No quantum variational curve is used in Fig.~\mainref{fig:compression_results}; the
plotted quantum certificate is the QFDR of the polar-completed agent. A future
refinement could minimise the same mixed-transfer rate, initialised at that
valid agent.

\begin{proposition}[Variational objective: maximise fidelity density]
\label{sm:prop:variational_objective}
Fix a stationary reference input process \(\cR\). Let \(|\Psi_{\cR}\rangle\) denote the
stationary history iMPS dilation of the original agent (Sec.~\ref{sm:sec:history_imps}), and let \(\widetilde W\) be a
candidate compressed controlled instrument inducing a stationary history iMPS
\(|\widetilde\Psi_{\cR}\rangle\) with memory dimension \(\tilde{d}_q\).

Define the \emph{fidelity density} between the two stationary dilations by
\begin{equation}
f(\widetilde W)
:=
\lim_{L\to\infty}
|\langle \widetilde\Psi_{\cR,L} \mid \Psi_{\cR,L}\rangle|^{1/L},
\end{equation}
whenever the limit exists. Under the normal/injective stationary-bulk assumptions used in
Sec.~\ref{sm:sec:certificates}, this quantity equals the spectral radius \(\mu_0\) of the mixed transfer operator
between the two iMPS tensors; see Eq.~\eqref{sm:eq:qfdr_mu_exact}. Maximising
\(f(\widetilde W)\) is therefore equivalent to minimising the QFDR:
\begin{equation}
\RFQ(\Psi_{\cR},\widetilde\Psi_{\cR})
=
-\frac12\log_2 f(\widetilde W).
\end{equation}

A variational refinement at fixed bond dimension \(\tilde{d}_q\) is therefore defined by the
constrained optimisation problem
\begin{equation}
\max_{\widetilde W}\;\; f(\widetilde W)
\quad \text{subject to:}\quad
\widetilde W \text{ defines a valid stimulus-conditioned instrument for each } x.
\label{sm:eq:var_objective}
\end{equation}
In practice one optimises a local parametrisation of \(\widetilde W\) within the isometric
manifold and uses the dominant left/right eigenvectors of the mixed transfer operator to
evaluate gradients, in direct analogy with uniform-MPS tangent-space optimisation for
fidelity density~\cite{Vanhecke2021Truncation,yang2024dimensionreductionquantumsampling}.
\end{proposition}

\section{Coherent stimulus rail and stationary comb/process-tensor viewpoint}
\label{sm:sec:comb}

This Section records the coherent formulation underlying the routed construction. The routed
isometry \(W\) defines a valid causal multi-time object before any reference input process is
chosen. Choosing a diagonal stationary reference input process closes the transfer map on memory and gives
the history iMPS studied in Secs.~\ref{sm:sec:history_imps}--\ref{sm:sec:certificates}. Leaving the stimulus rail open instead gives a
stationary comb/process tensor.

\subsection{Comb generated by the routed isometry}

Recall the one-step routed isometry
\begin{equation}
W
=
\sum_{x\in X} \ket{x}\!\bra{x}_X \otimes V_x,
\qquad
V_x=\sum_{y,\eta}K^{(x)}_{y,\eta}\otimes \ket{y}_Y\otimes\ket{\eta}_E,
\label{sm:eq:router_recall}
\end{equation}
acting as
\[
W:\cH_M\otimes\cH_X\to \cH_M\otimes\cH_Y\otimes\cH_X\otimes\cH_E.
\]
The key point is that \(W\) is linear in the stimulus state on \(\mathsf{X}\): it defines a well-posed
quantum interaction for arbitrary, possibly coherent, stimulus inputs.

For a horizon \(L\), concatenate identical copies of \(W\) in time, wiring the output memory
of one step into the next and leaving the forwarded copy of each stimulus register as an
explicit output system. Tracing out the environment registers \(\mathsf{E}_{0:L-1}\), and optionally
leaving the action registers \(\mathsf{Y}_{0:L-1}\) unmeasured, yields a completely positive causal
map from the stimulus sequence to the corresponding output sequence with memory.
Equivalently, via the Choi--Jamio{\l}kowski isomorphism, this defines a positive operator
\(\Upsilon^{(L)}\) on the tensor product of the stimulus and output spaces that satisfies the
standard comb causality constraints. We call \(\Upsilon^{(L)}\) the length-\(L\) process
tensor generated by the agent.

\begin{proposition}[Coherent rail yields a stationary comb; classical reference processes are restricted testers]
\label{sm:prop:comb_from_router}
Let \(W\) be the routed isometry in Eq.~\eqref{sm:eq:router_recall}. Then:
\begin{enumerate}
  \item[(i)] \emph{Comb/process tensor.} For each \(L\), iterating \(W\) in time and tracing
  the environment registers defines a valid length-\(L\) quantum comb/process tensor
  \(\Upsilon^{(L)}\) from the stimulus sequence to the output sequence, satisfying the usual
  causality constraints.

  \item[(ii)] \emph{Diagonal reference processes as testers.} A stationary diagonal reference input process
  \(\cR\), including memoryless driving and finitely correlated diagonal driving,
  corresponds to a restricted tester family on the stimulus rail. Such testers destroy
  coherence in the stimulus basis and prepare subsequent stimulus states according to a
  classical rule. Contracting \(\Upsilon^{(L)}\) with such a tester produces the classical
  input--output process studied in the Main Text.

  \item[(iii)] \emph{Recovery of the history iMPS.} For memoryless \(\cR\) with
  distribution \(p_{\cR}(x)\), contracting each copy of the routed update with the
  diagonal reference input process, or equivalently using the coherent-label construction of Sec.~\ref{sm:sec:history_imps}, yields
  the stationary history iMPS \(|\Psi_{\cR}\rangle\). Its transfer map is the
  reference-averaged memory channel of Sec.~\ref{sm:sec:canonical}.
\end{enumerate}
\end{proposition}

\begin{proof}[Proof sketch]
(i) This is the standard Stinespring-to-comb construction: Concatenating the same local
isometry in time produces a causal network; tracing inaccessible environments yields a CP
map from interleaved inputs to outputs; its Choi operator obeys the comb trace constraints
because earlier outputs cannot depend on later inputs.

(ii) Diagonal classical testers form a strict subset of all quantum testers, because they
measure or dephase the stimulus rail in the computational basis at each time. Finitely-correlated diagonal driving adds a finite classical controller to this tester, but still does
not preserve coherence between distinct stimulus symbols.

(iii) In the memoryless case the diagonal tester factorises across time. The coherent-label
construction of Sec.~\ref{sm:sec:history_imps} gives one copy of the same history
isometry at every step, so the
resulting network is a uniform history iMPS\@. Contracting one time step with its conjugate
gives the reference-averaged memory channel by Lemma~\ref{sm:lem:transfer_equals_phi}.
\end{proof}

The branching-tree picture of Sec.~\ref{sm:sec:trees} recovers the response to
each classical input trajectory. Keeping the stimulus rail coherent instead
specifies how superpositions of those inputs interfere. The next subsection
closes this same rail with a quantum tester and compares what that tester
can observe. A single specified driving process gives a stationary transfer
calculation, whilst uniform retention supplies a guarantee independent of
the driving strategy.

\subsection{Quantum testers}
\label{sm:sec:quantum-testers}

The controlled isometry $W$ of Eq.~\eqref{sm:eq:router_recall} specifies the
agent, including its action on coherent stimuli and the environment
$\mathsf E$ that remains inaccessible. This specification matters: replacing
the environment basis separately for different stimuli preserves their
classical instruments, but can change the coherences returned to a quantum
tester. The reduced agent is obtained from the same labelled dilation.

\begin{definition}[Stationary quantum tester]
A stationary quantum tester has memory $\mathsf C$, stimulus register
$\mathsf X$, accessible record $\mathsf F$, and a fixed isometry
\begin{equation}
T:\mathsf C\otimes\mathsf X\otimes\mathsf Y
\longrightarrow\mathsf C\otimes\mathsf X\otimes\mathsf F.
\end{equation}
Place time cuts after $T$. The tester's inter-round memory
$\mathcal C=\mathsf C\otimes\mathsf X$ then includes the next stimulus.
One round is $U_{\mathcal T}$ of Main Text Eq.~\mainref{eq:quantum-tester};
it applies the agent to $\mathsf X\otimes\mathsf M$ and then the tester to
the returned $\mathsf X\otimes\mathsf Y$ and $\mathsf C$.
Stationary here means the same update is repeated; a stationary joint state
is selected separately when constructing the infinite history.
\end{definition}

The classical feedback rule has an isometric implementation in this
convention. It records the current action, updates the environment, and
prepares the next stimulus:
\begin{equation}
T|c,x,y\rangle=
\sum_{c',x'}\sqrt{E(c'|c,x,y)q(x'|c')}
\,|c',x'\rangle_{\mathcal C}|c,x,y,c'\rangle_{\mathsf F}.
\label{sm:eq:classical-tester-isometry}
\end{equation}
The record makes different incoming basis states orthogonal, and
normalisation gives $T^\dagger T=\Id$. Ignoring $y$ in the environment update
recovers an autonomous classical reference. The ordinary input--output
readout discards $c,c'$ from this enlarged record; these reference labels
need not be accessible in the original experiment.

A fresh coherent source is another useful special case. It prepares
$|\xi\rangle=\sum_x\sqrt{p(x)}e^{i\theta_x}|x\rangle$ at every round and
records the returned stimulus. Matching original and reduced tensors carry
the same phases, which cancel in the mixed transfer. Its certificate is
therefore the one for the iid reference $p(x)$. Reusing the stimulus or
entangling inputs across rounds instead requires the joint tester memory.

\begin{proposition}[Routed history under a quantum tester]
\label{sm:prop:qtester-routing}
Expanding $U_{\mathcal T}$ in an orthonormal basis of
$\mathsf F\otimes\mathsf E$ gives operators $L_\omega$ on
$\mathcal C\otimes\mathsf M$ with
\begin{equation}
\sum_\omega L_\omega^\dagger L_\omega=\Id,
\qquad
\Phi_{\mathcal T}(\Omega)
=\Tr_{\mathsf{FE}}[U_{\mathcal T}\Omega U_{\mathcal T}^\dagger].
\end{equation}
They define a left-canonical uniform history MPS. If $\Phi_{\mathcal T}$ is
primitive, its unique stationary joint state is the canonical bond density
operator, with eigenvalues equal to the squared Schmidt coefficients across
a time cut.
\end{proposition}
\begin{proof}
Both factors in $U_{\mathcal T}$ are isometries, so
$U_{\mathcal T}^\dagger U_{\mathcal T}=\Id$. Expansion in the orthonormal
emitted labels gives completeness. The stationary identification is then
Proposition~\ref{sm:prop:canonical_bond_state}, applied to the joint channel.
\end{proof}

Truncation uses the full Hermitian marginal
$\rho_M=\Tr_{\mathcal C}\Omega_{\mathcal C M}$ and acts as
$\Id_{\mathcal C}\otimes P_M$. The marginal can have complex off-diagonal
entries: taking only its real part would select a different compression.

\begin{proposition}[Validity under quantum interventions]
\label{sm:prop:qtester-validity}
Let $\widetilde V_x$ be the stimulus-wise polar completion, including its
kernel completion when needed. Then
\begin{equation}
\widetilde W^\dagger\widetilde W
=\sum_x|x\rangle\langle x|\otimes
\widetilde V_x^\dagger\widetilde V_x
=\Id_{\mathsf X}\otimes P_M.
\end{equation}
Thus the reduced agent defines a valid round for every causal tester acting
through its interface, including adaptive and nonstationary testers with
quantum memory.
\end{proposition}
\begin{proof}
The returned stimulus states are orthogonal. Their cross terms vanish in
$\widetilde W^\dagger\widetilde W$, leaving the completeness identity for
each repaired input. Composition with any tester isometry preserves this
identity on the retained memory.
\end{proof}

\subsubsection{Certificates and accessible output}
\label{sm:subsec:tester-certificates}

\begin{proposition}[Certification under a specified tester]
\label{sm:prop:qtester-certificate}
Fix common emitted labels and purifications of the original and reduced
joint initial states. On the joint bond, propagate
\begin{equation}
Z_{t+1}=\sum_\omega\widetilde L_\omega Z_tL_\omega^\dagger.
\label{sm:eq:tester-cross-transfer}
\end{equation}
The root fidelity of the accessible record $\mathsf F_{0:L-1}$ alone is at
least $\|Z_L\|_1$. If the tester also retains its terminal memory
$\mathcal C$, its full output states $\rho_{\mathcal C F,L}$ and
$\widetilde\rho_{\mathcal C F,L}$ instead obey
\begin{equation}
F(\rho_{\mathcal C F,L},\widetilde\rho_{\mathcal C F,L})
\geq \big\|\Tr_{\mathcal C}Z_L\big\|_1.
\label{sm:eq:tester-memory-certificate}
\end{equation}
For stationary histories, the mixed-transfer rate formula of
Lemma~\ref{sm:lem:mixed_transfer_exact} applies when the chosen boundary
contraction reaches its leading sector without cancellation. Each
accessible readout then inherits the corresponding data-processing bound.
\end{proposition}
\begin{proof}
Purify the joint initial states using a common register and retain the
emitted environments. Tracing these common labels gives
Eq.~\eqref{sm:eq:tester-cross-transfer}. If both terminal bond memories are
purifying systems, Uhlmann's trace-norm identity gives $\|Z_L\|_1$ for the
remaining purified record. Discarding its inaccessible labels and initial
purifier cannot decrease fidelity. To include $\mathcal C$ in the output,
only the agents' terminal memories remain as purifying systems, so the cross
operator is $\Tr_{\mathcal C}Z_L$. Apply the same identity and readout.
The asymptotic argument is the reached-sector calculation already used in
Sec.~\ref{sm:sec:certificates}.
\end{proof}

The terminal system cannot be ignored simply because its dimension is fixed.
For orthogonal tester states $|+\rangle$ and $|-\rangle$ with no emitted
record, $Z=|-\rangle\langle+|$ has norm one but trace zero. The empty records
are identical, whilst the terminal states are perfectly distinguishable.
A final partial trace can also remove the leading sector of a mixed
transfer. For example, a qubit tester applies a recorded Pauli channel with
weights $(0.4625,0.2375,0.2375,0.0625)$ for $(I,X,Y,Z)$ between calls to two
agents that apply $I$ and $Z$, respectively. From the common initial state
$\sqrt{0.8}|0\rangle+\sqrt{0.2}|1\rangle$, the record-only rate is $-\tfrac12\log_2(0.4)$, whereas the
rate including terminal memory is $-\tfrac14\log_2(0.05)$. The eigenvalues
$\pm0.4$ occur in the traceless sector. Equality of these rates therefore
requires the final contraction to preserve the same leading exponential
contribution; finite terminal dimension alone does not supply that condition.

Theorem~\mainref{thm:physical-accuracy} also applies to this interface with
its stated local encoding of the initial memory. Its proof retains all
registers accessible to the common causal environment and does not assume
that the stimuli are classical.

\subsubsection{Accuracy against every tester}

\begin{theorem}[Strategy accuracy under uniform retention]
\label{sm:thm:strategy-fidelity}
Suppose the specified dilation satisfies Eq.~\eqref{sm:eq:full-retention}.
Fix the two initial agent memories independently of the tester, and choose
purifications with embedded overlap $g_0$. For any causal $L$-round quantum
tester $S$, let $\rho_S$ and $\widetilde\rho_S$ be the final states of all
its accessible registers; the agents' memories and $\mathsf E$ remain
inaccessible. Then
\begin{equation}
F(\rho_S,\widetilde\rho_S)\geq w^{L/2}|g_0|,
\qquad
\tfrac12\|\rho_S-\widetilde\rho_S\|_1
\leq\sqrt{1-w^L|g_0|^2}.
\label{sm:eq:all-tester-retention}
\end{equation}
Taking the infimum of the fidelity, or the supremum of the half trace
distance, over testers gives the corresponding strategy bounds
\cite{Gutoski2018StrategyFidelity,Gutoski2012StrategyNorm,GutoskiWatrous2007}.
\end{theorem}
\begin{proof}
Purify the tester and fix the two initial purifications. Its interventions
are common isometries on accessible registers and preserve the global
overlap. Let $J$ embed the reduced memory in the original one. The full-space
retention identity gives
\begin{equation}
\overline{\widetilde W}^{\dagger}W
=\sqrt w\,(\Id_{\mathsf X}\otimes J^\dagger),
\end{equation}
where the bar denotes embedding of the reduced outgoing memory. This is
Eq.~\eqref{sm:eq:retention-intertwining} with the stimulus rail left open.
Every agent round therefore multiplies the embedded overlap by $\sqrt w$,
whatever coherent or adaptive intervention precedes it. After $L$ rounds
the overlap is $w^{L/2}g_0$. Discarding the agents' memories and inaccessible
records gives the fidelity inequality; the root-fidelity/trace-distance
inequality gives the second bound. Neither bound depends on the tester.
\end{proof}

A bound uniform over testers also exists without constant retention. On a
retained input, the cross-isometry overlap for stimulus $x$ is
$\sqrt{G_x}$, so its least expectation is at least
$\sqrt{1-\delta_x}$, where $\delta_x=\|P_M-G_x\|_\infty$. With a coherent
stimulus the cross overlap is block diagonal, giving the one-step half
trace-distance bound $\max_x\sqrt{\delta_x}$ even with an entangled
spectator. In a hybrid comparison, replace the original rounds in temporal
order, so the common prefix up to the next changed round is reduced and its
memory lies in the retained space. Subsequent common operations contract
the distance. This yields
\begin{equation}
D_L\leq\min\{1,D_0+L\max_x\sqrt{\delta_x}\},
\label{sm:eq:generic-tester-bound}
\end{equation}
where $D_0$ compares the specified original and embedded reduced initial
memories. Uniform retention provides the sharper bound above by controlling
the overlap of the entire interaction rather than summing separate
one-step errors.

\subsubsection{A coherent clock and a dephased control}
\label{sm:subsec:quantum-tester-numerics}

A single qubit serves as the tester's inter-round stimulus memory. After
receiving the clock's action, which is written to the record, it applies an
amplitude damping channel with probability $\gamma$, records the damping outcome,
and rotates the qubit by
\begin{equation}
R=\begin{pmatrix}
\cos\theta&-e^{i\phi}\sin\theta\\
\sin\theta&e^{i\phi}\cos\theta
\end{pmatrix}.
\end{equation}
The dephased control performs the same operations and additionally dephases
the stimulus before it is lent to the agent. This defines a different
intervention protocol, whose stationary statistics need not coincide.
The clock's hidden environment is not included in the tester's record.

For $N=32,d=4$, choose $\gamma=0.02$, $\phi=0.7$, and $\theta$ so that the
dephased control has reset-basis occupation $r_N=1-e^{-1/(2N)}$.
The exact reset writes states $|j_n\rangle$ spanning an $N$-dimensional
hidden environment; ordinary evolution writes a fixed state $|u\rangle$
in that environment. Their overlaps specify the cross-stimulus coherences
of the device. Table~\ref{sm:tab:quantum-clock-tester} compares
$|u\rangle=|j_0\rangle$ with an exactly orthogonal flag, obtained by adding
one environmental dimension. The projectors use the full complex Hermitian
agent marginal of each stationary joint state.

\begin{table}[tbp]
\centering\small
\begin{tabular}{llrrr}
\hline
Hidden record & Tester & Reset-basis occupation & Own-design $NR_F^{(Q)}$ & Iid-design $NR_F^{(Q)}$\\
\hline
$u=j_0$ & coherent & 0.00775542 & 0.036532 & 0.040548\\
$u=j_0$ & dephased & 0.01550356 & 0.038294 & 0.042218\\
Orthogonal flag & coherent & 0.01550356 & 0.038294 & 0.042218\\
Orthogonal flag & dephased & 0.01550356 & 0.038294 & 0.042218\\
\hline
\end{tabular}
\caption{A $d=4$ reduced clock under a qubit tester at $N=32$. Own design
uses the tester's stationary agent marginal; iid design uses $h=1/2$.
The occupation is a stimulus-basis population, not a count of measured reset
events. The exact flag uses an $(N+1)$-dimensional hidden environment;
$u=j_0$ uses the original $N$ dimensions. The reported rates are the specified
purified-history certificates.}
\label{sm:tab:quantum-clock-tester}
\end{table}

With $|u\rangle=|j_0\rangle$, the coherent tester has approximately half the
reset-basis occupation of the control. This is a stimulus population, not a
recorded reset frequency; measuring the stimulus before each round gives
the dephased protocol. Designing with the coherent tester improves its
certificate relative to using the iid reference, even though the two tester
design subspaces are close (minimum principal cosine approximately
$0.999762$). In the exact-flag implementation, the hidden environment
identifies the stimulus and the coherent and dephased results agree at the
reported precision. This example shows why compression of coherent
interactions specifies the dilation as well as the classical instruments.

\section{Numerical implementation and benchmark agents}
\label{sm:sec:numerics}

The numerical comparisons follow the input--output record $(x,y)$ generated
by each agent. Hidden Kraus labels and internal environment states are traced
out. Rates use root fidelity and base-two logarithms with the factor $1/(2L)$;
QFDR supplies an upper bound on the statistical FDR of this record. Fidelities
at specified finite horizons are calculated separately. All parameters are
stored in \texttt{parameters.json}, and
Figs.~\mainref{fig:compression_results} and~\mainref{fig:feedback} are generated
from the accompanying result tables.

For the iid clock, the resolution sweep evaluates $N=24,32,\ldots,256$
at $d=4$, and the $N=256$ dimension sweep in Fig.~\mainref{fig:compression_results}(b)
uses every integer $d=1,\ldots,32$,
with $h=1/2$ and $r_N=1-e^{-h/N}$. Earlier evaluated dimensions at $N=24,64$
are retained in the supplementary tables. For the walk, all 30 resolutions
$N=24,32,\ldots,256$ and every rank $d=1,\ldots,\min(40,N)$ are evaluated,
with equally probable stimuli. There are 1176 evaluated walk agents.
The fixed-$d=27$ resolution sweep omits $N=24$, where that dimension is unavailable.
The fixed-dimension resolution sweeps are retained in
Fig.~\ref{sm:fig:resolution-sweeps}. The right-hand main-figure panels instead
hold the $N=256$ instruments fixed at their design reference and vary only the
iid deployment law (Sec.~\ref{sm:subsec:fixed-input-deployment}).
Only physical memory dimension is compared. Fresh randomness is not persistent
memory, but stored hidden-state labels, a retained seed, or an internal timestep
counter would count. The reduced agent receives no unrecorded original age or
position. The quantum benchmark instruments use polar repair without subsequent
behavioural fitting.

\subsection{Quantum resettable renewal clock}

The first benchmark is a nonorthogonal-memory realisation of the resettable
renewal clock of Ref.~\cite{Elliott2021QuantumAgents}.  The classical causal
states are clock ages $n=0,1,\dots,N-1$; the input alphabet is binary, with
$x=0$ ordinary evolution and $x=1$ an external reset, and the output alphabet
is binary, with $y=1$ a tick.  For the uniform waiting-time law used here the
survival probability is $\Phi(n)=(N-n)/N$.  Under $x=0$ and
$0\le n<N-1$, an age state advances with probability $a_n:=p(0,n+1|n,0)=\Phi(n+1)/\Phi(n)
=(N-n-1)/(N-n)$ or emits a tick and resets to age zero with probability
$b_n:=p(1,0|n,0)=1-a_n=1/(N-n)$; at $n=N-1$ a tick occurs with probability one.  Under
$x=1$ the clock maps deterministically to age zero without a tick; $p(0,0|n,1)=1$.

Following the prescription in Ref.~\cite{Elliott2021QuantumAgents}, the quantum memory states are nonorthogonal pure states
$\{\ket{\sigma_n}\}_{n=0}^{N-1}$ with Gram matrix
\begin{equation}
G_{nm}=\braket{\sigma_n|\sigma_m}
=\frac{N-\max(n,m)}{\sqrt{(N-n)(N-m)}}.
\end{equation}
Memory vectors are obtained by diagonalising $G=U\Lambda U^\dagger$ and setting
$\ket{\sigma_n}=\sum_\alpha\sqrt{\Lambda_\alpha}\,U^\ast_{n\alpha}\ket{\alpha}$,
so the coordinate matrix $S$ with columns $\ket{\sigma_n}$ satisfies
$S^\dagger S=G$.  The stimulus-conditioned isometries are
\begin{equation}
V_0\ket{\sigma_n}=\sqrt{a_n}\,\ket{\sigma_{n+1}}\ket{0}_Y\ket{0}_E
+\sqrt{b_n}\,\ket{\sigma_0}\ket{1}_Y\ket{0}_E,
\end{equation}
with the first term omitted at $n=N-1$, and
\begin{equation}
V_1\ket{\sigma_n}=\ket{\sigma_0}\ket{0}_Y\ket{j_n}_E,
\qquad \braket{j_n|j_m}=G_{nm},
\end{equation}
so the environment carries exactly the information needed to preserve inner
products when all visible memory states are reset to $\ket{\sigma_0}$.

The reference reset strategy is scaled with the time resolution.  Writing
$N=\tau/\Delta t$ and $\tau_R=2\tau$ gives $r_N=1-e^{-\Delta t/\tau_R}
=1-e^{-1/(2N)}$, and the stationary age distribution is
\begin{equation}
\pi_n^{(\cR_N)}
=\frac{(1-r_N)^n\,\Phi(n)}{\sum_{m=0}^{N-1}(1-r_N)^m\,\Phi(m)}.
\end{equation}
The driven quantum memory state is
$\rho_q^{(\cR_N)}=\sum_n\pi_n^{(\cR_N)}
\ket{\sigma_n}\!\bra{\sigma_n}$, and the plotted resources are
$C_\mu=H(\pi^{(\cR_N)})$, $C_q=S(\rho_q^{(\cR_N)})$, and
$D_q=\log_2\rank(\rho_q^{(\cR_N)})=\log_2 N$.

\subsection{Adaptive cyclic walk}

The second benchmark is the
input-controlled cyclic walk~\cite{Garner_2017}.  The classical states are positions
$j=0,\dots,N-1$.  A binary stimulus $x\in\{0,1\}$ selects one of two shift
laws $p_x(r)$ on the cyclic displacement $r$; the output is the next absolute
position $y=j+r\bmod N$, occurring with probability $p(y|j,x)=p_x(y-j)$, and
the next state is deterministically $y$.  The positions are $N$ equal bins on a unit circle. In the implementation,
$p_0$ is the bin-integrated wrapped uniform shift law of half-width
$\alpha=0.10$, and $p_1$ is the bin-integrated wrapped Gaussian shift law of
width $\sigma=0.06$.

For each stimulus, define the one-step amplitude state
\begin{equation}
\ket{\phi_j^{(x)}}=\sum_{y}\sqrt{p_x(y-j)}\,\ket{y}_x,
\end{equation}
whose overlaps are translation invariant,
$\braket{\phi_j^{(x)}|\phi_k^{(x)}}=B_x(k-j)$, with
\begin{equation}
B_x(d)=\sum_r\sqrt{p_x(r)\,p_x(r-d)}
\end{equation}
the cyclic Bhattacharyya autocorrelation of $\sqrt{p_x}$.  Because the agent
must be ready for either counterfactual stimulus, the memory state effectively stores both
one-step futures coherently,
\begin{equation}
\ket{\sigma_j}=\ket{\phi_j^{(0)}}\otimes\ket{\phi_j^{(1)}},
\end{equation}
so the Gram matrix is the product
\begin{equation}
G_{jk}=\braket{\sigma_j|\sigma_k}=B_0(k-j)\,B_1(k-j),
\end{equation}
which is positive semidefinite by construction and circulant by translation
invariance.  The stimulus-conditioned update is
\begin{equation}
V_x\ket{\sigma_j}=\sum_y\sqrt{p_x(y-j)}\,
\ket{\sigma_y}_M\ket{y}_Y\ket{\eta_j^{(x)}}_E,
\end{equation}
where the environment carries the unused counterfactual branch,
$\ket{\eta_j^{(0)}}=\ket{\phi_j^{(1)}}$ and
$\ket{\eta_j^{(1)}}=\ket{\phi_j^{(0)}}$.  A direct computation gives
$\braket{V_x\sigma_j|V_x\sigma_k}=B_0(k-j)B_1(k-j)=G_{jk}$, so each $V_x$
preserves inner products on the memory span and extends to a valid isometry;
measuring the output register reproduces $p_q(y|j,x)=p_x(y-j)$ with conditioned
update $\ket{\sigma_j}\mapsto\ket{\sigma_y}$.

For any position-independent i.i.d.\ strategy the stationary distribution over
positions is uniform, $\pi_j=1/N$, so the driven memory state is
$\rho_q=\tfrac1N\sum_j\ket{\sigma_j}\!\bra{\sigma_j}$, whose nonzero
eigenvalues are those of $G/N$.  Writing $a_x(r)=\sqrt{p_x(r)}$ and
$\beta_x(\ell)=|\widehat a_x(\ell)|^2$, the circulant Fourier eigenvalues of
$G(d)=B_0(d)B_1(d)$ are the circular convolution
$\lambda_\ell(G)=\tfrac1N\sum_m\beta_0(m)\beta_1(\ell-m)$.  For the
uniform/Gaussian kernels used here $G$ is full rank over the tested $N$, so
$D_q=\log_2 N$ while $C_q=S(\rho_q)<C_\mu=\log_2 N$: the exact quantum model
has an entropic advantage but no \emph{a priori} memory-dimension reduction.

\subsection{Classical reduction and the regenerative clock baseline}
\label{sm:subsec:classical-baseline}
For iid input probabilities $r_x$, let $\pi$ be the target's stationary classical
memory distribution. A partition into blocks $B_a$ gives a valid reduced
transducer
\begin{equation}
\widetilde T^{y|x}_{ba}
=\frac{\sum_{j\in B_a}\pi_j\sum_{k\in B_b}T^{y|x}_{kj}}
{\widetilde\pi_a},\qquad
\widetilde\pi_a=\sum_{j\in B_a}\pi_j.
\label{sm:eq:classical-merge}
\end{equation}
Blocks with zero stationary weight are omitted. Summing separately for each
$x,a$ gives unity; stationarity of $\pi$ gives stationarity of
$\widetilde\pi$ under the same iid reference. The aggregation therefore
preserves the stationary output law for each stimulus. Longer histories can
still change, since different original states in a block may have different
successors. For the same reason, the reduced update need not remain unifilar.
This stationary identity uses independent driving; a feedback environment
requires solving the reduced joint stationary state separately.

\paragraph*{Aggregation in both benchmarks.}
Figure~\mainref{fig:compression_results} uses Eq.~\eqref{sm:eq:classical-merge}
with the same contiguous blocks $B_a=\{n:\lfloor nd/N\rfloor=a\}$ for both
benchmarks, ordering the clock's states by age and the walk's by position.
For the clock, every tick and external reset maps each age to $n=0$, so the
aggregated clock is regenerative with $\alpha=e_0$,
$C_{ab}=\widetilde\pi_a^{-1}\sum_{n\in B_a,\,n+1\in B_b}\pi_na_n$, and tick
probability $\widetilde\pi_a^{-1}\sum_{n\in B_a}\pi_nb_n$ from block $a$. Its
observable FDR therefore follows from the renewal recurrence of
Sec.~\ref{sm:subsec:clock-renewal-new}. We evaluate all 1176 clock models with
$d\leq\min(40,N)$ at the 30 resolutions (\texttt{numerics/run\_clock\_aggregation.py}).
With one age per block the construction reproduces the target exactly, and
where $d$ divides $N$ it coincides with the age merging of
Fig.~\ref{sm:fig:clock-diagnostics}, whose rounded block boundaries change no
plotted value by more than $1.3\times10^{-4}$. Because aggregation yields a
regenerative clock, the fitted family below contains it: the fitted search adds
optimisation to the same structure.

For the clock, equal-age and nonuniform age partitions provide initial
models. We also use sequential-phase and random substochastic initialisations.
The fitted class has a row-substochastic matrix $C$ and row reset distribution
$\alpha$. In the column-propagation convention its visible transition matrices are
\begin{equation}
T^{0|0}=C^{\mathsf T},\quad
T^{1|0}=\alpha^{\mathsf T}(\boldsymbol1-C\boldsymbol1)^{\mathsf T},\quad
T^{0|1}=\alpha^{\mathsf T}\boldsymbol1^{\mathsf T},\quad T^{1|1}=0.
\end{equation}
This is the full $d$-state discrete phase-type class for the first-tick experiment,
but only a regenerative subclass of all input--output HMMs for full-process FDR.
Both event types reset to the same distribution. No optimum over a larger class
is claimed.

Each row of $[C,\boldsymbol1-C\boldsymbol1]$ and the distribution $\alpha$ are
softmax-parametrised. We minimise observable FDR with analytic derivatives and
L-BFGS-B, using logits in $[-24,24]$ and at most 2500 iterations per new start.
Previously saved candidates are re-evaluated rather than treated as new fits.
The denser grids require 32 further $(N,d)$ cases: 27 additional resolutions at
$d=4$, and $d=7,9,10,11,12$ at $N=256$. Each uses 20 starts, giving 640
additional fits. The case seed is $240927+1000N+d$. Initialisations include the
specified age partitions and stochastic phase models; where available, one
start rescales the nearest saved same-dimension candidate and is then refitted
at the new resolution. No error value is obtained by interpolation.
For each start we retain the stopping status, objective, gradient norm and
iteration count. All new runs meet the optimiser's stopping test. The plotted
entries are the best feasible candidates found, with boundary parameters
recorded explicitly; convergence does not certify a global optimum. The
quantum agents are obtained without fitting to this objective.

\subsection{Observable fidelity of the regenerative clock}
\label{sm:subsec:clock-renewal-new}
Let $a_\ell,\widetilde a_\ell$ be the target and reduced first-tick probabilities
under ordinary inputs, and $s_k,\widetilde s_k$ their survival probabilities.
For the target, $a_\ell=1/N$ on $1\leq\ell\leq N$ and
$s_k=(1-k/N)_+$. For a quantum reduced clock, these arrays are evaluated from its
actual no-tick operator and post-reset state; for the classical baseline they are
given by Eq.~\mainref{eq:classical-clock}. The event blocks are visible: they end
either in an ordinary-input tick or in an external reset. With $q=1-r_N$, their
combined root-fidelity weight is
\begin{equation}
b_\ell=q^\ell\sqrt{a_\ell\widetilde a_\ell}
+r_Nq^{\ell-1}\sqrt{s_{\ell-1}\widetilde s_{\ell-1}}.
\end{equation}
An event-free suffix of length $L$ has weight
$h_L=q^L\sqrt{s_L\widetilde s_L}$. Starting after a reset,
\begin{equation}
F_0=1,\qquad
F_L=h_L+\sum_{\ell=1}^{L}b_\ell F_{L-\ell}.
\label{sm:eq:renewal-fidelity-new}
\end{equation}
This follows by partitioning histories by their first event. After that event,
each model returns to its own fixed state, so the remainder has fidelity
$F_{L-\ell}$. Hidden paths are summed in each event probability before taking
square roots. The recurrence is therefore for observable fidelity, not an overlap
of arbitrarily aligned hidden histories.

Since the target cannot remain event-free for $N$ steps, $b_\ell=0$ for
$\ell>N$. Under the positive-event and reachable-boundary conditions satisfied by
the evaluated models, the positive root $z_\star\geq1$ of
\begin{equation}
\sum_{\ell=1}^{N}b_\ell z_\star^\ell=1
\end{equation}
determines $R_F^{(\mathrm{stat})}=\tfrac12\log_2z_\star$.
The implementation solves a logarithmic root equation and handles the exact
zero-error endpoint within its numerical tolerance.

\paragraph*{Fidelity lost per regenerative episode.}
An episode ends at the next natural tick or external reset. Its root fidelity
is $\beta=\sum_{\ell=1}^{N}b_\ell$, and its overlap-weighted mean duration is
$\bar\ell=\sum_\ell\ell b_\ell/\beta$.
Writing $z_\star=e^{s_\star}$ and taking the probability weights
$b_\ell/\beta$, Jensen's inequality gives
$1=\beta\mathbb E[e^{s_\star\ell}]
\geq\beta e^{s_\star\bar\ell}$.
On the other hand, $\ell\leq\ell_{\max}=N$ on the support of $b_\ell$, so
$1\leq\beta e^{s_\star\ell_{\max}}$. Thus
\begin{equation}
-\frac{\log_2\beta}{2\ell_{\max}}
\leq R_F^{(\mathrm{stat})}
\leq-\frac{\log_2\beta}{2\bar\ell}.
\label{sm:eq:renewal-bracket}
\end{equation}
This statement assumes $0<\beta\leq1$ and the regenerative recurrence and
accessible root specified above. It also applies to other pairs of
regenerative agents when their common episode weights have bounded support.

\begin{table}[tbp]
\centering\small
\begin{tabular}{crrrr}
\hline
$d$ & $1-\beta$ & $\bar\ell/N$ & $NR_F^{(\mathrm{stat})}$ & Upper bracket\\
\hline
2 & 0.039997 & 0.399055 & 0.073140 & 0.073785\\
4 & 0.018689 & 0.412073 & 0.032893 & 0.033026\\
8 & 0.008993 & 0.419781 & 0.015494 & 0.015524\\
\hline
\end{tabular}
\caption{Episode-fidelity interpretation for the $N=256$ repaired clock at
$h=1/2$. The upper bracket in Eq.~\eqref{sm:eq:renewal-bracket} differs from
the calculated observable rate by less than one per cent at these three
points. The episode weights are those of the same renewal calculation.}
\label{sm:tab:clock-episode-bracket}
\end{table}

For the repaired clock the upper bracket closely follows the calculated rate.
It relates the fidelity of a complete episode to its duration through the
same block probabilities used in the exact renewal calculation.

A stationary window has a different first-event law. Let $\pi$ be the target's
stationary age distribution and write $m_n=N-n$. Its first-tick and survival
probabilities before the first event are
\begin{equation}
a_\ell^{\mathrm{st}}=\sum_{n:m_n\geq\ell}\frac{\pi_n}{m_n},\qquad
s_k^{\mathrm{st}}=\sum_n\pi_n\frac{(m_n-k)_+}{m_n}.
\end{equation}
Compute the corresponding reduced probabilities from its own stationary memory.
Replacing the first-block weights by these stationary probabilities, and using
Eq.~\eqref{sm:eq:renewal-fidelity-new} only after the first event, gives the
stationary finite-window fidelity. The accessible leading root is unchanged for
these models, but the finite-window values are not. Neither formula applies to a
nonregenerative HMM without a separate derivation.

The QFDR is evaluated independently by the original/reduced mixed transfer.
Common input/output labels and explicit hidden-label alignment are retained.
The finite trace-norm certificate uses purifications of the two stated initial
memories, not the asymptotic rate substituted into a finite-window expression.

\begin{table}[tbp]
\centering\small
\begin{tabular}{rrrrrr}
\hline
$N$ & $d$ & Classical $NR_F$ & Quantum $NR_F$ & Quantum $NR_F^{(Q)}$ & Quantum first-tick TV\\
\hline
24 & 1 & 0.132855 & 0.150495 & 0.196075 & 0.282783\\
24 & 2 & 0.080582 & 0.063775 & 0.083769 & 0.168443\\
24 & 3 & 0.059985 & 0.038076 & 0.051147 & 0.113497\\
24 & 4 & 0.050139 & 0.026378 & 0.035797 & 0.086714\\
24 & 5 & 0.042997 & 0.019693 & 0.027072 & 0.070051\\
24 & 6 & 0.037848 & 0.015382 & 0.021338 & 0.060137\\
24 & 8 & 0.030694 & 0.010135 & 0.014303 & 0.050274\\
64 & 1 & 0.138282 & 0.160031 & 0.207050 & 0.283632\\
64 & 2 & 0.084037 & 0.070006 & 0.090685 & 0.170579\\
64 & 3 & 0.064184 & 0.042942 & 0.056538 & 0.116320\\
64 & 4 & 0.054330 & 0.030661 & 0.040498 & 0.088752\\
64 & 5 & 0.047414 & 0.023670 & 0.031489 & 0.071119\\
64 & 6 & 0.042786 & 0.019138 & 0.025537 & 0.059484\\
64 & 8 & 0.036061 & 0.013684 & 0.018405 & 0.044681\\
256 & 1 & 0.140840 & 0.164630 & 0.212226 & 0.284189\\
256 & 2 & 0.085570 & 0.073140 & 0.094089 & 0.171995\\
256 & 3 & 0.066240 & 0.045428 & 0.059206 & 0.118442\\
256 & 4 & 0.056106 & 0.032893 & 0.042853 & 0.090996\\
256 & 5 & 0.049608 & 0.025739 & 0.033670 & 0.073678\\
256 & 6 & 0.044800 & 0.021091 & 0.027582 & 0.061989\\
256 & 7 & 0.041272 & 0.017894 & 0.023444 & 0.053588\\
256 & 8 & 0.038331 & 0.015494 & 0.020298 & 0.047191\\
256 & 9 & 0.036072 & 0.013689 & 0.017953 & 0.042270\\
256 & 10 & 0.034085 & 0.012227 & 0.016036 & 0.038258\\
256 & 11 & 0.032346 & 0.011068 & 0.014527 & 0.035024\\
256 & 12 & 0.030912 & 0.010084 & 0.013239 & 0.032259\\
\hline
\end{tabular}

\caption{Selected clock comparison values. All rates are per physical lifetime at
$h=1/2$. Classical entries are best found in the stated regenerative class;
quantum entries are direct spectral/polar constructions, including the raw
rank-one model. First-tick TV is a separate reset-conditioned experiment.
Displayed digits are rounding of floating-point results, not rigorous error bars.}
\label{sm:tab:clock-full}
\end{table}

\subsection{Memory at matched error and fixed-agent deployment}
\label{sm:subsec:clock-deployment}

Table~\ref{sm:tab:clock-matched-error} reads the $N=256,h=1/2$ curves at
matched error. The quantum points are the direct spectral/polar agents.
For each point, we select the first recorded classical dimension whose
input--output FDR is no larger. The scan includes every integer dimension
through twelve. This reports the dimension attained by that search, not a
lower bound over all classical models. Comparing against the QFDR upper
bound rather than the directly evaluated quantum error gives a distinct,
more conservative comparison.

\begin{table}[tbp]
\centering\small
\begin{tabular}{crccrcc}
\hline
 & \multicolumn{3}{c}{Input--output FDR} & \multicolumn{3}{c}{QFDR certificate}\\
Quantum $d$ & $NR_F^{(\mathrm{stat})}$ & Aggregation & Fitted & $NR_F^{(Q)}$ & Aggregation & Fitted\\
\hline
2 & 0.073140 & 6 & 3 & 0.094089 & 4 & 2\\
3 & 0.045428 & 16 & 6 & 0.059206 & 9 & 4\\
4 & 0.032893 & 30 & 11 & 0.042853 & 18 & 7\\
5 & 0.025739 & $>40$ & $>12$ & 0.033670 & 29 & 11\\
\hline
\end{tabular}
\caption{Memory at matched error for $N=256,h=1/2$. Each classical column gives
the first dimension whose input--output FDR is no larger than the quantum value
to its left, for stationary aggregation with blocks $\lfloor nd/N\rfloor$ and for
the saved regenerative search. $>40$ and $>12$ mean that no model in the
evaluated range attains it; they are not classical lower bounds. The
certificate and observable-error columns concern the same quantum agents.}
\label{sm:tab:clock-matched-error}
\end{table}

For deployment, the $N=256,d=4$ quantum projector and repaired maps are
constructed once under $h=1/2$, and the fitted classical $C,\alpha$ are also
held fixed. At deployed intensity $h$, only the iid reset probability
$r_N(h)=1-e^{-h/N}$ changes. Each resulting process uses its own stationary
preparation. For comparison, a redesigned quantum agent repeats the
spectral selection and repair under $r_N(h)$. No additional fitting is
performed for any fixed agent.

\begin{table}[tbp]
\centering\small
\begin{tabular}{crrrr}
\hline
Deployed $h$ & \multicolumn{2}{c}{Quantum, fixed} & Quantum, redesigned & Classical, fixed\\
 & $NR_F^{(\mathrm{stat})}$ & $NR_F^{(Q)}$ & $NR_F^{(\mathrm{stat})}$ & $NR_F^{(\mathrm{stat})}$\\
\hline
$1/8$ & 0.042248 & 0.044951 & 0.038775 & 0.070047\\
$1/4$ & 0.038884 & 0.044141 & 0.036688 & 0.065038\\
$1/2$ & 0.032893 & 0.042853 & 0.032893 & 0.056106\\
$1$ & 0.023422 & 0.041462 & 0.026619 & 0.041918\\
$2$ & 0.011722 & 0.042336 & 0.017973 & 0.024036\\
$4$ & 0.003035 & 0.052510 & 0.008155 & 0.009545\\
$8$ & 0.000657 & 0.084080 & 0.000982 & 0.003645\\
\hline
\end{tabular}
\caption{Deployment of the $N=256,d=4$ clock. The fixed quantum and classical
agents retain their design at $h=1/2$; only the input distribution changes.
Redesign repeats spectral truncation and repair at each deployed intensity.
All rates are per physical lifetime. The classical candidate is the saved
regenerative fit used in the additional clock comparison, not a newly fitted model at
each intensity. These are seven sampled operating points.}
\label{sm:tab:clock-off-design}
\end{table}

The input--output rates in Table~\ref{sm:tab:clock-off-design} use the same
renewal calculation as the main comparison; QFDR is evaluated separately
from the frozen labelled Kraus operators. All seven tested intensities
favour the fixed quantum candidate over the fixed classical one. At the
sampled $h\geq1$, redesign gives a larger observable rate than retaining the
original projector. At high reset intensities the fixed QFDR becomes more
conservative even as the visible rate falls. The widening gap reflects the additional distinctions retained in the
certificate's inaccessible records.

\subsection{Finite-time timing and history diagnostics}
\label{sm:subsec:clock-timing-diagnostic}
A separate timing experiment resets the clock once and then supplies ordinary
inputs. The exact first-tick law is uniform on $1,\ldots,N$; for an approximate law
$\widetilde a_\ell$ we report
\begin{equation}
\dTV^{\mathrm{tick}}=1-\sum_{\ell=1}^N
\min\{1/N,\widetilde a_\ell\}.
\label{eq:clock-timing-error}
\end{equation}
Late ticks and no tick within the lifetime count as error; the approximation is
not renormalised onto the target interval. At $N=256$, the direct $d=8$ quantum
clock has error $\ClockEightTV{}$, whilst using three qubits rather than the eight
needed for a 256-dimensional register. This task test is distinct from iid-driven
history fidelity. Stationary one-lifetime and two-lifetime fidelities are reported below: a lower asymptotic rate does not imply better accuracy at every
finite horizon. In particular, the $d=2,4$ classical fits are slightly better over
one stationary lifetime, whereas the quantum models are better over two.

\begin{figure}[tbp]
\centering
\includegraphics[width=\textwidth]{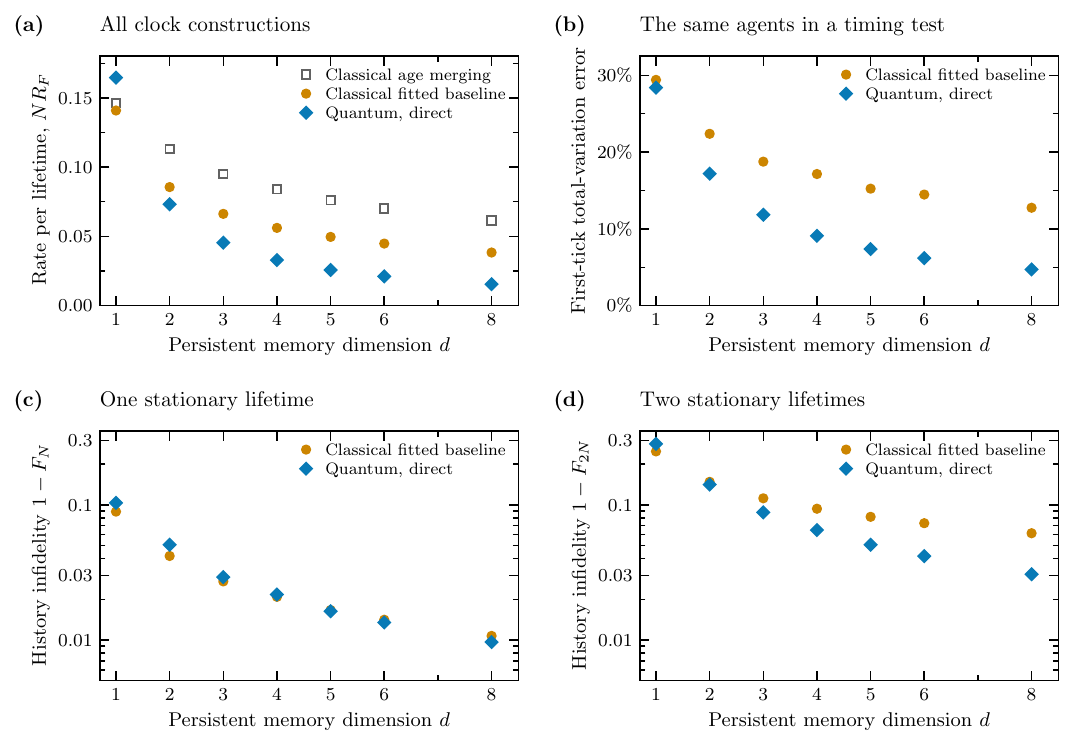}
\caption{\textbf{Clock constructions and finite-time diagnostics at $N=256$.}
\textbf{(a)}~Observable FDR per lifetime for direct classical age merging,
fitted classical regenerative models, and the direct quantum construction.
The raw rank-one quantum point is retained here; the fitted memoryless response
is also available to a one-dimensional quantum agent.
\textbf{(b)}~First-tick TV for the same agents after one external reset followed
by ordinary inputs. Classical models were fitted to FDR, not separately to this
timing objective; late and no-tick mass is included.
\textbf{(c,d)}~Observable stationary-history infidelity $1-F_L$ over one and two
lifetimes, respectively. Both axes use the same error scale, and lower is
better. These panels retain the finite-horizon crossover: superiority in the
asymptotic rate need not imply superiority at every finite horizon.
Each marker is an evaluated dimension; no points are joined.}
\label{sm:fig:clock-diagnostics}
\end{figure}

\begin{table}[tbp]
\centering\small
\begin{tabular}{rrrrrr}
\hline
$d$ & Classical $F_N$ & Quantum $F_N$ & Certificate $F_N$ & Classical $F_{2N}$ & Quantum $F_{2N}$\\
\hline
1 & 0.911004 & 0.896847 & 0.683013 & 0.751782 & 0.717839\\
2 & 0.958194 & 0.949328 & 0.831508 & 0.852164 & 0.858940\\
4 & 0.979158 & 0.978316 & 0.915944 & 0.906482 & 0.935000\\
8 & 0.989287 & 0.990359 & 0.958516 & 0.938413 & 0.969388\\
\hline
\end{tabular}

\caption{Stationary visible root fidelities at $N=256$ and the specified
finite-history quantum certificate. Higher is better. Each model starts in its
own stationary state under iid driving. A rate advantage does not imply a
pointwise advantage at every finite horizon.}
\label{sm:tab:clock-windows}
\end{table}

\subsection{State-independent retention and the cyclic walk}
\label{sm:sec:walk-retention-new}
A particularly simple repair is possible when the outgoing memory has the same
probability of lying in the retained space for every input state and stimulus.
In that case every step contributes the same factor to the comparison of
histories. Suppose a retained projector $P=JJ^\dagger$ satisfies, on the full
original memory,
\begin{equation}
\sum_{y,\eta}K_{y,\eta}^{(x)\dagger}P K_{y,\eta}^{(x)}=w\Id_M
\quad\text{for every }x,\qquad 0<w\leq1.
\label{sm:eq:full-retention}
\end{equation}
Then $G_x=w\Id_d$ on retained coordinates and the polar isometry is
$\widetilde V_x=(J^\dagger\otimes\Id)V_xJ/\sqrt w$.
Writing $\overline V_x=(J\otimes\Id)\widetilde V_x$ gives
\begin{equation}
\overline V_x^\dagger V_x
=\frac{J^\dagger V_x^\dagger(P\otimes\Id)V_x}{\sqrt w}
=\sqrt w\,J^\dagger.
\label{sm:eq:retention-intertwining}
\end{equation}
For the specified coherent stimulus rail, the same identity holds with the
identity on the input register. Purify a common causal tester and the initial
states. Its interventions preserve their global overlap, whilst each agent step
multiplies the embedded overlap by $\sqrt w$. If the initial overlap is $g_0$,
after $L$ rounds it is $w^{L/2}g_0$. Partial trace and readout therefore give
\begin{equation}
F_L\geq w^{L/2}|g_0|,\qquad
D_L\leq\sqrt{1-w^L|g_0|^2},\qquad
R_F\leq-\tfrac14\log_2w
\label{sm:eq:retention-error}
\end{equation}
for nonzero $g_0$ and an existing rate. The tester may have memory and coherent
feedback, but cannot directly access the protected agent memory. The guarantee
concerns this specified coherent extension; a classical instrument family alone
does not define a unique one. Observing only $G_x=w\Id_d$ on retained inputs
would not imply Eq.~\eqref{sm:eq:retention-intertwining}.

For the walk, $K_{y,\eta}^{(x)}=|\sigma_y\rangle\langle v_{y,\eta}^{(x)}|$.
A translation-invariant spectral projector gives
$\langle\sigma_y|P|\sigma_y\rangle=w$ independently of $y$; completeness gives
Eq.~\eqref{sm:eq:full-retention}. Since
$\rho_M=N^{-1}\sum_y|\sigma_y\rangle\langle\sigma_y|$, this $w$ equals
$\Tr(P\rho_M)=1-\eps_M$. Choosing Fourier eigenvectors within degenerate
stationary sectors preserves translation invariance. For an initially independent
reference and the canonical purification of the normalised projected memory,
$g_0=\sqrt w$; other preparations require their actual initial overlap.

\subsection{Spectral certification under uniform retention}
\label{sm:subsec:spectral-certification}
Call a rank-$d$ projector \emph{admissible} if it selects $d$ leading
eigenvectors of $\rho_M$ and satisfies the full-space retention identity
in Eq.~\eqref{sm:eq:full-retention}.
Stationarity then gives $w=\Tr(P_d\rho_M)=1-\eps_d$, where
$\eps_d=\Tr[(\Id-P_d)\rho_M]$.

\begin{corollary}[Spectral certification of a physical reduced agent]
\label{sm:cor:structured-spectral}
For an admissible projector and the stationary-history and boundary conventions
of Lemma~\ref{sm:lem:mixed_transfer_exact},
\begin{equation}
 R_F^{(\mathrm{stat})}\leq\RFQ
 \leq-\tfrac14\log_2(1-\eps_d).
 \label{sm:eq:structured-spectral-bound}
\end{equation}
Thus $\eps_d\leq\zeta_{\delta_\star}:=1-2^{-4\delta_\star}$ is sufficient to
certify rate at most $\delta_\star>0$. If an admissible projector exists at every
rank, the smallest rank selected by this spectral bound is
$d_{\zeta_{\delta_\star}}(\rho_M)$, as defined in Eq.~\eqref{sm:eq:smoothed-rank}.
\end{corollary}
\begin{proof}
Stationarity and Eq.~\eqref{sm:eq:full-retention} identify $w$ with the retained
weight. Apply Eq.~\eqref{sm:eq:retention-error} and data processing. The
threshold follows by solving $-\tfrac14\log_2(1-\eps_d)\leq\delta_\star$.
Leading projectors retain the greatest weight at their rank, which gives the
smoothed-rank statement when those projectors are admissible.
\end{proof}

The corollary applies to the physical repaired instrument. If its bound is not
attained, direct evaluation of QFDR may certify a smaller rank. For the walk's
specified product encoding and leading Fourier projectors, however, equality
holds under iid driving, as the following calculation shows.

\paragraph*{Exact certificate for the cyclic walk.}
All indices in the following calculation are taken modulo $N$. Write
$a_x(r)=\sqrt{p_x(r)}$ and
$B_x(k)=\sum_r a_x(r)a_x(r-k)$, as in the benchmark construction. We use the
Fourier convention
\begin{equation}
 \widehat f(t)=\sum_{k=0}^{N-1}f(k)e^{-2\pi i tk/N},\qquad
 f(k)=\frac1N\sum_{t=0}^{N-1}\widehat f(t)e^{2\pi i tk/N}.
 \label{sm:eq:walk-fourier-convention}
\end{equation}
Put $A_x(t)=\widehat a_x(t)$. Since $a_x$ is real,
$A_x(-t)=\overline{A_x(t)}$ and
$\widehat B_x(t)=|A_x(t)|^2$. The stationary eigenvalues, indexed here by
Fourier mode rather than decreasing size, are
\begin{equation}
 \lambda_\ell
 =\frac{\widehat{B_0B_1}(\ell)}{N}
 =\frac{1}{N^2}\sum_t|A_x(t)|^2|A_{1-x}(\ell-t)|^2.
 \label{sm:eq:walk-stationary-fourier-spectrum}
\end{equation}
Let $I$ contain $d$ modes with the largest eigenvalues and set
\begin{equation}
 w=\sum_{\ell\in I}\lambda_\ell,\qquad
 h(k)=\frac1N\sum_{\ell\in I}e^{2\pi i\ell k/N}.
 \label{sm:eq:walk-projection-kernel}
\end{equation}
The corresponding projector acts on the memory orbit as
$P_d|\sigma_j\rangle=\sum_kh(k)|\sigma_{j+k}\rangle$.
Being diagonal in translation modes, it commutes with translations even when
$I$ contains only one member of a degenerate conjugate pair. Thus there is an
admissible projector at every rank in the memory support; requiring a real
basis with complete sine/cosine pairs would be an additional restriction.

The Kraus operators are $K^{(x)}_{y,\eta}=|\sigma_y\rangle\langle v^{(x)}_{y,\eta}|$.
After embedding the reduced memory in the original space, the mixed transfer
has range in the span of
$X_y=P_d|\sigma_y\rangle\langle\sigma_y|$.
Factor it as $\mathcal E=\mathcal A\mathcal B$, where
$\mathcal A z=\sum_yz_yX_y$ and
$[\mathcal B(X)]_y=w^{-1/2}\sum_{x,\eta}r_x\langle v^{(x)}_{y,\eta}|X|v^{(x)}_{y,\eta}\rangle$.
Its nonzero spectrum equals that of the $N\times N$ coefficient map
$\mathcal B\mathcal A$.
For a single stimulus this map is circulant, with first column
\begin{equation}
 m_x(r)=\frac{a_x(r)}{\sqrt w}
       \sum_k h(k)B_{1-x}(k)a_x(r-k).
 \label{sm:eq:walk-coefficient-column}
\end{equation}
To see this directly, insert $P_d|\sigma_j\rangle$ into the overlap defining the
coefficient of $X_y$. The original output amplitudes supply
$a_x(y-j)a_x(y-j-k)$, and the unused counterfactual memory supplies the
environmental overlap $B_{1-x}(k)$. Polar repair supplies $1/\sqrt w$.
The dependence on $y-j$ gives Eq.~\eqref{sm:eq:walk-coefficient-column}.
For iid input probabilities $r_x$, the full coefficient map is the convex
combination of these two circulant maps.

Define
\begin{equation}
 D_x(t)=\widehat{hB_{1-x}}(t)
       =\frac1N\sum_{\ell\in I}|A_{1-x}(t-\ell)|^2\geq0.
 \label{sm:eq:walk-fourier-weights}
\end{equation}
The Fourier eigenvalue of the single-stimulus coefficient map in sector $s$ is
\begin{equation}
 \nu_{x,s}=\widehat m_x(s)
 =\frac{1}{N\sqrt w}\sum_t A_x(s-t)A_x(t)D_x(t).
 \label{sm:eq:walk-coefficient-eigenvalue}
\end{equation}
The two nonnegative quadratic sums needed for Cauchy--Schwarz are
\begin{align}
 \sum_t|A_x(t)|^2D_x(t)&=Nw,\notag\\
 \sum_t|A_x(s-t)|^2D_x(t)
 &=N\sum_{\ell\in I}\lambda_{s-\ell}.
 \label{sm:eq:walk-fourier-quadratic-sums}
\end{align}
Both follow by substituting Eq.~\eqref{sm:eq:walk-fourier-weights} into
Eq.~\eqref{sm:eq:walk-stationary-fourier-spectrum}; the first also uses the
symmetry $|A_{1-x}(-t)|^2=|A_{1-x}(t)|^2$. Consequently,
\begin{equation}
 |\nu_{x,s}|
 \leq\frac{\sqrt{Nw}\sqrt{N\sum_{\ell\in I}\lambda_{s-\ell}}}{N\sqrt w}
 =\sqrt{\sum_{\ell\in I}\lambda_{s-\ell}}
 \leq\sqrt w.
 \label{sm:eq:walk-fourier-upper-bound}
\end{equation}
The last inequality uses the maximal weight of the $d$ leading modes: the
shifted index set has the same cardinality as $I$. At $s=0$, the factors in
Eq.~\eqref{sm:eq:walk-coefficient-eigenvalue} are complex conjugates, so
$\nu_{x,0}=\sqrt w$. Averaging over any iid stimulus probabilities preserves
both this eigenvalue and the upper bound on every sector. Hence
$\spr(\mathcal E)=\sqrt w$.

Under the reached-boundary condition of Lemma~\ref{sm:lem:mixed_transfer_exact},
\begin{equation}
 \RFQ=-\tfrac14\log_2(1-\eps_d),\qquad
 d_{\mathrm{cert}}(\delta_\star)
 =d_{\,1-2^{-4\delta_\star}}(\rho_M).
 \label{sm:eq:walk-certified-rank}
\end{equation}
This is an equality for the specified product-encoding dilation, leading
Fourier projectors and iid routed certificate. The guarantee against a common
causal tester in Eq.~\eqref{sm:eq:retention-error} does not require iid inputs,
but that general bound does not assert saturation for every tester or readout.
Likewise, Eq.~\eqref{sm:eq:walk-certified-rank} concerns certification within
this spectral construction, not the best possible reduced quantum agent or the
smallest rank meeting the observable target.

\paragraph*{Spectrum, entropy and the evaluated ranks.}
The exact criterion uses the spectral tail. A single entropy can give a
sufficient rank through Eq.~\eqref{sm:eq:renyi-rank}, but discards information
about the shape of that tail. For the $N=256$ walk at $\delta_\star=10^{-2}$,
minimising that bound over $0<\alpha<1$ gives an unrounded rank of about
$239.20$, hence a sufficient dimension of 240, whereas the exact spectral
criterion first passes at 27. Its optimal order is about $\alpha=0.119$.
The proximity of $\log_2d$ to $C_q$ in the main figure is therefore an empirical
property of these spectra, not a consequence of this worst-case entropy
estimate.

\begin{table}[b]
\centering\small
\begin{tabular}{lccccc}
\hline
Criterion & Target & Resolutions & $S_{1/2}$ & $S_1=C_q$ & $S_2$\\
\hline
Input--output FDR & $2\times10^{-2}$ & 30 & 0.98 & 0.22 & 0.22\\
Input--output FDR & $10^{-2}$ & 30 & 0.60 & 0.32 & 0.59\\
Input--output FDR & $5\times10^{-3}$ & 30 & 0.22 & 0.80 & 1.07\\
QFDR certificate & $2\times10^{-2}$ & 30 & 0.58 & 0.34 & 0.60\\
QFDR certificate & $10^{-2}$ & 30 & 0.22 & 0.80 & 1.07\\
QFDR certificate & $5\times10^{-3}$ & 20 & 0.27 & 0.94 & 1.17\\
\hline
\end{tabular}
\caption{Spread, in bits, of $\log_2d-S_\alpha(\rho_M)$ across the resolutions
whose first passing rank lies in $d\leq\min(40,N)$. The first five rows use all
30 resolutions $N=24,32,\ldots,256$. The last uses
$N\in\{24,32,\ldots,152,168,176,184\}$, a set of 20 resolutions; $N=160$ does
not meet this target in the stored scan. For comparison, $D_q-C_q$ has a spread
of $2.84$ bits over all 30 resolutions. Each row concerns its stated family,
criterion and target; these spreads do not establish a dimension formula for
other spectra.}
\label{sm:tab:renyi-tracking}
\end{table}

The table compares the spread of $\log_2d-S_\alpha(\rho_M)$ across the computed
resolutions. Smaller spread means a more nearly constant offset within that
row, not a smaller prediction error or a universal dimension formula.
The strictest certified target uses only the resolutions reached within the
stored rank scan. The exact subset is recorded with the data; spectra outside
that subset are not assumed to fail the target at larger ranks.

\subsection{Observable fidelity of the reduced walk}
\label{sm:subsec:walk-visible-fidelity}
After each visible output $y$, the repaired memory is
$|\widetilde\sigma_y\rangle=P|\sigma_y\rangle/\sqrt w$.
Thus its visible dynamics are a controlled Markov chain with transitions
\begin{equation}
\widetilde p_x(y|j)=\sum_\eta
|\langle v_{y,\eta}^{(x)}|\widetilde\sigma_j\rangle|^2.
\label{sm:eq:walk-visible}
\end{equation}
For iid input probabilities $r_x$, define
$B_{yj}=\sum_xr_x\sqrt{p_x(y|j)\widetilde p_x(y|j)}$.
If each process starts in its own stationary distribution, its first-output
fidelity vector is
$z_y=\sum_xr_x\sqrt{p_x^{(1)}(y)\widetilde p_x^{(1)}(y)}$, where
$p_x^{(1)}(y)=\sum_j\pi_jp_x(y|j)$ denotes the output distribution conditional
on input~$x$, and similarly for the reduced model. Then
\begin{equation}
F_L^{(\mathrm{stat})}=\boldsymbol1^\mathsf T B^{L-1}z\quad(L\geq1).
\end{equation}
The first-step probabilities are conditioned on $x$; $r_x$ appears only once.
For the translation-covariant construction both stationary position distributions
are uniform. Its visible rate reduces to
\begin{equation}
R_F^{(\mathrm{stat})}
=-\tfrac12\log_2\!\left(\sum_xr_x\sum_k
\sqrt{p_x(k)\widetilde p_x(k)}\right).
\end{equation}
The implementation also evaluates the Perron root of $B$ as a check. Together,
these calculations distinguish the guarantee obtained from uniform retention
from the actual error in the walk's visible transition law.

\subsection{Classical stationary aggregation of the walk}
\label{sm:subsec:walk-classical-aggregation}
The reduced classical walk stores a label group of positions. We fix the partition before
calculating its error, with
\begin{equation}
g(j)=\left\lfloor\frac{jd}{N}\right\rfloor,\qquad
B_a=\{j:g(j)=a\},\qquad n_a=|B_a|,
\quad j=0,\ldots,N-1.
\label{sm:eq:walk-classical-partition}
\end{equation}
Its $d$ contiguous groups have sizes $\lfloor N/d\rfloor$ or
$\lceil N/d\rceil$. The same groups are used for both inputs. Since the
original position distribution is uniform, the stationary aggregation of
Eq.~\eqref{sm:eq:classical-merge} has response probabilities and transitions
\begin{equation}
q_x(y|a)=\frac1{n_a}\sum_{j\in B_a}p_x(y-j),\qquad
\widetilde T^{y|x}_{ba}=q_x(y|a)\,\delta_{b,g(y)}.
\label{sm:eq:walk-classical-aggregation}
\end{equation}
The displacement $y-j$ is evaluated modulo $N$. The reduced device emits
one of the original $N$ positions and retains only its group label; it has
$d$ internal states. Its responses are fixed by aggregation, with no
subsequent fitting or choice of partition based on the error.

Summing over $y,b$ gives unity separately for each $x,a$. The stationary
memory distribution is $\widetilde\pi_a=n_a/N$, and for every stimulus
\begin{equation}
\sum_a\widetilde\pi_a q_x(y|a)
=\frac1N\sum_j p_x(y-j)=\frac1N.
\label{sm:eq:walk-classical-marginal}
\end{equation}
Thus aggregation preserves the stationary output distribution as well as the
interface. Positions in the same block share their response probabilities,
so the transition statistics change even though the one-step marginal is
retained. The fixed block boundaries do not impose translation covariance on
every transition.

The next memory label is determined by $y$, so this reduced transducer is
unifilar. Its visible transition law is $\widetilde p_x(y|j)=q_x(y|g(j))$.
For the shared iid reference $r_x=1/2$, define the overlap matrix
\begin{equation}
B_{yj}=\sum_x r_x\sqrt{p_x(y-j)q_x(y|g(j))}.
\label{sm:eq:walk-classical-overlap}
\end{equation}
Both processes have a uniform stationary position marginal. Summing the square roots of corresponding
word probabilities gives, for $L\geq1$,
\begin{equation}
F_L^{(\mathrm{stat})}
=\boldsymbol1^{\mathsf T}B^{L-1}\frac{\boldsymbol1}{N},\qquad
R_F^{(\mathrm{stat})}=-\tfrac12\log_2\operatorname{spr}(B).
\label{sm:eq:walk-classical-fidelity}
\end{equation}
In particular $F_1^{(\mathrm{stat})}=1$. The nonnegative matrix is irreducible
and has a positive diagonal for the evaluated shift laws, so the stationary
boundary vectors reach its Perron sector. This expression for the visible histories agrees with the full
paired-unifilar transfer calculation.

We evaluate every integer $d=1,\ldots,\min(40,N)$ at all 30 resolutions
$N=24,32,\ldots,256$, giving 1,176 reduced classical agents. The $N=256$
dimension panel and the 29-point $d=27$ supplementary resolution sweep are subsets of this
scan. The first rank meeting $10^{-2}$ is selected from the evaluated integers
without assuming monotonicity or interpolating a crossing. At $N=256$, rank 31
gives $\WalkClassThirtyOne{}$ and rank 32 gives $\WalkClassThirtyTwo{}$.
The selected rank varies with resolution because both the shift discretisation
and the integer block sizes change.

\begin{table}[tbp]
\centering\small
\begin{tabular}{rrrr}
\hline
$d$ & Classical aggregation & Quantum input--output FDR & QFDR upper bound\\
\hline
1 & 0.499464 & 0.499464 & 0.645556\\
4 & 0.145021 & 0.104976 & 0.195717\\
8 & 0.060618 & 0.027431 & 0.055858\\
16 & 0.024336 & 0.009786 & 0.020419\\
27 & 0.012446 & 0.004902 & 0.009994\\
31 & 0.010379 & 0.004334 & 0.008324\\
32 & 0.009890 & 0.004113 & 0.007990\\
40 & 0.007372 & 0.002877 & 0.005696\\
\hline
\end{tabular}

\caption{Classical stationary aggregation and direct quantum reduction for
the cyclic walk at $N=256$. All rates are per step under the same iid input
reference. The two input--output columns use the same observable FDR. The
QFDR column is the upper bound for the quantum agent. At $d=1$ the two direct
constructions give the same uniform response, which is not the aggregated memoryless response used in the clock panel.}
\label{sm:tab:walk-classical-comparison}
\end{table}

The classical and quantum calculations use the same bin-averaged shift laws. We verify stochasticity for each input, the stationary
memory and output distributions, and the no-compression endpoint $d=N$ at
every resolution. Direct small-history sums and the independent full
paired-unifilar transfer check Eq.~\eqref{sm:eq:walk-classical-fidelity}.
Seven larger pair-operator checks agree to within $3\times10^{-15}$ in rate.
At nine selected resolutions, doubling the source and shift grids to 512 and
8193 points preserves the first passing classical rank in every case.
The retained data and residuals are numerical checks at stated tolerances,
not interval enclosures. Exploratory behavioural refinement is described
separately below; the aggregation curves compare the specified direct
reductions.

\subsection{Fixed agents under varied iid inputs}
\label{sm:subsec:fixed-input-deployment}
For each benchmark, Fig.~\mainref{fig:compression_results}(c,f) uses a quantum
agent and classical aggregation constructed once at the marked reference.
Their retained subspaces, repaired operators, partitions and aggregated
transition probabilities remain fixed throughout the iid input sweep. We
evaluate each model in its stationary state under the deployed input law and
compare the resulting input--output histories using the FDR of the dimension
sweeps.

\paragraph*{Clock.}
We construct the $N=256,d=4$ agents at $h_0=1/2$. The quantum projector is
selected from the original stationary memory at this reference, and the
classical age averages in Eq.~\eqref{sm:eq:classical-merge} use the same
reference. At deployment intensity $h$, the input probabilities are
$(1-r_N(h),r_N(h))$, with $r_N(h)=1-e^{-h/N}$. We evaluate 54 intensities:
$h=0,1/64,1/32,1/16,3/32$ and $h=2^{k/8}$ for integer $-24\leq k\leq24$.
The horizontal scale is linear between zero and $1/8$, and logarithmic
thereafter. At $h=0$ there are no external resets, but natural ticks still
regenerate both clocks.

The waiting-time and survival laws of each fixed agent are calculated once.
The deployment weights then enter the renewal equation of
Sec.~\ref{sm:subsec:clock-renewal-new}. QFDR is evaluated independently by
changing the input weights in the mixed transfer of the fixed original and
repaired Kraus operators. Table~\ref{sm:tab:fixed-clock-aggregation} gives
selected points. The quantum input--output rate is lower than the frozen
aggregation rate at every sampled intensity. The seven-point fitted
classical comparison in Table~\ref{sm:tab:clock-off-design} is retained
separately; neither classical model is refitted for deployment.

\begin{table}[tbp]
\centering\small
\begin{tabular}{rrrr}
\hline
$h$ & Quantum $NR_F^{(\mathrm{stat})}$ & Quantum $NR_F^{(Q)}$ & Aggregation $NR_F^{(\mathrm{stat})}$\\
\hline
$0$ & 0.045881 & 0.045881 & 0.117077 \\
$1/8$ & 0.042248 & 0.044951 & 0.107754 \\
$1/4$ & 0.038884 & 0.044141 & 0.099157 \\
$1/2$ & 0.032893 & 0.042853 & 0.083934 \\
$1$ & 0.023422 & 0.041462 & 0.060093 \\
$2$ & 0.011722 & 0.042336 & 0.030986 \\
$4$ & 0.003035 & 0.052510 & 0.009475 \\
$8$ & 0.000657 & 0.084080 & 0.003195 \\
\hline
\end{tabular}
\caption{Selected deployment intensities for the fixed $N=256,d=4$ clock.
Both agents retain their $h_0=1/2$ design. The classical column is frozen
stationary aggregation, distinct from the fitted model in
Table~\ref{sm:tab:clock-off-design}. All rates are per physical lifetime.}
\label{sm:tab:fixed-clock-aggregation}
\end{table}

\paragraph*{Walk.}
For $N=256,d=27$, the design reference assigns probability $1/2$ to each
stimulus. We fix the leading Fourier projector, its two polar-repaired
updates and the classical aggregation, and evaluate $p=k/50$ for
$k=0,\ldots,50$, where $p$ is the probability of the Gaussian input.
Thus the endpoints supply only the top-hat input or only the Gaussian input.
The shift laws retain the production source-bin averaging and six-standard-deviation Gaussian cutoff; their widths are unchanged.

The original position distribution is uniform under either input, so its
stationary quantum memory and the selected spectral subspace are independent
of $p$. The fixed classical aggregation also retains its stationary block
weights under either input. We substitute $r=(1-p,p)$ into the visible
transition calculations of Secs.~\ref{sm:subsec:walk-visible-fidelity}
and~\ref{sm:subsec:walk-classical-aggregation}. The classical fidelity is
obtained from its nonnegative overlap transfer, and the quantum fidelity is
also checked with the explicit output-conditioned transitions.
Equation~\eqref{sm:eq:walk-certified-rank} gives the same QFDR at every
mixture. The directly evaluated error changes, reaching its largest value
for top-hat-only driving in this sweep, and remains below classical
aggregation at all 51 points (Table~\ref{sm:tab:fixed-walk-mixture}).
The common certificate is below the displayed $10^{-2}$ target throughout.

\begin{table}[tbp]
\centering\small
\begin{tabular}{rrrr}
\hline
$p$ & Quantum $R_F^{(\mathrm{stat})}$ & Quantum $R_F^{(Q)}$ & Aggregation $R_F^{(\mathrm{stat})}$\\
\hline
0 & 0.009653 & 0.009994 & 0.022174 \\
0.25 & 0.007273 & 0.009994 & 0.017294 \\
0.5 & 0.004902 & 0.009994 & 0.012446 \\
0.75 & 0.002538 & 0.009994 & 0.007631 \\
1 & 0.000181 & 0.009994 & 0.002848 \\
\hline
\end{tabular}
\caption{The same $N=256,d=27$ walk under different iid input mixtures.
The agents are constructed at $p=1/2$ and retain their update maps.
$p=0$ supplies only the top-hat stimulus and $p=1$ only the Gaussian stimulus.
All rates are per step.}
\label{sm:tab:fixed-walk-mixture}
\end{table}

These panels concern iid deployment families. The broader guarantee for the
structured walk follows from Eq.~\eqref{sm:eq:retention-error}, which also
allows a causal tester to choose inputs adaptively. The clock's sampled
rates are not used as an accuracy bound for arbitrary input strategies.

\paragraph*{Checks.}
The saved operator and transition arrays have identical hashes throughout
each sweep. We check normalisation for every stimulus, the deployment fixed
points, the renewal roots and data processing. Independent production
clock coordinates reproduce the new rates at selected intensities, including
the zero-reset limit; dense and matrix-free mixed-transfer calculations
agree. At smaller clock sizes, explicit input--output histories agree with
the stationary renewal recurrence. Explicit production walk instruments
reproduce the new transition probabilities, observable rates and certificates
at both endpoints and intermediate mixtures. Doubling the walk quadrature
grids preserves the ordering at the selected checks. These are numerical
consistency tests at the tolerances recorded with the code.

\subsection{Fixed register size as resolution changes}
\label{sm:subsec:resolution-sweeps}
Figure~\ref{sm:fig:resolution-sweeps} shows the accuracy as resolution increases
at fixed retained dimension. At each $N$, both reduced agents are constructed
from the corresponding exact model using the same design reference. The clock keeps $d=4$ and $h=1/2$ over all 30 resolutions
$N=24,32,\ldots,256$; the walk keeps $d=27$ and equally probable inputs over
the 29 resolutions $N=32,40,\ldots,256$. 

\begin{figure}[tbp]
\centering
\includegraphics[width=.9\textwidth]{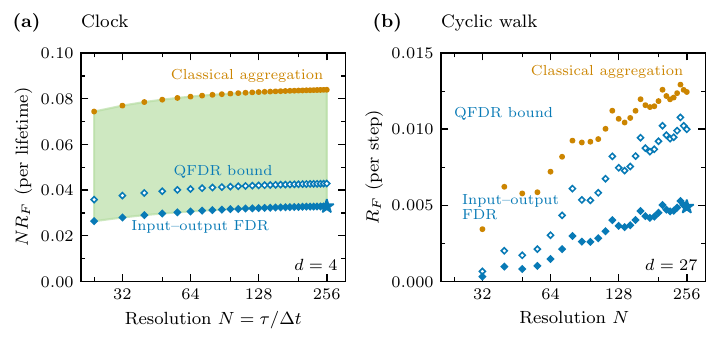}
\caption{\textbf{Accuracy at increasing resolution with a fixed register size.}
\textbf{(a)}~Clock with $d=4$ at reset intensity $h=1/2$, reporting FDR per
physical lifetime. \textbf{(b)}~Walk with $d=27$ and equally probable stimuli,
reporting FDR per step. At each resolution, both agents are constructed from
the corresponding original model. Filled blue diamonds show input--output
FDR, hollow diamonds its QFDR upper bound, and orange circles classical
stationary aggregation. Stars identify $N=256$; the green area is the
clock's difference in evaluated error.}
\label{sm:fig:resolution-sweeps}
\end{figure}

\subsection{Effect of tailoring the design reference}
\label{sm:subsec:reference-tailoring}
The fixed-agent sweep tests how a reduced clock behaves when its input
statistics change. To isolate the effect of the design reference, we also
repeat the reduction at each deployed reset intensity. The original clock
and retained dimension remain $N=256,d=4$. For the quantum agent, tailoring
reselects the leading subspace of the original stationary memory and repeats
polar repair. For classical aggregation, the same partition $g(n)=\lfloor nd/N\rfloor$
is kept and its transition averages are recalculated with the stationary
weights at the new intensity. Neither construction is behaviourally fitted.

Let $\mathcal A_a$ denote either reduced construction designed at intensity
$a$, and let $R_F^{(\mathrm{stat})}(\mathcal A_a;h)$ compare its input--output
histories with the original clock under the deployed intensity $h$.
We subtract the tailored error from that of the fixed design at $h_0=1/2$:
\begin{equation}
\Delta(h)=N\left[
R_F^{(\mathrm{stat})}(\mathcal A_{h_0};h)
-R_F^{(\mathrm{stat})}(\mathcal A_h;h)
\right].
\label{sm:eq:reference-tailoring-difference}
\end{equation}
Both agents are evaluated in their stationary states under the deployed input
process. Positive $\Delta(h)$ favours tailoring, and negative $\Delta(h)$ favours
the fixed design. The difference vanishes at the design intensity $h=h_0$.

Figure~\ref{sm:fig:reference-tailoring} uses the 54 deployment intensities of
Sec.~\ref{sm:subsec:fixed-input-deployment}. Tailoring improves the quantum clock at the sampled reset intensities below
$h_0$, while the fixed design is more accurate at those above $h_0$. Thus the
subspace with greatest stationary weight need not give the smallest
input--output error. For classical aggregation, $\Delta(h)$ becomes positive
again at the highest sampled intensities. Table~\ref{sm:tab:reference-tailoring} gives selected absolute rates alongside
the differences.

\begin{figure}[tbp]
\centering
\includegraphics[width=.76\textwidth]{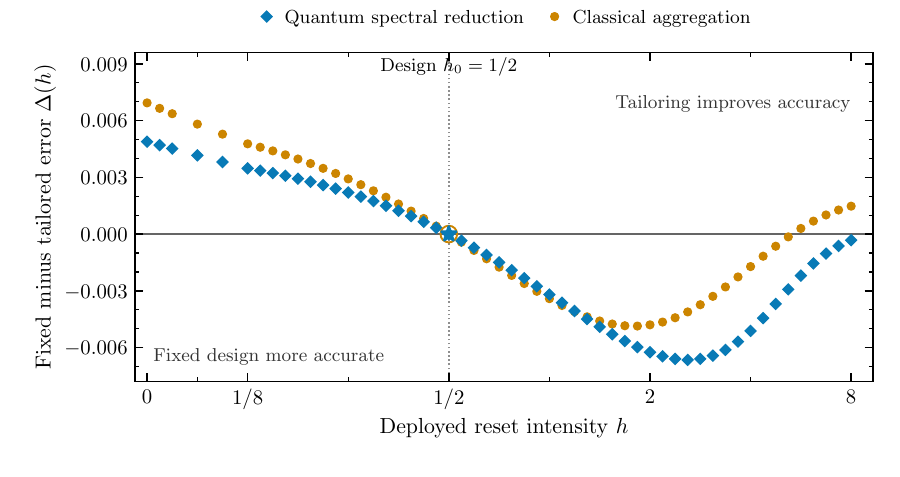}
\caption{\textbf{Effect of tailoring the clock's design reference.}
At $N=256,d=4$, each point shows the fixed-design error minus the error of the agent reconstructed for the deployed intensity. The fixed agent is constructed at $h_0=1/2$. Rates are per physical lifetime. Blue diamonds show spectral
truncation and polar repair; orange circles show stationary aggregation with
the same age partition. The vertical guide marks the common design point.
Values below zero favour keeping the original design.}
\label{sm:fig:reference-tailoring}
\end{figure}

\begin{table}[tbp]
\centering\small
\begin{tabular}{crrrrrr}
\hline
& \multicolumn{3}{c}{Quantum spectral reduction}
& \multicolumn{3}{c}{Classical aggregation}\\
Deployed $h$ & Fixed & Tailored & $\Delta(h)$ & Fixed & Tailored & $\Delta(h)$\\
\hline
$0$ & 0.045881 & 0.040999 & +0.004882 & 0.117077 & 0.110143 & +0.006934 \\
$1/8$ & 0.042248 & 0.038775 & +0.003474 & 0.107754 & 0.102983 & +0.004771 \\
$1/4$ & 0.038884 & 0.036688 & +0.002197 & 0.099157 & 0.096241 & +0.002916 \\
$1/2$ & 0.032893 & 0.032893 & +0.000000 & 0.083934 & 0.083934 & +0.000000 \\
$1$ & 0.023422 & 0.026619 & -0.003197 & 0.060093 & 0.063510 & -0.003416 \\
$2$ & 0.011722 & 0.017973 & -0.006250 & 0.030986 & 0.035785 & -0.004799 \\
$4$ & 0.003035 & 0.008155 & -0.005120 & 0.009475 & 0.011194 & -0.001719 \\
$8$ & 0.000657 & 0.000982 & -0.000326 & 0.003195 & 0.001721 & +0.001474 \\
\hline
\end{tabular}
\caption{Fixed and tailored $N=256,d=4$ clocks at selected deployment
intensities. All entries use the observable input--output FDR per lifetime.
The fixed design is constructed at $h_0=1/2$; the tailored design is
reconstructed at $h$. Classical entries use stationary aggregation, not the
regenerative fit in Table~\ref{sm:tab:clock-off-design}.}
\label{sm:tab:reference-tailoring}
\end{table}

For the cyclic walk, either input preserves the uniform stationary position
distribution. The original stationary quantum memory and the classical
aggregation averages are therefore independent of the iid mixture. With the
same Fourier-mode convention, repeating either reduction returns the same
agent, so the fixed-minus-tailored error is zero throughout the mixture
interval. This was also checked at the 51 mixtures in the fixed-agent sweep.

The calculations use the existing renewal evaluator and the saved frozen
instruments. Reconstructing the clock independently in the production memory
basis reproduces the tailored rates at selected intensities, including zero
external reset. Normalisation, regeneration, stationary states, the signed
differences and agreement with the earlier seven-point redesign table were
checked separately. The full grid, reconstructed instruments and validation
scripts accompany the numerical data.

\subsection{A clock with a responding two-mode environment}
\label{sm:subsec:feedback-numerics}
The environment has modes $c=0,1$ and supplies an external reset with probability
$r_c=1-e^{-h_c/N}$, where $(h_0,h_1)=(1/4,3/2)$. A natural tick sets $c'=1$;
an external reset sets $c'=0$. After an ordinary no-tick step, mode 0 stays fixed
and mode 1 changes to mode 0 with probability $a=1-e^{-1/N}$.
With column propagation, the environmental matrices for no event, natural tick,
and external reset are
\begin{equation}
H_0=\begin{pmatrix}1&a\\0&1-a\end{pmatrix}
       \operatorname{diag}(1-r_0,1-r_1),\quad
H_t=\begin{pmatrix}0&0\\1-r_0&1-r_1\end{pmatrix},\quad
H_r=\begin{pmatrix}r_0&r_1\\0&0\end{pmatrix}.
\end{equation}
The agent map is applied separately for each visible event and coupled to the
corresponding $H$. These rules are fixed before choosing the retained dimension.
We use $N=24,64$ and $d=1,2,4,8$. The original stationary state is obtained from
the joint classical age--mode chain and encoded through the exact clock memory
vectors. It is correlated across $C$ and $M$. We diagonalise only its agent
marginal, repair the original clock stimulus by stimulus, and solve the reduced
joint stationary state under the unchanged environment rules.

The feedback history can be split at recorded events. Each event fixes both
the environment mode and the agent's next memory: external resets give mode 0;
natural ticks give mode 1. This renewal structure lets us sum the history
probabilities directly, without sampling trajectories.
For a post-event mode $c$, compute the probabilities of $\ell-1$ ordinary no-tick
rounds followed by an event into mode $c'$ by propagating the joint maps. Denote
the root product of the target and reduced probabilities by $b_{c'c,\ell}$,
and the root product of the no-event suffix probabilities by $h_c(L)$.
The two conditional fidelities satisfy
\begin{equation}
F_c(0)=1,\qquad
F_c(L)=h_c(L)+\sum_{\ell=1}^L\sum_{c'=0}^1
b_{c'c,\ell}F_{c'}(L-\ell).
\label{sm:eq:feedback-renewal}
\end{equation}
A stationary start is handled by its own first-event probabilities before applying
this recurrence. Their square-root products $b^{\mathrm{st}}_{c',\ell}$ and
suffix $h^{\mathrm{st}}(L)$ give
\begin{equation}
F^{\mathrm{st}}(L)=h^{\mathrm{st}}(L)+
\sum_{\ell=1}^L\sum_{c'=0}^1 b^{\mathrm{st}}_{c',\ell}F_{c'}(L-\ell).
\end{equation}
Hidden environmental transitions are summed before square roots. The visible
$(x,y)$ record determines the first event uniquely, which gives the recurrence.
It is this observed regeneration that makes the present evaluation simple;
other feedback environments remain within the routing construction but need
an appropriate calculation of their history probabilities.

We evaluate stationary histories and, separately, histories starting from the
local encoding of the original joint state. The latter preserves the environment
marginal but need not be stationary for the reduced channel.
Table~\ref{sm:tab:feedback} records the stationary comparison. The fidelity
therefore includes changes both in the clock's outputs and in the later inputs
elicited from the environment. The shift in reset frequency quantifies one
part of that response.

\begin{table}[tbp]
\centering\small
\begin{tabular}{rrrrrr}
\hline
$N$ & $d$ & $F_N$ & $F_{2N}$ & Original reset frequency & Reduced reset frequency\\
\hline
24 & 1 & 0.908819 & 0.749280 & 0.033538 & 0.030305\\
24 & 2 & 0.958901 & 0.882653 & 0.033538 & 0.032396\\
24 & 4 & 0.984081 & 0.950433 & 0.033538 & 0.033135\\
24 & 8 & 0.994267 & 0.980780 & 0.033538 & 0.033400\\
64 & 1 & 0.901460 & 0.733794 & 0.012723 & 0.011333\\
64 & 2 & 0.952216 & 0.869430 & 0.012723 & 0.012184\\
64 & 4 & 0.979360 & 0.940649 & 0.012723 & 0.012514\\
64 & 8 & 0.990893 & 0.972966 & 0.012723 & 0.012628\\
\hline
\end{tabular}

\caption{Clock reduction in the same responding environment. Each joint
process starts in its own stationary state. $F_N$ and $F_{2N}$ are root
fidelities of the input--output records over one and two physical lifetimes.
Reset frequencies are per step; their change reflects the environment's
response to the smaller clock.}
\label{sm:tab:feedback}
\end{table}

\subsection{Evaluation, validation and reproducibility}
The code uses the audited original agent constructors. Classical arrays use
\texttt{T[input,output,next,current]}; adapters transpose the original
\texttt{TransducerSpec} convention explicitly. Input/output labels are aligned
before hidden-label padding. Completeness is verified separately for each input.
For the plotted agents, each projected Gram operator $G_x$ is invertible on the retained memory, so no kernel completion is needed.
A singular completion requires an explicit visible-output assignment and, for a
specialised renewal calculation, verification of regeneration.

For the walk, the uniform and Gaussian laws use the original source-bin averaging,
with 256 source points per bin and 4097 shift points; the Gaussian is symmetrically
truncated at six standard deviations. Prior fine-grid and analytic-bin checks
are retained in the accompanying provenance material. The dense walk sweep evaluates the same repaired model in Fourier coordinates.
Let $F$ be the discrete Fourier matrix, $H_d=F\operatorname{diag}(\chi_d)F^\dagger$
the selected-mode projector in causal-state coordinates, $h=H_de_0$, and
$w=N^{-1}\sum_{k\in I_d}g_k$, where $g_k$ are the eigenvalues of the full Gram
matrix. For $a_{x,y}(j)=h_j\sqrt{p_x(y-j)}$ and the unused-input Gram matrix
$G_{1-x}$, the reduced transition law from position zero is
\begin{equation}
\widetilde p_x(y|0)=w^{-1}a_{x,y}^\dagger G_{1-x}a_{x,y}.
\end{equation}
Circulant multiplication is evaluated by FFT. Translation covariance gives the
other columns, and the overlap transfer gives the observable rate. The
state-independent-retention identity yields the same specified-dilation QFDR
$-\tfrac14\log_2w$. Selected dense points are checked independently using the
original grouped-effect construction and the rectangular mixed transfer.
The latter factors as $\mathcal E=AB$, with $BA$ an $N\times N$ matrix having
the same nonzero spectrum; that independently evaluated spectrum agrees at
the checked points. These changes avoid materialising large grouped effects
for every new point; they do not change the instrument, reference, or metric.

The underlying benchmark and feedback tables passed 369 and 220 checks, respectively, at the stated
tolerances. They include normalisation, stationarity, independent clock-rate
calculations, finite- and asymptotic-fidelity data processing, full-space walk
retention, and conservation of timing probabilities. For $N=4$, the feedback
recurrence is checked against exhaustive visible histories through five rounds;
the no-compression endpoint gives unit fidelity. The original and audited releases
were also compared before these revisions. These are floating-point consistency
checks, not interval-certified errors or global optimiser certificates.

The accompanying prior set of 347 saved-model checks verifies independently reconstructed rates,
selected parameter gradients, feedback stationarity and generated table consistency.
On the previously evaluated clock grid, the generic reset bound clips to one at $L=N$;
adding the polar leakage term cannot improve it. The sharp numerical claims
in the figures use the direct observable and mixed-transfer calculations, not those
generic bounds. Their preparation conventions also remain distinct.
The dense tables undergo 151 further consistency checks, including the integer
grid, data processing, stationary residuals, and the first passing rank at each
resolution. Forty-five independent walk checks compare the FFT expression with
saved data and explicit production agents. Selected extrema of the dense walk
curves are recomputed at 512 source points and 8193 shift points: the first
passing observable and certified ranks are unchanged at all nine checked
resolutions. These are numerical checks, not interval-enclosed crossings.
The old reference-probability sweep rebuilt the stationary state, projector and
repair at every reference: it is a redesign study, not an off-design test. It is
not used as robustness evidence in the revised main figures. The deployment sweeps of Sec.~\ref{sm:subsec:fixed-input-deployment} freeze the
instruments while changing only the input probabilities and the stationary
states used in their evaluation.

\subsection{Scope of the classical comparison and exploratory refinement}
\label{sm:subsec:classical-scope}
The main figure compares stationary aggregation with the direct quantum
construction for both benchmarks. The additional regenerative classical search
gives a stronger clock comparison; it does not bound every classical HMM. A separate small-system pilot at $N=16,d=4$ gave observable FDR about $0.098544$ for
contiguous classical merging and $0.063714$ for the direct quantum construction.
Behavioural fitting changes this ordering: a predictive classical fit reached
about $0.035980$. In a matched output-reset architecture, allowing stochastic
classical post-output states gave $0.026459$, whilst quantum states and effects gave
$0.017546$. These exploratory fits were not used for the principal curves and do
not establish global optima. Both fitted families sacrificed the target's uniform
position marginal, unlike the direct spectral construction. Their saved models,
methods and separate parameter conventions accompany the numerical package.

The fitted strategies can themselves be encoded with quantum memory and
compressed. Joint optimisation of behavioural error and physical representation
is therefore complementary to the direct construction studied here.

\section{Registers for a physical implementation}
\label{sm:sec:implementation-resources}

Each round acts on the persistent memory and the supplied stimulus, using fresh action and hidden environmental registers. We count their logical qubit
widths, excluding external tester memory and any particular gate
implementation. The action is subsequently read and the environment
discarded.

\subsection{Classical stimuli and measured actions}

For a reduced outcome map $\mathcal J_{y|x}$ on $d$ memory dimensions, its
Choi operator has dimension $d^2$ and rank $r_{xy}\leq d^2$. An outcome
Kraus representation with $r_{xy}$ operators therefore exists. Choose a
common hidden register of dimension $r=\max_{x,y}r_{xy}$, padding each
outcome's list with zeros. For each classical stimulus,
\begin{equation}
V_x|\psi\rangle
=\sum_y\sum_{\eta=1}^{r}
K_{y,\eta}^{(x)}|\psi\rangle\,|y\rangle_{\mathsf Y}|\eta\rangle_{\mathsf E}
\label{sm:eq:outcome-rank-dilation}
\end{equation}
is an isometry by the instrument's completeness. Tracing $\mathsf E$ and
measuring $\mathsf Y$ implements the same outcome maps. Thus $r$ bounds the
hidden register; the output register still has dimension $|\mathcal Y|$.

If an outcome prepares a normalised pure state $|s_y\rangle$, it has the
form $\mathcal J_{y|x}(\rho)=\Tr(H_{y|x}\rho)|s_y\rangle\langle s_y|$.
Its Choi rank is $\rank H_{y|x}\leq d$. A single-Kraus outcome has rank one.
All outcomes of the benchmark reductions are of these forms, so $r\leq d$
suffices for their measured instruments. A mixed-state replacement does not
in general have the same rank bound.

At $N=256$, the exact clock has eight memory qubits and a reset channel that prepares a pure state and has Kraus rank $N$. With eight hidden qubits and one each for input and output,
it admits a round of width 18; at $d=4$ the corresponding width is 6.
The walk retains an eight-qubit register naming the output position.
Table~\ref{sm:tab:implementation-width} gives the counts.

\begin{table}[tbp]
\centering\small
\begin{tabular}{llrrrrr}
\hline
Agent & Implementation & Memory & Input & Output & Hidden & Total\\
\hline
Clock, $N=256$ & Stated exact dilation & 8 & 1 & 1 & 8 & 18\\
Clock, $d=4$ & Classical readout & 2 & 1 & 1 & 2 & 6\\
Clock, $d=4$ & Coherent, $u=j_0$ & 2 & 1 & 1 & 3 & 7\\
Walk, $N=256$ & As constructed & 8 & 1 & 8 & 8 & 25\\
Walk, $d=27$ & Classical readout & 5 & 1 & 8 & $\leq5$ & $\leq19$\\
\hline
\end{tabular}
\caption{Achievable qubit widths per round for the stated register
implementations at $N=256$. Classical readout recompiles each outcome while
preserving the measured input--output instrument. The coherent clock row
preserves its shared environmental implementation and is not the orthogonal
flag control. Input and output sizes are included; no minimum gate count,
circuit depth, control precision or error-correction overhead is inferred.}
\label{sm:tab:implementation-width}
\end{table}

\subsection{Preserving the coherent implementation}

Changing Kraus bases separately for different stimuli or outcomes can change
their relative coherences. Preserving the specified coherent channel instead
requires one common transformation on the full hidden environmental span.
Its dimension is the Kraus rank of that channel, not the largest rank of a
separately measured outcome.

For the clock implementation used in
Sec.~\ref{sm:subsec:quantum-tester-numerics}, choose an orthonormal memory
basis in which external reset writes the incoming coordinates into
$\mathsf E$. If $J:\mathbb C^d\to\mathbb C^N$ is the retained embedding,
its reset records span $\operatorname{ran}J$. Ordinary evolution writes the
fixed vector $u$ in the same environment. After repair, the Gram matrix of
the full coherent-channel Kraus list, indexed by that environment, is
\begin{equation}
H_E=JJ^\dagger+d|u\rangle\langle u|.
\label{sm:eq:coherent-clock-environment}
\end{equation}
The reset contribution is $JJ^\dagger$; ordinary-input completeness gives
$\sum_y\|\widetilde K^{(0)}_y\|_F^2=d$. Distinct stimulus and action
registers remove cross terms in this Kraus Gram matrix.
It has rank $d+1$ when $u\notin\operatorname{ran}J$, and rank $d$ otherwise.
For the $N=256,d=4,h=1/2$ projector and the choice
$u=|j_0\rangle$, the first case holds. Five environmental dimensions
therefore use three qubits, and the coherent round in the table has width 7.
The exactly orthogonal flag control discussed with the tester is a different
implementation and is not the one counted in this row.

\ifdefined\supplementincluded
  \def\finishsupplement{ }
\else
  \def\finishsupplement{\end{document}}
\fi
\finishsupplement

\fi

\end{document}